\documentclass[11pt]{article}

\usepackage[margin=1in]{geometry}
\usepackage{amsmath,amssymb,amsthm}
\usepackage{bbm}
\usepackage{newtxtext}
\usepackage[subscriptcorrection]{newtxmath}
\usepackage{booktabs}
\usepackage{graphicx}
\usepackage[round,authoryear]{natbib}
\usepackage{bibunits}
\usepackage{microtype}
\usepackage{hyperref}
\hypersetup{hidelinks}
\graphicspath{{Figures/}}
\defaultbibliographystyle{plainnat}
\defaultbibliography{copula}

\def\T{{ \mathrm{\scriptscriptstyle T} }}

\newcommand{\ind}{\perp\!\!\!\perp}

\newcommand{\ATE}{\mbox{\tiny{ATE}}}

\newtheorem{theorem}{Theorem}

\newtheorem{proposition}{Proposition}

\newtheorem{assumption}{Assumption}
\newtheorem{remark}{Remark}

\newtheorem{algo}{Algorithm}

\newenvironment{keywords}{\par\vspace{0.5em}\noindent\textbf{Keywords: }}{\par\vspace{1em}}

\title{Causal inference with ordinal outcomes: copula-based identification, estimation and sensitivity analysis}
\author{Peiyu He\\
	Department of Biostatistics and Bioinformatics, Duke University\\
	2424 Erwin Road, Durham, North Carolina 27710, U.S.A.\\
	\texttt{ph172@duke.edu}
	\and
	Fan Li\\
	Department of Statistics, Duke University\\
	214 Old Chemistry, Box 90251, Durham, North Carolina 27708, U.S.A.\\
	\texttt{fl35@duke.edu}}
\date{}

\begin{document}
	
	\maketitle
	\begin{bibunit}
		
		\begin{abstract}
			In causal inference with ordinal outcomes, several interpretable estimands are functions of the probability that the potential outcome under one treatment is larger than that under another treatment for the same unit. This probability depends on the joint distribution of both potential outcomes and is generally not identifiable. Existing work has focused on sharp bounds of this probability based on partial identification, but bounds are often too wide to be informative. We propose a copula-based method that links the identifiable marginal distributions of the potential outcomes via a parametric copula, treating the copula association parameter as a sensitivity parameter. 
			With a fixed copula parameter, the estimands become identified functionals of the observed data. 
			Working under unconfoundedness, we derive the efficient influence function in the nonparametric model and construct one-step estimators that accommodate flexible nuisance estimation. 
			The resulting procedure is rate-doubly-robust and attains the semiparametric efficiency bound under standard conditions.
			Varying the copula parameter yields a sensitivity curve with point-wise confidence bands that typically lie within the sharp bounds, providing an interpretable bridge between partial identification and point estimation.
			We further provide a comprehensive sensitivity analysis with respect to both the copula specification and the unconfoundedness assumption. 
			We develop an associated R package \texttt{ordinalCI}.
		\end{abstract}
		
		
		\section{Introduction}
		Many biomedical and social science studies aim to evaluate treatment effects on outcomes that are measured on ordered scales, e.g., disease severity stages and educational attainment levels. Unlike continuous outcomes, estimands defined in average differences may lack a coherent individual-level interpretation. Various alternative causal estimands have been proposed for ordinal outcomes, including the conditional median effect \citep{volfovsky2015causal}, relative treatment effect \citep{lu2020sharp}, and the probability of necessity for causes of effects \citep{zhang2025identifying}. In this paper, we focus on the arguably most common causal estimand of ordinal outcomes: the probability that a treatment is strictly beneficial (or harmful) to an individual, defined as $\psi=\text{pr}\{Y(1)> Y(0)\}$ and its variants, where $Y(1)$ and $Y(0)$ denote the potential outcomes under treatment and control, respectively \citep{li2008bayesian,huang2017inequality,lu2018treatment, lu2020sharp}. This estimand has a natural interpretation as the proportion of individuals whose outcome would improve (or worsen) under treatment. 
		
		Because the potential outcomes $Y(0)$ and $Y(1)$ are never jointly observed for the same individual, $\psi$ is not identifiable without additional assumptions. A growing literature has addressed this challenge through partial identification, deriving sharp bounds under minimal assumptions using only the marginal distributions of $Y(1)$ and $Y(0)$ \citep{huang2017inequality,lu2018treatment,lu2020sharp}.  However, bounds are often too wide to be informative and difficult to interpret; therefore, they are rarely adopted in biomedical research. Another strand of literature imposes Bayesian models on the joint distribution of $Y(1)$ and $Y(0)$ and draws posterior inference of the estimands \citep{volfovsky2015causal, chiba2018bayesian}, but the model assumptions are generally untestable and can be sensitive in the presence of poor overlap. 
		Moreover, the posterior inference usually relies on computationally intensive Markov Chain Monte Carlo algorithms, which may be undesirable in large population-based observational studies.

		In this paper, we propose a copula-based framework that addresses these limitations. 
		We employ a parametric copula to link the joint distribution of the 
		potential outcomes with their identifiable marginals, adding the 
		minimal structure required for identification, with the association 
		parameter $\rho$, or equivalently the copula-free Kendall's $\tau$, 
		serving as a single interpretable sensitivity parameter.
		Unlike nonparametric sharp bounds, 
		each fixed $\rho$ delivers a point estimate with valid uncertainty 
		quantification; varying $\rho$ produces a sensitivity curve contained 
		within the sharp bounds of \citet{lu2018treatment, lu2020sharp} and 
		passing through the true value under correct specification, with the 
		envelope across copula families recovering the bounds. 
		The framework 
		thus offers a transparent and informative bridge between partial 
		identification and point estimation.
		For estimation and inference, we adopt the semiparametric efficiency 
		theory \citep{tsiatis2006semiparametric} under unconfoundedness. 
		At each fixed $\rho$, we derive the efficient influence function for 
		$\psi$ and construct rate-doubly-robust one-step estimators that accommodate 
		flexible machine-learning nuisance estimation \citep{chernozhukov2018double}. 
		The resulting estimators attain the semiparametric efficiency bound 
		\citep{newey1990semiparametric} under standard conditions and admit 
		closed-form, computationally efficient inference.
		We further provide a comprehensive sensitivity analysis for copula specification and a Rosenbaum-type sensitivity analysis \citep{rosenbaum2002observational,yadlowsky2022bounds} for unconfoundedness assumption. 
		We develop an \textsf{R} package, \texttt{ordinalCI}, implementing the 
		methodology, available at \url{https://github.com/xiaoshiziguoer/ordinalCI}.
		
		Copula models have been used in causal inference. A prominent example is principal stratification, where a copula was used to model the distribution of the latent principal strata, defined as the joint potential values of a post-treatment variable \citep{bartolucci2011modeling, lu2025principal}. Copula-based sensitivity analysis has also been considered for settings with multiple treatments and unobserved confounding \citep{zheng2025copula}. However, none of the previous applications focus on ordinal outcomes.

		\section{Estimands and copula-based identification}

		Consider a sample of $n$ units indexed by $i = 1, \dots, n$. For each unit, let 
		$A_i \in \{0, 1\}$ denote a binary treatment indicator and $X_i$ a 
		vector of pre-treatment covariates. The outcome 
		$Y_i \in \{0, 1, \ldots, L-1\}$ is ordinal with $L$ ordered levels, 
		larger values representing better responses (or worse, depending on 
		context). 
		Under the Stable Unit Treatment Value Assumption (SUTVA), 
		each unit has two potential outcomes, $Y_i(1)$ and $Y_i(0)$, under 
		treatment and control, of which only 
		$Y_i = A_i Y_i(1) + (1 - A_i) Y_i(0)$ is observed. For simplicity, we shall suppress the subscript $i$ henceforth unless otherwise specified.

		The goal is to evaluate the causal effect of the treatment $A$ on the outcome $Y$. The standard causal estimand is the average treatment effect $\tau^{\ATE}=E\{Y(1)-Y(0)\}$. However, averages of an ordinal outcome are often challenging to interpret. Common alternative causal estimands are the probability that the treatment is strictly beneficial, or beneficial for the outcome \citep{huang2017inequality,lu2018treatment}, and the relative treatment effect \citep{lu2020sharp}, defined as:
		\begin{equation*}
			\psi = \text{pr}\{Y(1)> Y(0)\},  \quad \phi = \text{pr}\{Y(1)\geq Y(0)\}, \quad \xi = \text{pr}\{Y(1)> Y(0)\} - \text{pr}\{Y(1)< Y(0)\}, \label{eq:estimands}
		\end{equation*}
		respectively. Unlike conventional additive estimands, $\psi$, $\phi$ and $\xi$ involve the joint distribution of $\{Y(1), Y(0)\}$. 
		Let 
		$\Pi = \{\pi_{kj}\}_{0 \leq k, j \leq L-1}$ with 
		$\pi_{kj} = \text{pr}\{Y(1) = k, Y(0) = j\}$ denote the joint 
		cell probability; then the estimands can be expressed as functions of $\Pi$ as
		$\psi = \sum_{k=0}^{L-1} \sum_{j=0}^{k-1} \pi_{kj}$, 
		$\phi = \sum_{k=0}^{L-1} \sum_{j=0}^{k} \pi_{kj}$, and 
		$\xi = \phi + \psi - 1$.
		Therefore, once $\Pi$ is identified, $(\psi, \phi, \xi)$ are identified.
		However, $Y(1)$ and $Y(0)$ are 
		never jointly observed by the fundamental problem of causal inference, and only the marginal distributions 
		$\text{pr}\{Y(a) \leq k\}$, $a = 0, 1$, are identifiable under the 
		following standard assumptions.
		\begin{assumption}\label{asmp: causal}
			(i) (Unconfoundedness) $\{Y(0),Y(1)\}\ind A\mid X$; (ii) (Overlap) $0<\text{pr}(A=1\mid X) <1$.
		\end{assumption}
		
		Under Assumption~\ref{asmp: causal}, $\text{pr}\{Y(a) \leq k\}$ is 
		identifiable through either outcome regression or inverse probability 
		weighting,
		\[\text{pr}\{Y(a) \leq k\} = E\left\{\text{pr}\left(Y\leq k\mid A=a, X\right)\right\} = 
		E\left\{\frac{\mathbbm{1}\left(A=a\right)\mathbbm{1}\left(Y\leq k\right)}{\text{pr}\left(A =a\mid X\right)}\right\}, \quad a = 0,1.\]
		The marginal distributions are sufficient to identify additive estimands such as the average treatment effect, but are not sufficient to identify $(\psi, \phi,\xi)$. The literature has focused on partial identification, namely, deriving nonparametric sharp bounds of $(\psi, \phi,\xi)$ using the margins with minimal assumptions \citep{huang2017inequality,lu2018treatment,lu2020sharp,zhang2025identifying}.  However, in addition to being wide, bounds are usually challenging to interpret for biomedical researchers, who typically prefer point estimates accompanied by uncertainty quantification.

		A copula provides a natural way to model the unidentifiable joint 
		distribution from the identifiable marginals. 
		Let $F_a(k\mid X) = \text{pr}\{Y(a)\leq k\mid X\}$ be the potential outcome margins that can be identified as $\text{pr}\{Y\leq k\mid X, A=a\}$ under Assumption \ref{asmp: causal}.
		\begin{assumption}\label{asmp: copula}
			There exists a known copula $C_\rho$ such that
			\begin{equation}\label{eq: copula}
				\text{pr}\{Y(1)\le k,\;Y(0)\le j\mid X\}
				=
				C_\rho\!\left\{F_1(k\mid X),\,F_0(j\mid X)\right\}, \quad \forall k, j = 0,\ldots,L-1,
			\end{equation}
		\end{assumption}
		Here $C_\rho$ governs the conditional dependence between $Y(1)$ and 
		$Y(0)$ given $X$, with the parameter $\rho$ indexing the strength of 
		dependence and assumed constant in $X$. Neither the copula family nor 
		$\rho$ is identified from the observed data; both are varied and 
		assessed in sensitivity analyses. 
		By Sklar's theorem \citep{sklar1959fonctions}, a copula representation always exists and is unique when the margins are continuous. 
		For discrete margins, uniqueness fails: $\rho$ is not identified even if the joint and 
		marginal distributions were both available, introducing ambiguity in 
		its specification and interpretation. This non-identifiability, 
		however, does not undermine the approach, because for any fixed copula 
		specification the estimands are uniquely determined by the identifiable 
		marginals through~\eqref{eq: copula}. The following result further 
		clarifies the role of $\rho$ by connecting Assumption~\ref{asmp: copula} 
		to the latent threshold model commonly used for ordinal 
		outcomes.
		
		\begin{proposition}[Latent threshold model interpretation]\label{prop:latent-ordinal-copula}
			Suppose that, for $a = 0,1$,
			\begin{align} \label{eq:latent-model}
				Y^*(a)&=\eta_a(X)+\varepsilon_a,\quad
				Y(a)=\ell \iff \lambda_{\ell-1}<Y^*(a)\le \lambda_\ell,\quad \ell=0,\ldots,L-1,
			\end{align}
			where $
			-\infty=\lambda_{-1}<\lambda_0<\cdots<\lambda_{L-2}<\lambda_{L-1}=+\infty$
			are thresholds shared by two treatment arms, and $\eta_a(\cdot)$ is an arbitrary measurable function of $X$. If the conditional joint distribution of the latent
			residuals satisfies
			\[
			\mathrm{pr}(\varepsilon_1\le e_1,\;\varepsilon_0\le e_0\mid X)
			=
			C_\rho\!\left\{F_{\varepsilon_1\mid X}(e_1),\,F_{\varepsilon_0\mid X}(e_0)\right\},
			\]
			where $F_{\varepsilon_a\mid X}(e)=\mathrm{pr}(\varepsilon_a\le e\mid X)$, $a = 0,1$, then Assumption~\ref{asmp: copula} holds with the same copula $C_{\rho}$. 
		\end{proposition}
		The proof is given in the Supplementary Material Section \ref{sec: supp-proofs-others}.
		Proposition~\ref{prop:latent-ordinal-copula} provides a structural interpretation for
		Assumption~\ref{asmp: copula}: under the common latent threshold model with a
		flexible regression structure, a copula $C_\rho$ on the regression residuals of the
		latent continuous potential outcomes implies the same copula specification on the ordinal potential outcomes. 
		The converse does not hold in general, because discretization is a
		many-to-one map. 
		Nevertheless, the implication provides a natural and interpretable
		modelling rationale to arrive at Assumption~\ref{asmp: copula}, with $\rho$ encoding
		the conditional dependence of $(\varepsilon_1,\varepsilon_0)$
		given~$X$. 
		The component $\eta_a(X)$ is arbitrary, and the
		model~\eqref{eq:latent-model} naturally accommodates common ordinal regression models
		such as the proportional odds model by setting $\eta_a(X) = \beta^{T}X + \delta a$.

		We henceforth write the fixed copula function as $C(\cdot, \cdot)$, suppressing the subscript $\rho$ for simplicity.
		For $k, j  = 0, \ldots, L-1$, the joint conditional cell probabilities $\pi_{kj}(x)= \text{pr}\{Y(1) = k,Y(0)=j\mid X=x\}$ can be identified by rectangle increment on the ordinal grid using the copula: 
		$\pi_{kj}(x) = C\{F_1(k\mid x), F_0(j\mid x)\} - C\{F_1(k-1\mid x), F_0(j\mid x)\} - C\{F_1(k\mid x), F_0(j-1\mid x)\} + C\{F_1(k-1\mid x), F_0(j-1\mid x)\}.$
		Then $\psi$ is identified by 
		\begin{equation}\label{eq: idenpsi}
			\psi =  E\left\{m_{\psi}(X)\right\},\quad m_{\psi}(x) = \text{pr}\{Y(1) > Y(0) \mid X = x\} = \sum_{k=0}^{L-1} \sum_{j=0}^{k-1} \pi_{kj}(x),
		\end{equation}
		$\phi$ is identified analogously with $m_{\phi}(x) = \sum_{k=0}^{L-1} \sum_{j=0}^{k} \pi_{kj}(x)$, and $\xi$ is identified by $\phi+\psi-1$.
		
		\begin{remark}\label{remark: constant-rho-unc}
			Assumption~\ref{asmp: copula} implicitly assumes $\rho$ to be constant in $X$.
			This restriction parallels familiar homogeneous structures: under Gaussian errors, it
			reduces to a fixed correlation coefficient in the error covariance matrix.  Our framework
			extends in principle to covariate-dependent $\rho = \rho(X)$, but specifying its
			functional form is difficult and sensitivity analysis becomes unwieldy.
			Simulation evidence 
			in Section~\ref{sec: supp-addsimu} of the Supplementary Material shows 
			that the constant-$\rho$ approximation incurs little bias when the truth 
			is heterogeneous, so we maintain a constant $\rho$ throughout and 
			discuss its sensitivity analysis in Section~\ref{sec: sensitivity}. 
			An alternative strategy that avoids the constant-$\rho$ restriction is to 
			impose the copula directly on the unconditional marginals,
			$\mathrm{pr}\{Y(1)\le k,\;Y(0)\le j\} = C_\rho\!\left\{F_1(k),\,F_0(j)\right\},$
			where $F_a(k) = \mathrm{pr}\{Y(a)\leq k\} =
			E\!\left(\mathrm{pr}\{Y\leq k\mid X, A=a\}\right)$ under
			Assumption~\ref{asmp: causal}.
			Here, $\rho$ captures the unconditional dependence
			between $Y(1)$ and $Y(0)$, yielding a less
			structured but more parsimonious model.
			Full results on estimation, inference, and
			sensitivity analysis under this unconditional copula model are deferred to
			Section~\ref{sec: supp-unc} of the Supplementary Material.
			
		\end{remark}
		
		\begin{remark}
			Another popular estimand for ordinal outcomes is the win ratio \citep{pocock2012win,mao2018causal} -- a function of the probability $r=\text{pr}\{Y_i(1)>Y_j(0)\}$, where $i$ and $j$ are two independent units sampled randomly. The probability $r$ reduces to $\psi$ if $Y_i(1)\perp Y_i(0)\mid X_i$, which is a special case of the copula model in Assumption \ref{asmp: copula}, for example, Gaussian copula with $\rho = 0$.
		\end{remark}

		\section{Semiparametric estimation and inference}\label{sec: est}
		We work with the unrestricted nonparametric model for the observed data $ O = (Y,A,X)$,
		\[
		\mathcal{M}_{\mathrm{np}}
		=
		\left\{
		\text{pr}(Y,A,X): Y \in \{0,\ldots, L-1\}, 0<\text{pr}(A=1\mid X)<1
		\right\},
		\]
		which places no restrictions on the observed-data law beyond positivity of treatment assignment. 
		The nuisance components are the marginal distribution of $X$, the propensity score $e(x)=\text{pr}(A=1\mid X=x)$, and the conditional distribution of $Y$ given $(A,X)$, encoded by the collection $\{F_a(k\mid x):k=0,\ldots,L-1, a=0,1\}$.
		For brevity, we present results only for $\psi = E\{m_\psi(X)\}$, with $m_\psi(X)$ the copula-induced functional defined in \eqref{eq: idenpsi}; analogous 
		results for $\phi$ and $\xi$ are deferred to Section~\ref{sec: supp-otherestimands} of the Supplementary Material.
		
		Assume that the copula function $C(u,v)$ is differentiable in each argument in $(0,1)^2$, with the partial derivatives denoted by
		$\dot{C}_1(u,v)=\partial C(u,v)/\partial u$ and $\dot{C}_0(u,v)=\partial C(u,v)/\partial v$.
		Define
		\begin{equation}\label{eq: deltas}
			\begin{aligned}
				\Delta_{1k}^{\psi}(x)
				&= \frac{\partial m_{\psi}(x)}{\partial F_1(k\mid x)} = \dot{C}_1\left\{F_1(k\mid x),F_{0}(k-1\mid x)\right\}-\dot{C}_1\left\{(F_1(k\mid x),F_0(k\mid x)\right\},\\
				\Delta_{0k}^{\psi}(x)
				&= \frac{\partial m_{\psi}(x)}{\partial F_0(k\mid x)} = \dot{C}_0\left\{F_1(k+1\mid x),F_{0}(k\mid x)\right\}-\dot{C}_0\left\{F_1(k\mid x),F_0(k\mid x)\right\},
			\end{aligned}
		\end{equation}
		which collect the partial derivatives of $m_\psi$ with respect to the 
		conditional margins.
		The following proposition characterizes the efficient influence function of $\psi$. 
		
		\begin{proposition}[Efficient influence function]\label{prop: eif}
			Under Assumptions~\ref{asmp: causal}--\ref{asmp: copula}, the efficient influence function of $\psi$
			in the nonparametric model $\mathcal{M}_{\mathrm{np}}$ is
			\begin{equation}\label{eq:eif}
				IF_\psi(O)
				=
				\{m_\psi(X)-\psi\}
				+
				\sum_{a=0}^{1}\frac{\mathbbm{1}(A=a)}{e(X)^{a}\left\{1-e(X)\right\}^{1-a}} \sum_{k=0}^{L-2} \Delta_{ak}^{\psi}(X)\left\{\mathbbm{1}(Y\le k)-F_a(k\mid X)\right\}.
			\end{equation}
			The semiparametric efficiency bound for $\psi$ in $\mathcal{M}_{\rm np}$ equals $E\{IF_\psi(O)^2\}$.
		\end{proposition}
		The proof is given in Section \ref{sec: supp-proofmain} of the Supplementary Material.
		
		Given flexible estimators $\hat e(x)$ and $\hat F_a(k \mid x)$, let 
		$\hat m_\psi(x)$, $\hat\Delta_{1k}^{\psi}(x)$, and 
		$\hat\Delta_{0k}^{\psi}(x)$ denote the corresponding plug-in estimators 
		from \eqref{eq: idenpsi} and \eqref{eq: deltas}. The semiparametric 
		one-step estimator is
		\begin{equation}\label{eq:onestep}
			\hat\psi
			=
			\hat E\left(
			\hat m_\psi(X)
			+
			\sum_{a=0}^{1}\frac{\mathbbm{1}(A=a)}{\hat{e}(X)^{a}\left\{1-\hat{e}(X)\right\}^{1-a}} \sum_{k=0}^{L-2} \hat{\Delta}^{\psi}_{ak}(X)\left\{\mathbbm{1}(Y\le k)-\hat{F}_a(k\mid X)\right\}\right),
		\end{equation}
		where $\hat{E}$ denotes the empirical mean operator.
		
		For inference, we estimate the 
		asymptotic variance by $\hat E\{\hat{IF}_\psi^2(O)\}$, where 
		$\hat{IF}_\psi$ is the plug-in version of \eqref{eq:eif}, and construct 
		confidence intervals accordingly. 
		Denote $\|g\|_2=(E\{g(X)^2\})^{1/2}$ for a function $g$. Define
		$
		\|\hat F_a-F_a\|_2^2
		=
		\sum_{k=0}^{L-2} \|\hat F_a(k\mid X)-F_a(k\mid X)\|_2^2$ for $a = 0,1$.
		The following theorem summarizes the asymptotic properties of the one-step estimator.
		
		\begin{theorem}[Rate-doubly-robust expansion]\label{thm:rate}
			Under Assumptions \ref{asmp: causal}--\ref{asmp: copula}, suppose (i) $c\le e(X), \hat e(X),F_a(k\mid X), \hat F_a(k\mid X)\le 1-c$ for some constant $c>0$ and all $a,k$; (ii) the copula $C$ is twice continuously differentiable on $[c, 1-c]^2$ with bounded second partial derivatives; (iii) the nuisance estimator satisfying $\|\hat e-e\|_2=o_p(1)$ and
			$\|\hat F_a-F_a\|_2=o_p(1)$ for $a = 0,1$; (iv)
			$e(x),F_a(k\mid x),\Delta_{ak}^{\psi}(x)$ and their estimators are in a Donsker class.
			Let $\hat\psi$ be the one-step estimator defined in \eqref{eq:onestep}, then
			\[
			\hat\psi-\psi
			=
			\hat{E}\{\,IF_\psi(O)\}
			+
			O_p(\mathrm{Rem})+o_p(n^{-1/2}),
			\]
			where $
			\mathrm{Rem}
			=
			\|\hat F_1-F_1\|_2^2+\|\hat F_0-F_0\|_2^2
			+
			\|\hat e-e\|_2\{\|\hat F_1-F_1\|_2+\|\hat F_0-F_0\|_2\}.$
			
			If $\|\hat e-e\|_2=o_p(n^{-1/4})$ and $\|\hat F_a-F_a\|_2=o_p(n^{-1/4})$ for $a =0,1$, then
			$\mathrm{Rem}=o_p(n^{-1/2})$ and
			\[
			n^{1/2}\left(\hat\psi-\psi\right)\ \Rightarrow\ \mathcal{N}\left(0,\ E\{IF_\psi\left(O\right)^2\}\right).
			\]
		\end{theorem}
		
		The proof is in Section~\ref{sec: supp-proofmain} of the Supplementary 
		Material. The theorem shows that $\hat\psi$ is asymptotically linear 
		with the nonparametric efficient influence function provided 
		$\mathrm{Rem} = o_p(n^{-1/2})$; the bias is second order, involving 
		products or squares of the nuisance estimation errors. Beyond correctly 
		specified parametric working models, this supports flexible, 
		data-adaptive nuisance estimation while preserving 
		$n^{1/2}$-consistency and efficiency, as long as each nuisance 
		estimator converges at rate $o_p(n^{-1/4})$ 
		\citep{robins2008higher, chernozhukov2018double, kennedy2024semiparametric}.
		Theorem~\ref{thm:rate} imposes the Donsker condition \citep{vaart1997weak} on the
		complexity of the nuisance models, which can be relaxed by employing cross-fitting
		\citep{chernozhukov2018double}; see Section~\ref{sec: supp-cf} of the Supplementary
		Material for implementing details.

		\section{Sensitivity analysis}\label{sec: sensitivity}
		
		\subsection{Sensitivity to the specification of the copula model} 
		The estimation and inference procedures in Section~\ref{sec: est} are implicitly indexed by the copula family $C$ with a correlation parameter $\rho$. The choice of $C$ particularly depends on the tail dependence structure, often informed by domain knowledge. For example, if tail dependence is absent and the dependence is more uniform across the distribution, use the Gaussian copula. If large values of the potential outcomes tend to co-occur, e.g., patients who respond well under one treatment also respond well under another, use an upper-tail-dependent copula like Gumbel. If small values tend to co-occur, use a lower-tail-dependent copula like Clayton.  When subject-matter knowledge is not available, one can repeat the analysis with different copula families to assess the sensitivity.   
		
		Given a copula family, the sensitivity analysis with respect to $\rho$ proceeds by evaluating the estimator over a grid of values, producing a sensitivity curve $\rho \mapsto \hat\psi(\rho)$ with pointwise confidence bands.
		For a correctly specified family, this curve is guaranteed to pass through the true causal parameters $\psi$ (or $\phi, \xi$) at the data-generating value of $\rho$ and to be contained within the nonparametric sharp bounds of \citet{lu2018treatment,lu2020sharp}. 
		The envelope of sensitivity curves across all copula families recovers the sharp bounds, whereas any single family's curve is typically much narrower and more informative for causal effects.
		A flat curve indicates robustness to the assumed dependence strength, whereas a steep curve signals sensitivity.
		
		The range and interpretation of $\rho$ depend on the copula family;  instead, one can reparameterize the dependence using Kendall's $\tau$. For most one-parameter copula families, $\rho$ admits a one-to-one mapping to $\tau$, which lies in $(-1,1)$ and depends only on the copula, providing a scale-free measure of concordance that is directly comparable across copula classes. 
		Under the latent threshold model of Proposition~\ref{prop:latent-ordinal-copula}, the copula characterizes the dependence between the latent continuous potential outcomes $\{Y^*(1),Y^*(0)\}$ given covariates $X$, and $\tau$ indexes the corresponding latent rank dependence. 
		We therefore recommend reporting sensitivity results on the $\tau$ scale, aligning the sensitivity curve through the copula-specific $\tau$--$\rho$ relationship. 
		In practice, the range of $\tau$ can be informed by domain knowledge.
		
		Finally, we consider the implicit constant-$\rho$ condition in
		Assumption~\ref{asmp: copula}.  When the true dependence structure varies with
		covariates, the constant-$\rho$ model can be viewed as a working approximation, and
		the sensitivity curve reflects an average dependence level.  When appropriate, one can
		further parameterise $\rho(X) = \rho(X;\theta)$ and treat $\theta$ as an additional
		sensitivity parameter capturing the degree of heterogeneity.  The unconditional copula
		model of Remark~\ref{remark: constant-rho-unc} can also serve as a robustness check,
		as it avoids the constant-$\rho$ restriction altogether.  Simulation evidence in
		Section~\ref{sec: supp-addsimu} of the Supplementary Material confirms that the
		constant-$\rho$ approximation performs well under moderate misspecification, and the estimated sensitivity curves under the conditional and unconditional copula models
		are closely aligned across various data generating processes.
		
		\subsection{Sensitivity to the unconfoundedness assumption}\label{subs: unconfoundedness SA}
		To assess sensitivity to unconfoundedness in observational studies, we adopt a Rosenbaum-type bounds approach \citep{rosenbaum2002observational} under the following assumption. 
		\begin{assumption}\label{ass:li}
			(Latent ignorability). There exists an unmeasured variable $U$ such that $\{Y(0),Y(1)\}\not\!\ind A\mid X$ but $\{Y(0),Y(1)\}\ind A\mid X,U$; $U$ has a bounded influence on the odds of treatment assignment so that 
			$\Gamma^{-1}\leq \{\text{pr}(A=1\mid X, U=u)\text{pr}(A=0\mid X, U=\tilde{u})\}/\{\text{pr}(A=0\mid X, U=u)\text{pr}(A=1\mid X, U=\tilde{u})\} \leq \Gamma$ for some $\Gamma \geq 1$ and all $u, \tilde{u}$. 
		\end{assumption}
		The parameter $\Gamma \geq 1$ governs the strength of hidden 
		confounding: larger $\Gamma$ allows greater violation of 
		unconfoundedness, and $\Gamma = 1$ recovers 
		Assumption~\ref{asmp: causal}(i). Under Assumption~\ref{ass:li}, the 
		conditional margins $F_a(k \mid x) = \mathrm{pr}\{Y(a) \leq k \mid X = x\}$ 
		are no longer identifiable. Following \cite{yadlowsky2022bounds}, we 
		derive sharp bounds on $F_a(k \mid x)$ and propagate them through the 
		copula map to bound the causal parameters. For each $a = 0, 1$ and 
		$k = 0, \ldots, L-2$, let 
		$p_{a,k}(x) = \mathrm{pr}(Y \leq k \mid A = a, X = x)$ denote the 
		observed conditional cumulative distribution function.
		\begin{proposition}\label{prop:binarybounds}
			Under Assumption~\ref{ass:li}, for $\forall a, k$, and $x$, we have sharp bounds 
			$F_1(k\mid x)\in [F^-_{1,\Gamma}(k\mid x), F^+_{1,\Gamma}(k\mid x)]$, and $F_0(k\mid x)\in [F^-_{0,\Gamma}(k\mid x), F^+_{0,\Gamma}(k\mid x)],$ where $F^*_{1,\Gamma}(k\mid x)=
			e(x)p_{1,k}(x)+\{1-e(x)\}r^*_\Gamma\{p_{1,k}(x)\}$ and $F^*_{0,\Gamma}(k\mid x) =
			\{1-e(x)\}p_{0,k}(x)+e(x)r^*_\Gamma\{p_{0,k}(x)\}$, 
			for $*=+,-$, where $r^-_\Gamma(p)=p\{p+\Gamma(1-p)\}^{-1}$ and
			$r^+_\Gamma(p)=\Gamma p(\Gamma p+1-p)^{-1}.$  
		\end{proposition}
		Plugging these margin bounds into the conditional copula map yields 
		endpoint functionals $m^\pm_\psi(x; \Gamma, \rho)$ and the corresponding 
		parameters $\psi^\pm_\Gamma(\rho) = E\{m^\pm_\psi(X; \Gamma, \rho)\}$. 
		Efficient 
		one-step estimators $\hat\psi^\pm_\Gamma(\rho)$ for the endpoints 
		follow by paralleling the derivation of~\eqref{eq:onestep}; the 
		resulting interval 
		$[\hat\psi^-_\Gamma(\rho), \hat\psi^+_\Gamma(\rho)]$ brackets the 
		asymptotic confounding bias and collapses to the original one-step estimate $\hat\psi(\rho)$ at 
		$\Gamma = 1$. 
		Full derivations and 
		proofs are given in Sections~\ref{sec:supp_sa_unconf} 
		and~\ref{sec: supp-proofs-others} of the Supplementary Material.
		At each fixed $\rho$, we 
		summarize the analysis by reporting the largest $\Gamma$ at which the 
		interval still excludes a substantively meaningful null value (for 
		example, $\xi = 0$ for the relative treatment effect): this threshold 
		quantifies the magnitude of hidden confounding required to overturn 
		the conclusion, with larger values indicating greater robustness. 
		Varying $\rho$ in addition produces a two-dimensional sensitivity 
		assessment over $(\rho, \Gamma)$, with the aggregation convention 
		adopted in our application described in Section~\ref{sec:application}.

		\section{Simulations}\label{sec: simu}
		We assess the finite-sample performance of the proposed estimators through simulations. We generate covariates $X = (X_1, X_2, X_3)^\T$ as $X_j \overset{\rm iid}{\sim} \mathrm{Unif}(-1, 1)$, and treatment as $A \mid X \sim \mathrm{Bernoulli}\{e(X)\}$ with $\mathrm{logit}\{e(X)\} = \beta_1^{\T}X$ and $\beta_1 = (0.5,-0.2, 0.2,-0.2)^{\T}$.
		We generate potential outcomes with $L=5$ levels by thresholding a latent continuous model
		$Y^* = \eta_a(X) + \varepsilon_a$, where $\eta_a(X) = \beta_2^{\T}X + \delta a$ with $\beta_2 = (0.6, 0.15,0.15,0.15), \delta = 0.4$, and thresholds $ \lambda_k = \text{logit}\{(k+1)/5\}, k=0,1,2,3$.
		The errors are $\varepsilon_a = \text{logit}(U_a)$ with $(U_1,U_0)$ following the Gumbel copula model with parameter $\rho = 2$, corresponding to Kendall's $\tau = 0.5$: $\text{pr}(U_1 \leq u_1, U_0 \leq u_0) = \exp ( - [ \{-\log(u_1)\}^2 + \{-\log(u_0)\}^2 ]^{1/2} ) $, so the marginal distributions satisfy $\mathrm{logit}\{F_a(k \mid X)\} = \lambda_k - \eta_a(X)$, a proportional odds model.
		
		We consider four estimators of $\psi$: (i) $\hat\psi_{\rm Par}$: the one-step estimator with correctly specified parametric nuisance models and the Gumbel copula; (ii) $\hat\psi_{\rm ML}$: a ten-fold cross-fitted estimator with random forest nuisance models via the \textsf{R} package \texttt{grf} and the correctly specified Gumbel copula; (iii) $\hat\psi_{\rm G}$: the one-step estimator with correctly-specified parametric nuisance models and a misspecified Gaussian copula whose parameter $\rho = 0.7071$  is calibrated to match the true Kendall's $\tau$; (iv) $\hat\psi_{\rm Gb3}$: the one-step estimator with correctly-specified parametric nuisance models and a misspecified Gumbel parameter $\rho = 3$, corresponding to stronger dependence than the truth.
		Sample sizes are $n = 200$ and $1000$ with $500$ replications. The true values are computed by Monte Carlo integration with $5 \times 10^5$ draws.
		\begin{table}[htbp]
			\centering
			\caption{Simulation results for $\psi$ under a fixed copula.  Bias, SD, and RMSE are
				multiplied by~$10^3$; Cov~(\%) is the empirical coverage of 95\% confidence intervals;
				SBC~(\%) is the frequency that the 95\% confidence interval lies within the sharp bounds.}
			\label{tab:sim_main}
			\setlength{\tabcolsep}{1.5pt}
			\begin{tabular}{lrrrrrrrrrr}
				\toprule
				& \multicolumn{5}{c}{$n=200$} & \multicolumn{5}{c}{$n=1000$} \\
				\cmidrule(lr){2-6}\cmidrule(lr){7-11}
				Method & Bias & SD & RMSE & Cov (\%) & SBC (\%) & Bias & SD & RMSE & Cov (\%) & SBC (\%) \\
				\midrule
				Par      &  6.9 & 68.8 & 69.2 & 92.8 &  98.6 &  -0.0 & 30.0 & 30.0 & 95.6 & 100.0 \\
				ML       &  0.5 & 74.6 & 74.6 & 88.4 &  98.6 &  -3.9 & 31.6 & 31.9 & 94.4 & 100.0 \\
				PG       &  7.7 & 68.9 & 69.3 & 93.0 &  98.6 &   1.4 & 29.9 & 29.9 & 96.0 & 100.0 \\
				PGb$_3$  & -15.0 & 84.7 & 86.0 & 91.6 & 89.2 & -26.1 & 38.8 & 46.8 & 88.4 & 100.0 \\
				\bottomrule
			\end{tabular}
		\end{table}
		
		Table~\ref{tab:sim_main} summarises the results.  The correctly specified estimator
		$\hat\psi_{\mathrm{Par}}$ exhibits negligible bias with coverage
		close to the nominal 95\% level when $n=2000$.
		The cross-fitted estimator $\hat\psi_{\mathrm{ML}}$ also has little bias but shows under-coverage at
		$n = 200$, reflecting finite-sample variability in the nonparametric nuisance
		estimates.  
		For copula sensitivity, the misspecified Gaussian copula
		$\hat\psi_{\mathrm{G}}$ incurs only mild bias when Kendall's~$\tau$ is correctly
		matched, whereas misspecifying the dependence strength in
		$\hat\psi_{\mathrm{Gb3}}$ leads to substantial bias and degraded coverage,
		underscoring the importance of calibrating the dependence level through~$\tau$
		rather than matching the copula family.  The SBC column confirms that confidence
		intervals fall within the sharp bounds at near-unit frequency, consistent with the
		theoretical nesting relationship.  Additional scenarios assessing robustness to
		nuisance-model misspecification and varying overlap are reported in
		Section~\ref{sec: supp-addsimu} of the Supplementary Material; these simulations show similar patterns.   
		
		\begin{figure}[htbp]
			\centering
			\includegraphics[width=.9\textwidth]{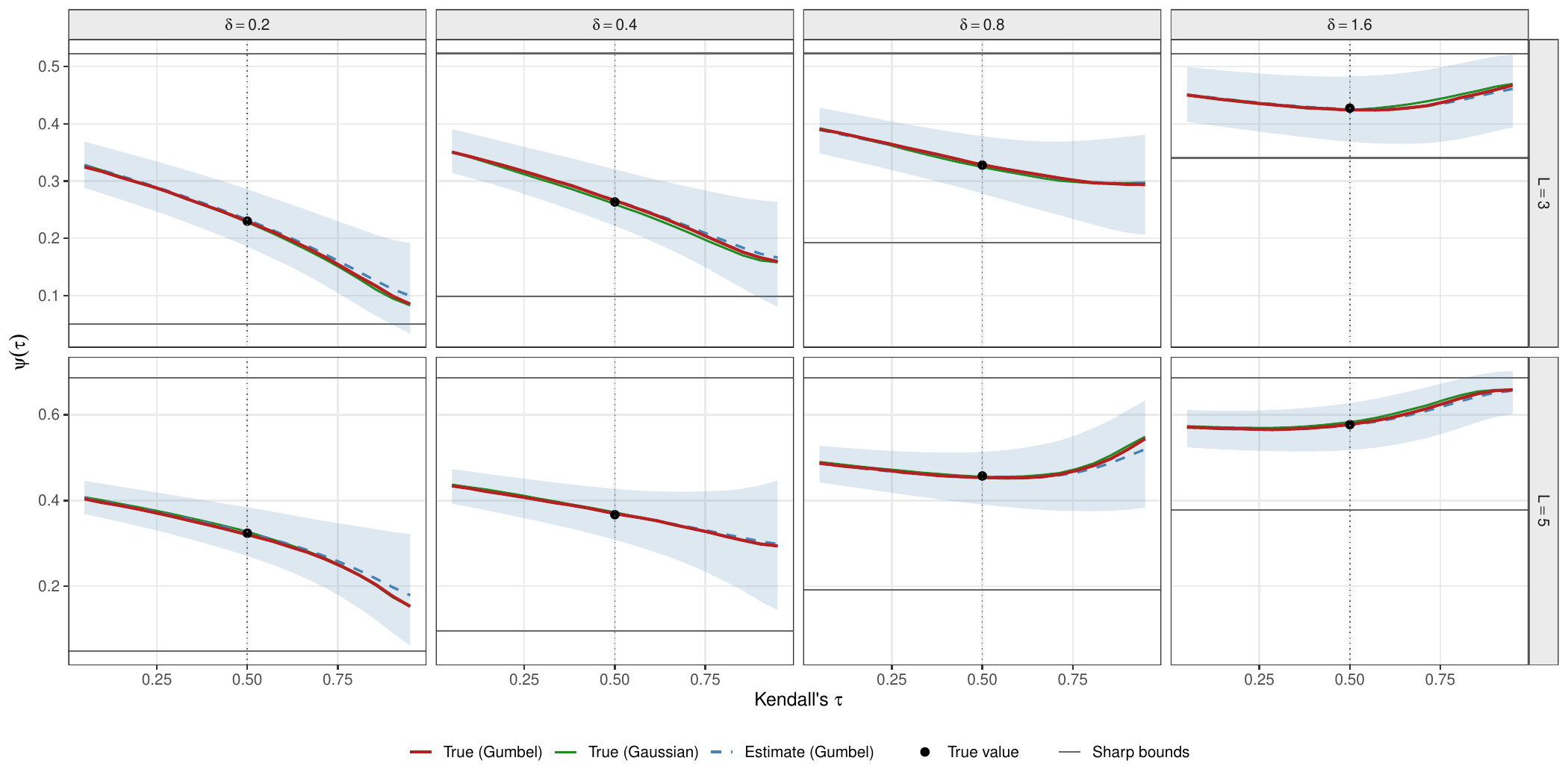}
			\caption{Sensitivity curves $\psi(\tau)$ with rows indexed by $L\in\{3,5\}$ and
				columns by $\delta\in\{0.2,0.4,0.8,1.6\}$.  Solid red and green curves are population truth under Gumbel and Gaussian copulas; dashed blue curve with shaded band is the one-step estimator and Monte Carlo 95\% range over 500 replications at $n=1{,}000$; solid grey
				lines are sharp bounds; black dot marks the truth at Kendall's $\tau=0.5$.}
			\label{fig:sensitivity-panel}
		\end{figure}
		Figure~\ref{fig:sensitivity-panel} displays the sensitivity curves under varying the copula dependence parameter.
		The informativeness of the copula structure depends on the problem
		configuration: finer ordinal scales and weaker treatment effects both widen the sharp
		bounds and steepen the sensitivity curves, as each increases the degrees of freedom in the joint distribution not pinned down by the marginals.  
		Across all panels, the
		one-step estimator $\hat{\psi}_{\mathrm{Par}}$ closely tracks the true Gumbel curve.
		The Gaussian curve closely approximates the Gumbel curve after matching $\tau$,
		although the discrepancy grows at high $\tau$ for large~$L$, suggesting caution in
		copula selection under strong dependence.
		Throughout, the copula-based confidence bands
		are substantially narrower than the sharp bounds, illustrating the inferential gain of the copula-based approach over partial identification alone.
		
		\section{Application: only-child status and psychological health}
		\label{sec:application}
		We re-analyse the publicly available data from the China Family Panel Studies survey studied by
		\cite{zeng2020being}, which evaluated whether being an only child affects self-reported
		psychological health in China. The treatment is the binary only-child status.  Three ordinal
		outcomes are considered: confidence, anxiety, and desperation, each recorded on a
		five-point Likert scale with larger values indicating better psychological condition.
		The study sample consists of individuals born after 1979, aged 16 to 31, stratified
		into four subgroups by residence and sex: rural females ($n=1747$), rural males
		($n=1708$), urban females ($n=407$), and urban males ($n=416$).
		
		Following \citet{zeng2020being}, we apply the proposed approach to each 
		subgroup separately, using their nuisance specifications but without 
		the instrumental variable. The propensity score is modelled by logistic 
		regression on covariates including maternal and paternal education, 
		age, ethnicity, household income, maternal and paternal ages at birth, 
		number of siblings, parental divorce, and parental remarriage. The 
		outcome models are proportional odds regressions on the same covariates 
		and the treatment indicator~$A$. We use Gaussian and Gumbel copulas to 
		capture distinct tail dependence structures and provide a mutual 
		robustness check.
		We focus on the relative treatment effect~$\xi$, which measures whether 
		only-child status is, on balance, beneficial or harmful. Because larger 
		outcome values correspond to better psychological condition, a 
		negative~$\xi$ indicates that only-child status is harmful.
		
		\begin{figure}[tbp]
			\centering
			\includegraphics[width=.9\textwidth]{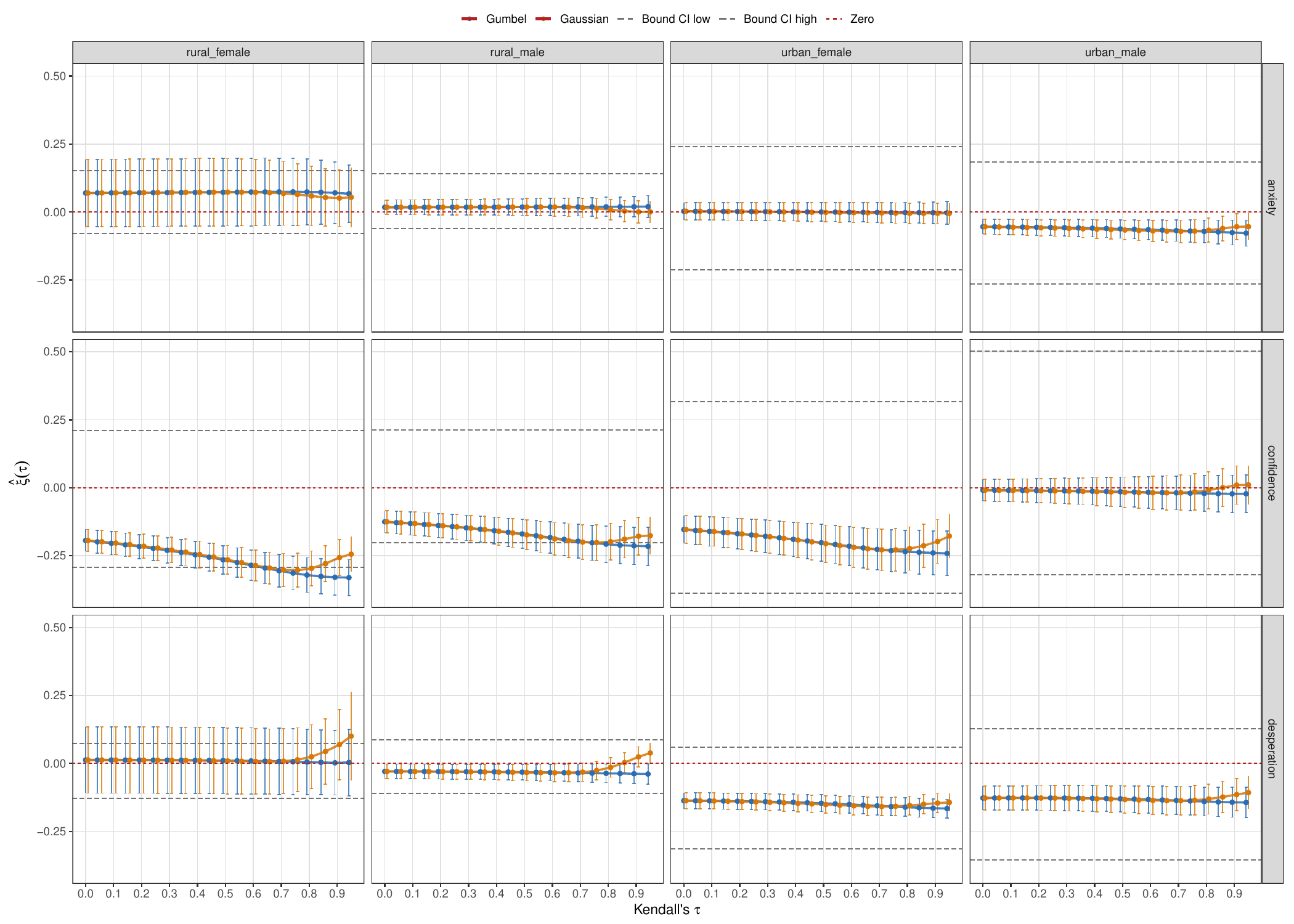}
			\caption{Subgroup sensitivity curves $\hat\xi(\tau)$ under the 
				Gaussian and Gumbel copulas, with pointwise 95\% confidence intervals. 
				The horizontal dashed red line marks $\xi=0$; the dashed grey lines mark the 95\% bootstrap confidence interval based on the sharp bounds 
				of \citet{lu2020sharp}. A curve lying below zero indicates that only-child status is harmful for the corresponding subgroup-outcome 
				combination.}
			\label{fig:application-xi}
		\end{figure}
		
		Figure~\ref{fig:application-xi} displays the subgroup-specific 
		sensitivity curves for~$\xi$ under both copula families, revealing an overall negative effect with heterogeneity across subgroups and 
		outcomes. For anxiety, both the estimated $\hat\xi(\tau)$ curve and its pointwise confidence intervals for urban males lie below zero 
		across the entire displayed range of~$\tau$ under both copula models; the adverse conclusion is therefore robust to both the choice of 
		copula family and the strength of within-pair dependence. No other subgroup exhibits a significant negative effect on anxiety. This pattern is substantively plausible: only-child urban males in China often face heightened parental and societal expectations, which may exacerbate anxiety. For confidence, significant negative effects emerge in the rural female, rural male, and urban female groups, with 
		somewhat larger effects for females than for males and for rural residents than for urban residents, consistent with sex- and region-based disparities in social expectations. For desperation, significant negative effects appear in both urban subgroups; among rural males, the Gumbel curve stays below zero across the grid while the Gaussian curve is more sensitive at large~$\tau$, indicating some copula-family sensitivity. This urban concentration may reflect greater pressure and weaker peer support among urban-only children, while rural family and community ties offer a partial buffer. 
		Our findings are mostly directionally consistent with the original 
		analysis of \citet{zeng2020being}: both indicate adverse effects of 
		only-child status, particularly in urban subgroups. Relative to their 
		analysis, however, we detect a weaker signal for anxiety and a 
		stronger signal for confidence in rural subgroups.
		For comparison, the bootstrapped confidence intervals for the covariate-assisted sharp bounds of \cite{lu2020sharp}, shown as dashed grey lines in each panel, are substantially wider and fail to detect significant effects across all subgroups and outcomes.
		
		We further apply the sensitivity analysis of
		Section~\ref{subs: unconfoundedness SA} to assess robustness to violations of
		unconfoundedness.  For each subgroup, outcome, copula family, and $\tau$, we compute
		the largest $\Gamma \ge 1$ such that the estimated bound for $\xi$ does not cross zero,
		and take the minimum across the $\tau$ grid as the curve-level robustness parameter.
		The consensus value is defined as the smaller of the Gaussian and Gumbel thresholds
		when both copula families yield a uniformly negative baseline curve, anchoring the
		result to the conditional copula curves in Figure~\ref{fig:application-xi} via the
		nesting property.  The consensus $\Gamma$ values are $1.53$ for rural female
		confidence, $1.34$ for rural male confidence, $1.51$ for urban female confidence,
		$1.29$ for urban male anxiety, $1.83$ for urban female desperation, and $1.63$ for
		urban male desperation.  The urban subgroup effects on desperation exhibit the greatest
		robustness, while the confidence effect in rural males is the most sensitive to hidden
		confounding.  Overall, these values indicate that the detected negative effects would
		persist under moderate departures from unconfoundedness.
		As an additional robustness check, we conduct the unconditional copula analysis and detect the same pattern of causal effect directions; details are in Section~\ref{sec:supp-addapp} of the
		Supplementary Material.

		\putbib
	\end{bibunit}
	
		\clearpage
		\makeatletter
		\def\@extra@b@citeb{-supp}
		\def\@extra@binfo{-supp}
		\makeatother
		
		\begin{bibunit}
			\appendix
			
			\begin{center}
				{\LARGE\bfseries Supplementary Material for `` Causal inference with ordinal outcomes: copula-based identification, estimation and sensitivity analysis"\par}
				\vspace{0.8em}
				{\large Peiyu He, Fan Li\par}
			\end{center}
			
			\addcontentsline{toc}{section}{Supplementary Material}
			\setcounter{section}{0}
			\renewcommand{\thesection}{\Alph{section}}
			\renewcommand{\theHsection}{supp.section.\arabic{section}}
			\setcounter{equation}{0}
			\renewcommand{\theequation}{S\arabic{equation}}
			\renewcommand{\theHequation}{supp.equation.\arabic{equation}}
			\setcounter{figure}{0}
			\renewcommand{\thefigure}{S\arabic{figure}}
			\renewcommand{\theHfigure}{supp.figure.\arabic{figure}}
			\setcounter{table}{0}
			\renewcommand{\thetable}{S\arabic{table}}
			\renewcommand{\theHtable}{supp.table.\arabic{table}}
			\providecommand{\theHtheorem}{\thetheorem}
			\providecommand{\theHlemma}{\thelemma}
			\providecommand{\theHcorollary}{\thecorollary}
			\providecommand{\theHproposition}{\theproposition}
			\providecommand{\theHdefinition}{\thedefinition}
			\providecommand{\theHassumption}{\theassumption}
			\providecommand{\theHremark}{\theremark}
			\providecommand{\theHstep}{\thestep}
			\providecommand{\theHcondition}{\thecondition}
			\providecommand{\theHproperty}{\theproperty}
			\providecommand{\theHrestrictions}{\therestrictions}
			\providecommand{\theHexample}{\theexample}
			\providecommand{\theHalgo}{\thealgo}
			\setcounter{theorem}{0}
			\setcounter{lemma}{0}
			\setcounter{corollary}{0}
			\setcounter{proposition}{0}
			\setcounter{definition}{0}
			\setcounter{assumption}{0}
			\setcounter{remark}{0}
			\setcounter{step}{0}
			\setcounter{condition}{0}
			\setcounter{property}{0}
			\setcounter{restrictions}{0}
			\setcounter{example}{0}
			\setcounter{algo}{0}
			\renewcommand{\thetheorem}{S\arabic{theorem}}
			\renewcommand{\thelemma}{S\arabic{lemma}}
			\renewcommand{\thecorollary}{S\arabic{corollary}}
			\renewcommand{\theproposition}{S\arabic{proposition}}
			\renewcommand{\thedefinition}{S\arabic{definition}}
			\renewcommand{\theassumption}{S\arabic{assumption}}
			\renewcommand{\theremark}{S\arabic{remark}}
			\renewcommand{\thestep}{S\arabic{step}}
			\renewcommand{\thecondition}{S\arabic{condition}}
			\renewcommand{\theproperty}{S\arabic{property}}
			\renewcommand{\therestrictions}{S\arabic{restrictions}}
			\renewcommand{\theexample}{S\arabic{example}}
			\renewcommand{\thealgo}{S\arabic{algo}}
			\renewcommand{\theHtheorem}{supp.theorem.\arabic{theorem}}
			\renewcommand{\theHlemma}{supp.lemma.\arabic{lemma}}
			\renewcommand{\theHcorollary}{supp.corollary.\arabic{corollary}}
			\renewcommand{\theHproposition}{supp.proposition.\arabic{proposition}}
			\renewcommand{\theHdefinition}{supp.definition.\arabic{definition}}
			\renewcommand{\theHassumption}{supp.assumption.\arabic{assumption}}
			\renewcommand{\theHremark}{supp.remark.\arabic{remark}}
			\renewcommand{\theHstep}{supp.step.\arabic{step}}
			\renewcommand{\theHcondition}{supp.condition.\arabic{condition}}
			\renewcommand{\theHproperty}{supp.property.\arabic{property}}
			\renewcommand{\theHrestrictions}{supp.restrictions.\arabic{restrictions}}
			\renewcommand{\theHexample}{supp.example.\arabic{example}}
			\renewcommand{\theHalgo}{supp.algo.\arabic{algo}}
			
			\section{Cross-fitting procedure}\label{sec: supp-cf}
			Following \cite{chernozhukov2018double}, we describe a cross-fitting variant of the one-step estimator that accommodates complex machine learning estimators. The procedure is summarized in Algorithm~\ref{alg:crossfit}.
			
			\begin{algo}\label{alg:crossfit}
				Cross-fitted one-step estimator for $\psi$.
				\begin{tabbing}
					\quad \= \quad \= \quad \= \kill
					\> \textit{Input:} Data $\{O_i=(Y_i,A_i,X_i):i=1,\ldots,n\}$, copula $C$, number of folds $K$.\\
					\> Randomly partition $\{1,\ldots,n\}$ into folds $\mathcal I_1,\ldots,\mathcal I_K$.\\
					\> For $s=1,\ldots,K$,\\
					\> \> Let $\mathcal I_s^c=\{1,\ldots,n\}\setminus\mathcal I_s$.\\
					\> \> Using observations in $\mathcal I_s^c$, estimate
					$\hat e^{(-s)}(x)$ and $\hat F_a^{(-s)}(k\mid x)$,
					$a=0,1$, $k=0,\ldots,L-2$.\\
					\> \> For each $i\in\mathcal I_s$, compute
					$\hat m_\psi^{(-s)}(X_i)$ and
					$\hat\Delta_{ak}^{\psi,(-s)}(X_i)$ from
					\eqref{eq: idenpsi} and \eqref{eq: deltas}.\\
					\> \> For each $i\in\mathcal I_s$, set\\
					\> \> \>
					$\displaystyle
					\hat\Psi^{(-s)}(O_i)
					=
					\hat m_\psi^{(-s)}(X_i)
					+
					\sum_{a=0}^{1}
					\frac{\mathbbm{1}(A_i=a)}
					{\hat e^{(-s)}(X_i)^a\{1-\hat e^{(-s)}(X_i)\}^{1-a}}
					\sum_{k=0}^{L-2}
					\hat\Delta_{ak}^{\psi,(-s)}(X_i)
					\{\mathbbm{1}(Y_i\le k)-\hat F_a^{(-s)}(k\mid X_i)\}.
					$\\
					\> Compute\\
					\> \> 
					$\displaystyle
					\hat\psi_{\rm cf}
					=
					\frac{1}{n}
					\sum_{s=1}^{K}\sum_{i\in\mathcal I_s}
					\hat\Psi^{(-s)}(O_i).
					$\\
					\> Estimate the asymptotic variance by\\
					\> \> 
					$\displaystyle
					\hat\sigma_{\rm cf}^{2}
					=
					\frac{1}{n}
					\sum_{s=1}^{K}\sum_{i\in\mathcal I_s}
					\{\hat\Psi^{(-s)}(O_i)-\hat\psi_{\rm cf}\}^{2}.
					$\\
					\> \textit{Output}: Point estimate $\hat\psi_{\rm cf}$ and the confidence interval\\[2pt]
					\> \> 
					$\displaystyle
					\hat\psi_{\rm cf}
					\pm
					z_{1-\alpha/2}\hat\sigma_{\rm cf}/\sqrt n.
					$
				\end{tabbing}
			\end{algo}
			
			Other estimands are similar and omitted.

			\section{Estimation and inference for \texorpdfstring{$\phi$ and $\xi$}{phi and xi}}\label{sec: supp-otherestimands}
			
			This section provides analogous copula-based results for the probability of beneficial effect $\phi = \text{pr}\{Y(1)\ge Y(0)\}$ and the relative treatment effect $\xi = \text{pr}\{Y(1)>Y(0)\} - \text{pr}\{Y(1)<Y(0)\}$.
			Under Assumptions \ref{asmp: causal}--\ref{asmp: copula}, these estimands are identified by $
			\phi = E\{m_\phi(X)\}, \xi = E\{m_\xi(X)\}$, where
			\[\begin{aligned}
				m_\phi(x) &= \text{pr}\{Y(1)\ge Y(0)\mid X=x\}
				= \sum_{k=0}^{L-1}\sum_{j=0}^{k}\pi_{kj}(x),\\
				m_\xi(x) &= \text{pr}\{Y(1)>Y(0)\mid X=x\} - \text{pr}\{Y(1)<Y(0)\mid X=x\} \\
				&= \sum_{k=0}^{L-1}\left\{\sum_{j=0}^{k}\pi_{kj}(x)+ \sum_{j=0}^{k-1}\pi_{kj}(x)\right\}.
			\end{aligned}\]
			Here we use $\text{pr}\{Y(1)<Y(0)\}=1-\phi$ and $
			\xi = \phi + \psi - 1, m_\xi(x)=m_\phi(x)+m_\psi(x)-1.$
			
			For $k=0,\ldots,L-2$, define
			\[\begin{aligned}
				\Delta^\phi_{1k}(x)
				&=
				\frac{\partial m_\phi(x)}{\partial F_1(k\mid x)}
				=
				C_1\{F_1(k\mid x),F_0(k\mid x)\}
				-
				C_1\{F_1(k\mid x),F_0(k+1\mid x)\}\\
				\Delta^\phi_{0k}(x)
				&=
				\frac{\partial m_\phi(x)}{\partial F_0(k\mid x)}
				=
				C_0\{F_1(k\mid x),F_0(k\mid x)\}
				-
				C_0\{F_1(k-1\mid x),F_0(k\mid x)\}.
			\end{aligned}\]
			Because $m_\xi(x)=m_\phi(x)+m_\psi(x)-1$, the corresponding derivatives for $\xi$ are
			\[
			\begin{aligned}
				\Delta^\xi_{1k}(x)
				&=
				\Delta^\phi_{1k}(x)+\Delta^{\psi}_{1k}(x)
				=
				\dot{C}_1\{F_1(k\mid x),F_0(k-1\mid x)\}
				-
				\dot{C}_1\{F_1(k\mid x),F_0(k+1\mid x)\},\\
				\Delta^\xi_{0k}(x)
				&=
				\Delta^\phi_{0k}(x)+\Delta^{\psi}_{0k}(x)
				=
				\dot{C}_0\{F_1(k+1\mid x),F_0(k\mid x)\}
				-
				\dot{C}_0\{F_1(k-1\mid x),F_0(k\mid x)\}.
			\end{aligned}
			\]
			The efficient influence functions for $\phi$ and $\xi$ are summarized as follows.
			\begin{proposition}
				Under Assumptions \ref{asmp: causal}--\ref{asmp: copula}, the efficient influence functions for $\phi$ and $\xi$ in the nonparametric model $\mathcal{M}_{\rm np}$ are
				\[
				\begin{aligned}
					IF_\phi(O)
					&=
					\{m_\phi(X)-\phi\}
					+
					\sum_{a=0}^1
					\frac{\mathbbm{1}(A=a)}{e(X)^a\{1-e(X)\}^{1-a}}
					\sum_{k=0}^{L-2}
					\Delta^\phi_{ak}(X)\,
					\{\mathbbm{1}(Y\le k)-F_a(k\mid X)\},\\
					IF_\xi(O)
					&=
					\{m_\xi(X)-\xi\}
					+
					\sum_{a=0}^1
					\frac{\mathbbm{1}(A=a)}{e(X)^a\{1-e(X)\}^{1-a}}
					\sum_{k=0}^{L-2}
					\Delta^\xi_{ak}(X)\,
					\{\mathbbm{1}(Y\le k)-F_a(k\mid X)\}.
				\end{aligned}
				\]
				The corresponding semiparametric efficiency bounds are $
				E\{IF_\phi(O)^2\} $ and $E\{IF_\xi(O)^2\}$.
			\end{proposition}
			
			Let $\hat e(x)$ and $\hat F_a(k\mid x)$ be flexible estimators of $e(x)$ and $F_a(k\mid x)$, and let
			$\hat m_\phi(x)$, $\hat m_\xi(x)$, $\hat\Delta^\phi_{ak}(x)$, and $\hat\Delta^\xi_{ak}(x)$
			be the corresponding plug-in estimators. The efficient one-step estimators are
			\[
			\begin{aligned}
				\hat\phi
				&=
				\hat E\left(
				\hat m_\phi(X)
				+
				\sum_{a=0}^1
				\frac{\mathbbm{1}(A=a)}{\hat e(X)^a\{1-\hat e(X)\}^{1-a}}
				\sum_{k=0}^{L-2}
				\hat\Delta^\phi_{ak}(X)\,
				\{\mathbbm{1}(Y\le k)-\hat F_a(k\mid X)\}
				\right),\\
				\hat\xi
				&=
				\hat E\left(
				\hat m_\xi(X)
				+
				\sum_{a=0}^1
				\frac{\mathbbm{1}(A=a)}{\hat e(X)^a\{1-\hat e(X)\}^{1-a}}
				\sum_{k=0}^{L-2}
				\hat\Delta^\xi_{ak}(X)\,
				\{\mathbbm{1}(Y\le k)-\hat F_a(k\mid X)\}
				\right).
			\end{aligned}
			\]
			
			If the same nuisance estimators are used throughout, then $
			\hat\xi = \hat\phi + \hat\psi - 1$.
			
			\begin{theorem}
				Assume the conditions of Theorem \ref{thm:rate}. 
				Then
				\[
				\begin{aligned}
					\hat\phi - \phi
					&=
					\hat E\{IF_\phi(O)\}
					+
					O_p(\mathrm{Rem})
					+
					o_p(n^{-1/2}),\\
					\hat\xi - \xi
					&=
					\hat E\{IF_\xi(O)\}
					+
					O_p(\mathrm{Rem})
					+
					o_p(n^{-1/2}),
				\end{aligned}
				\]
				where $
				\mathrm{Rem}
				=
				\|\hat F_1-F_1\|_2^2
				+
				\|\hat F_0-F_0\|_2^2
				+
				\|\hat e-e\|_2
				\left(
				\|\hat F_1-F_1\|_2
				+
				\|\hat F_0-F_0\|_2
				\right).$
				
				If $
				\|\hat e-e\|_2=o_p(n^{-1/4}),
				\|\hat F_a-F_a\|_2=o_p(n^{-1/4}), a=0,1,$
				then $\mathrm{Rem}=o_p(n^{-1/2})$ and
				\[
				\begin{aligned}
					n^{1/2}(\hat\phi-\phi)
					&\Rightarrow
					N\left(0,E\{IF_\phi(O)^2\}\right),\\
					n^{1/2}(\hat\xi-\xi)
					&\Rightarrow
					N\left(0,E\{IF_\xi(O)^2\}\right).
				\end{aligned}
				\]
				Therefore, $\hat\phi$ and $\hat\xi$ achieve the semiparametric efficiency bound under the same
				rate conditions as $\hat\psi$.
			\end{theorem}
			
			Proofs in this section are entirely parallel to those of Proposition \ref{prop: eif} and Theorem \ref{thm:rate}. 
			For $\phi$, one applies the same pathwise
			differentiation argument to the functional $m_\phi(x)$ above. For $\xi$, the result follows either
			by repeating the same calculation with $m_\xi(x)$, or directly from the linear identity $
			\xi=\phi+\psi-1,$
			which implies $IF_\xi=IF_\phi+IF_\psi $
			and yields the one-step estimator and rate-doubly-robust expansion immediately. 
			Similarly, the Donsker condition can be removed by cross-fitting. Further details are omitted.
			
			\section{The alternative unconditional copula model}\label{sec: supp-unc}
			
			We retain Assumption 1 and the nonparametric observed-data model $\mathcal{M}_{\mathrm{np}}$, but replace Assumption 2 by an unconditional copula model for the joint distribution of the potential outcomes:
			\begin{assumption}\label{asmp: copulauncond}
				There exists a known copula $C_\rho$ such that
				\begin{equation}\label{eq: copulauncond}
					\text{pr}\{Y(1)\le k,\;Y(0)\le j\}
					=
					C_\rho\!\left\{F_1(k),\,F_0(j)\right\}, \quad \forall k, j = 0,\ldots,L-1.
				\end{equation}
			\end{assumption}
			Here $F_a(k) = \text{pr}\{Y(a) \le k\}, a=0,1.$
			This assumption impose the copula structure directly on the unconditional joint law of $\{Y(1),Y(0)\}$, the copula parameter encodes the unconditional dependence between $\{Y(1),Y(0)\}$, or can be interpreted to characterize the rank dependence between latent continuous variable $\{Y^*(1), Y^*(0)\}$ through an analogous argument of Proposition \ref{prop:latent-ordinal-copula} in the main text. 
			Compared with Assumption \ref{asmp: copula}, this model is more general because the dependence between $Y(1)$ and $Y(0)$ is no longer restricted to be the same across $X$. However, it has no ordinal regression interpretations.
			
			\begin{proposition}[Latent continuous interpretation]\label{prop:unc-latent-ordinal-copula}
				Suppose for $a = 0,1$,
				\[Y(a) = \ell \iff \lambda_{\ell-1} < Y^*(a) \leq \lambda_\ell, \quad \ell = 0, \ldots, L-1,\]
				where $Y^*(a)$ is latent continuous potential outcome and $-\infty = \lambda_{-1} < \lambda_0 < \cdots < \lambda_{L-2} < \lambda_{L-1} = +\infty$ are thresholds shared by both treatment arms.
				If the joint distribution of the latent
				continuous potential outcomes satisfies
				\[
				\mathrm{pr}\!\left\{Y^*(1) \leq y_1,\, Y^*(0) \leq y_0\right\} = C_\rho\!\left\{F^*_1(y_1),\, F^*_0(y_0)\right\}, \quad \forall\, y_1, y_0 \in \mathbb{R},
				\]
				where $F^*_a(y) = \mathrm{pr}\{Y^*(a) \leq y\}$, $a = 0, 1$, then Assumption~\ref{asmp: copulauncond} holds with the same copula $C_{\rho}$. 
			\end{proposition}

			Under Assumption \ref{asmp: causal}, the unconditional margins are identified by either outcome regression or inverse probability weighting:
			\begin{equation}
				F_a(k) = E\{F_a(k \mid X)\}
				= E\left[\frac{\mathbbm{1}(A=a)\mathbbm{1}(Y \le k)}{\text{pr}(A=a\mid X)}\right], \quad a=0,1.
			\end{equation}
			Set $F_a(-1)=0$ and $F_a(L-1)=1$. 
			Then the joint cell probabilities $\pi_{kj} = \text{pr}\{Y(1)=k, Y(0)=j\}$ are identified by rectangle increments under Assumption \ref{asmp: copulauncond}:
			$
			\pi_{kj} = C\{F_1(k),F_0(j)\} - C\{F_1(k-1),F_0(j)\}
			- C\{F_1(k),F_0(j-1)\} + C\{F_1(k-1),F_0(j-1)\}.$
			And $\psi$ is identified by
			\begin{equation}\label{eq: uncond-iden-formula}
				\psi = \text{pr}\{Y(1)>Y(0)\}
				= \sum_{k=0}^{L-1}\sum_{j=0}^{k-1}\pi_{kj}
				= \sum_{k=0}^{L-1}\left[C\{F_1(k),F_0(k-1)\} - C\{F_1(k-1),F_0(k-1)\}\right].
			\end{equation}
			Therefore, under the unconditional copula model, $\psi$ is identified as a smooth functional of the finite-dimensional vector $\{F_1(0),\ldots,F_1(L-2),F_0(0),\ldots,F_0(L-2)\}$.
			
			Assume that $C(u,v)$ is differentiable in each argument, with partial derivatives
			\[
			\dot{C}_1(u,v) = \partial C(u,v)/\partial u, \quad
			\dot{C}_0(u,v) = \partial C(u,v)/\partial v.
			\]
			For $k=0,\ldots,L-2$, define
			\begin{equation}\label{eq: uncond-deriv}
				\begin{aligned}
					\Delta_{1k}^{\psi}
					&= \frac{\partial \psi}{\partial F_1(k)}
					= \dot{C}_1\{F_1(k),F_0(k-1)\} - \dot{C}_1\{F_1(k),F_0(k)\},\\
					\Delta_{0k}^{\psi}
					&= \frac{\partial \psi}{\partial F_0(k)}
					= \dot{C}_0\{F_1(k+1),F_0(k)\} - \dot{C}_0\{F_1(k),F_0(k)\}.
				\end{aligned}
			\end{equation}
			
			For each $a=0,1$ and $k=0,\ldots,L-2$, define
			\[
			D_{ak}(O)
			=
			\frac{\mathbbm{1}(A=a)}{e(X)^a\{1-e(X)\}^{1-a}}
			\{\mathbbm{1}(Y \le k)-F_a(k \mid X)\}
			+
			F_a(k \mid X)-F_a(k).
			\]
			Here $D_{ak}(O)$ is the efficient influence function for the marginal parameter $F_a(k)$ in $\mathcal{M}_{\mathrm{np}}$. By the chain rule, the efficient influence function for $\psi$ is
			\[
			IF_\psi(O)
			=
			\sum_{k=0}^{L-2}\Delta^{\psi}_{1k}D_{1k}(O)
			+
			\sum_{k=0}^{L-2}\Delta^{\psi}_{0k}D_{0k}(O).
			\]
			We summarize the results in the following proposition.
			\begin{proposition}[Efficient influence function]\label{prop:eifuncond}
				Under Assumptions~\ref{asmp: causal} and \ref{asmp: copulauncond}, the efficient influence function for $\psi$
				in the nonparametric model $\mathcal{M}_{\mathrm{np}}$ is
				\[
				IF_\psi(O)
				=
				\sum_{a=0}^1\sum_{k=0}^{L-2}\Delta^{\psi}_{ak}
				\left\{
				\frac{\mathbbm{1}(A=a)}{e(X)^a\{1-e(X)\}^{1-a}}
				\{\mathbbm{1}(Y \le k)-F_a(k \mid X)\}
				+
				F_a(k \mid X)-F_a(k)
				\right\}
				\]
				The semiparametric efficiency bound for $\psi$ in $\mathcal{M}_{\rm np}$ equals $E\{IF_\psi(O)^2\}$.
			\end{proposition}

			We next construct a doubly-robust estimator that is consistent if either working model of $e(X)$ and $ F_a(k \mid X)$ is correct, and it can achieve the semiparametric efficiency bound under both models are correct. 
			Given estimators $\hat e(X)$ and $\hat F_a(k \mid X)$, define
			\[
			\hat F_a^{\rm dr}(k)
			=
			\hat E\left[
			\hat F_a(k \mid X)
			+
			\frac{\mathbbm{1}(A=a)}{\hat e(X)^a\{1-\hat{e}(X)\}^{1-a}}
			\{\mathbbm{1}(Y \le k)-\hat F_a(k \mid X)\}
			\right], \quad a=0,1,\; k=0,\ldots,L-2,\]
			where $\hat E$ is the empirical mean operator. 
			Set $\hat F_a^{\,dr}(-1)=0$ and $\hat F_a^{\,dr}(L-1)=1$, and define
			\begin{equation}\label{eq: uncdrest}
				\hat\psi_{\rm dr}
				=
				\sum_{k=0}^{L-1}
				\left[
				C\{\hat F_1^{\rm dr}(k),\hat F_0^{\rm dr}(k-1)\}
				-
				C\{\hat F_1^{\rm dr}(k-1),\hat F_0^{\rm dr}(k-1)\}
				\right].
			\end{equation}
			This is a substitution estimator obtained by plugging the doubly robust estimators of the unconditional margins into the copula functional.
			
			\begin{theorem}\label{thm: uncthm}
				Under Assumptions \ref{asmp: causal} and \ref{asmp: copulauncond}, suppose $c \le e(X), \hat e(X) \le 1-c$ for some $c>0$, 
				then:
				
				\begin{enumerate}
					\item[(i)] \textit{(Double robustness)} If either $\hat{e}(X)$ is consistent for $e(X)$, or $\hat{F}_a(k\mid X)$ is consistent
					for $F_a(k\mid X)$ for all $a = 0,1$ and $k=0,\ldots,L-2$, then
					\[
					\hat\psi_{\rm dr} \stackrel{p}{\longrightarrow} \psi.
					\]
					
					\item[(ii)] \textit{(Asymptotic normality and efficiency.)} If $c\leq F_a(k\mid X), \hat F_a(k\mid x)\leq 1-c$, $\|\hat e-\bar e\|_2 = o_p(1), \|\hat F_a-\bar F_a\|_2 = o_p(1), a=0,1$, the copula $C_{\rho}$ is twice continuously differentiable on $[c, 1-c]^2$ with bounded second partial derivatives, and all nuisance functions and their estimators are in a Donsker class, then
					\[
					\hat\psi_{\rm dr}-\psi
					=
					\hat E\{IF_\psi(O)\}
					+ O_p(\mathrm{Rem})
					+ o_p(n^{-1/2}),
					\]
					where $
					\mathrm{Rem}
					=
					\|\hat e-e\|_2
					\sum_{a=0}^1 \|\hat F_a-F_a\|_2.$
					Consequently, if $\|\hat e-e\|_2=o_p(n^{-1/4})$ and $\|\hat F_a-F_a\|_2=o_p(n^{-1/4})$ for $a =0,1$, then
					$\mathrm{Rem}=o_p(n^{-1/2})$ and
					\[
					n^{1/2}(\hat\psi_{\rm dr}-\psi)
					\Rightarrow
					N\!\left(0, E\{IF_\psi(O)^2\}\right).
					\]
				\end{enumerate}
			\end{theorem}
			
			\begin{remark}
				Outcome regression and inverse probability weighting estimators can be obtained by
				substituting into~\eqref{eq: idenpsi} the marginal estimators
				\[
				\hat{F}_a^{\mathrm{or}}(k) = \hat{E}\{\hat{F}_a(k \mid X)\}, \quad
				\hat{F}_a^{\mathrm{ipw}}(k) = \hat{E}\left\{
				\frac{\mathbbm{1}(A=a)\,\mathbbm{1}(Y \leq k)}
				{\hat{e}(X)^a\{1-\hat{e}(X)\}^{1-a}}\right\},
				\]
				respectively.
				Under correctly specified parametric working models
				$F_a(k \mid X;\alpha)$ and $e(X;\gamma)$, inference on $\psi$ can be conducted
				via the generalised method of moments with sandwich standard errors, treating
				$(\psi, \alpha, \gamma)$ as a joint parameter vector.
			\end{remark}
			
			For other estimands under the unconditional copula model of Assumption \ref{asmp: copulauncond}, the probability of beneficial effect $\phi$ can be identified as
			\[
			\phi = \text{pr}\{Y(1)\ge Y(0)\}
			= \sum_{k=0}^{L-1}
			\left[
			C\{F_1(k),F_0(k)\}
			-
			C\{F_1(k-1),F_0(k)\}
			\right].
			\]
			For $k=0,\ldots,L-2$, let
			\[
			\Delta^\phi_{1k}
			=
			C_1\{F_1(k),F_0(k)\}
			-
			C_1\{F_1(k),F_0(k+1)\},
			\quad
			\Delta^\phi_{0k}
			=
			C_0\{F_1(k),F_0(k)\}
			-
			C_0\{F_1(k-1),F_0(k)\}.
			\]
			The efficient influence function for $\phi$ in $\mathcal{M}_{\rm np}$ is
			\[
			IF_\phi(O)
			=
			\sum_{a=0}^1\sum_{k=0}^{L-2}\Delta^{\phi}_{ak}
			\left\{
			\frac{\mathbbm{1}(A=a)}{e(X)^a\{1-e(X)\}^{1-a}}
			\{\mathbbm{1}(Y \le k)-F_a(k \mid X)\}
			+
			F_a(k \mid X)-F_a(k)
			\right\}
			\]
			The corresponding efficient and doubly robust estimator is
			\[
			\hat\phi_{\rm dr}
			=
			\sum_{k=0}^{L-1}
			\left[
			C\{\hat F_1^{\rm dr}(k),\hat F_0^{\rm dr}(k)\}
			-
			C\{\hat F_1^{\rm dr}(k-1),\hat F_0^{\rm dr}(k)\}
			\right].
			\]
			
			For the relative effect $
			\xi = \text{pr}\{Y(1)>Y(0)\} - \text{pr}\{Y(1)<Y(0)\} = \psi+\phi-1$,
			The efficient influence function is $
			IF_\xi(O)=IF_\psi(O)+IF_\phi(O).$
			Equivalently, for $k=0,\ldots,L-2$, define 
			\[
			\Delta^\xi_{1k}
			=
			\Delta^{\psi}_{1k}+\Delta^\phi_{1k}
			=
			C_1\{F_1(k),F_0(k-1)\}
			-
			C_1\{F_1(k),F_0(k+1)\},
			\]
			\[
			\Delta^\xi_{0k}
			=
			\Delta^{\psi}_{0k}+\Delta^\phi_{0k}
			=
			C_0\{F_1(k+1),F_0(k)\}
			-
			C_0\{F_1(k-1),F_0(k)\},
			\]
			so that
			\[
			IF_\xi(O)
			=
			\sum_{a=0}^1\sum_{k=0}^{L-2}\Delta^{\xi}_{ak}
			\left\{
			\frac{\mathbbm{1}(A=a)}{e(X)^a\{1-e(X)\}^{1-a}}
			\{\mathbbm{1}(Y \le k)-F_a(k \mid X)\}
			+
			F_a(k \mid X)-F_a(k)
			\right\}
			\]
			A natural doubly robust estimator is 
			\[
			\hat\xi_{\rm dr}=\hat\psi_{\rm dr}+\hat\phi_{\rm dr}-1.\]
			The variance can be estimated through a plug-in estimator of the efficient influence functions.
			Take $\psi$ as an example, obtain
			\[
			\hat\Delta_{1k}^{\psi}
			=
			\dot{C}_1\{\hat F_1^{\rm dr}(k),\hat F_0^{\rm dr}(k-1)\}
			-
			\dot{C}_1\{\hat F_1^{\rm dr}(k),\hat F_0^{\rm dr}(k)\},
			\]
			\[
			\hat\Delta_{0k}^{\psi}
			=
			\dot{C}_0\{\hat F_1^{\rm dr}(k+1),\hat F_0^{\rm dr}(k)\}
			-
			\dot{C}_0\{\hat F_1^{\rm dr}(k),\hat F_0^{\rm dr}(k)\},
			\]
			and
			\[
			\hat D_{ak}(O)
			=
			\frac{\mathbbm{1}(A=a)}{\hat e(X)^a\{1-\hat e(X)\}^{1-a}}
			\{\mathbbm{1}(Y \le k)-\hat F_a(k \mid X)\}
			+
			\hat F_a(k \mid X)-\hat F_a^{\rm dr}(k).
			\]
			Then
			$\hat{IF}_\psi(O)
			=
			\sum_{a=0}^1\sum_{k=0}^{L-2}\hat\Delta_{ak}^{\psi}\hat D_{ak}(O)$, and the asymptotic variance can be estimated by $\hat E\{\hat{IF}_\psi(O)^2\}$.
			
			\section{Unconfoundedness sensitivity analysis for the conditional copula one-step estimator}
			\label{sec:supp_sa_unconf}
			This section develops a formal sensitivity analysis that quantifies how the conclusions from the conditional copula one-step estimator~\eqref{eq:onestep} would change under violations of unconfoundedness. 
			For brevity, we present results for $\psi = \text{pr}\{Y(1) > Y(0)\}$; the construction for $\phi$ and $\xi$ is identical upon replacing the copula map and derivative coefficients accordingly.
			
			The development proceeds in three stages. First, we characterize the probability limit of the one-step estimator when unconfoundedness fails, showing that it targets an observed-data pseudo-target rather than the causal parameter. Second, we derive bounds on the causal parameter under a Rosenbaum-type sensitivity model. Third, we derive the efficient influence function and construct efficient one-step estimators for the bound endpoints and discuss the nesting property to connect with original one-step estimator.
			
			We begin by establishing what the one-step estimator estimates when unconfoundedness fails. 
			For $k = 0,\ldots,L-2$, define $\mathbbm{1}(Y\leq k) = \mathbbm{1}(Y \le k)$ and
			$p_{a,k}(x) = \text{pr}(Y \le k \mid A=a, X=x)$, $a = 0,1$,
			with conventions $p_{a,-1}(x) = 0$ and $p_{a,L-1}(x) = 1$. Under unconfoundedness, $p_{a,k}(x)$ coincides with the causal conditional margin $F_a(k \mid x) = \text{pr}\{Y(a) \le k \mid X = x\}$; without unconfoundedness, these two quantities generally differ. Define the observed-law copula functional
			\[
			\bar{m}_\psi(x;\rho)
			=
			\sum_{k=0}^{L-1}
			\left[
			C_{\rho}\{p_{1,k}(x), p_{0,k-1}(x)\}
			-
			C_{\rho}\{p_{1,k-1}(x), p_{0,k-1}(x)\}
			\right],
			\]
			which has the same form as $m_\psi(x)$ under unconfoundedness but with $F_a(k \mid x)$ replaced by the observed conditional $p_{a,k}(x)$. The pseudo-target is $\bar\psi(\rho) = E\{\bar{m}_\psi(X;\rho)\}$. 
			The next proposition summarizes the formal results, the discrepancy $\psi(\rho) - \bar\psi(\rho)$ reflects the confounding bias induced by unmeasured confounders, and is generally nonzero.
			\begin{proposition}[Pseudo-target]\label{prop:pseudolimit}
				Suppose $\hat{e}(x) \to e(x)$ and $\hat{F}_a(k \mid x) \to p_{a,k}(x)$ in probability for $a = 0,1$ and $k = 0,\ldots,L-2$. Then the one-step estimator in \eqref{eq:onestep} converges in probability to $\bar\psi(\rho)$.
			\end{proposition}
			
			To bound this confounding bias, we work under Assumption \ref{ass:li}, which posits that a latent variable $U$ restores unconfoundedness with the sensitivity parameter $\Gamma \ge 1$ controlling the worst-case influence of $U$ on treatment assignment odds. When $\Gamma = 1$, Assumption \ref{ass:li} reduces to the unconfoundedness of Assumption \ref{asmp: causal}.
			
			The true causal margin $F_a(k \mid x)$ can be decomposed as a mixture of the observed conditional $p_{a,k}(x)$ and the unidentifiable counterfactual conditional $r_{a,k}(x) = \text{pr}\{Y(a) \le k \mid A = 1-a, X = x\}$:
			\begin{equation}\label{eq:Fa_decomp}
				F_a(k \mid x) = \text{pr}(A=a \mid X=x)\, p_{a,k}(x) + \{1 - \text{pr}(A=a \mid X=x)\}\, r_{a,k}(x).
			\end{equation}
			Under unconfoundedness, $r_{a,k}(x) = p_{a,k}(x)$ and the decomposition collapses. In general, $r_{a,k}(x)$ is not identified, but Assumption \ref{ass:li} constrains it to a known interval. 
			Following \cite{yadlowsky2022bounds}, we can derive the closed-form sharp bounds for $r_{a,k}$. Define the bounding transforms
			\[
			r^-_\Gamma(p) = \frac{p}{p + \Gamma(1-p)}, \quad r^+_\Gamma(p) = \frac{\Gamma\, p}{\Gamma\, p + 1-p}, \quad p \in [0,1].
			\]
			Note that $r^-_\Gamma(p) \le p \le r^+_\Gamma(p)$ for $\Gamma \ge 1$, with equality when $\Gamma = 1$.
			Then the sharp bounds are characterized as $r_{a,k}(x) \in [r^-_\Gamma\{p_{a,k}(x)\},\; r^+_\Gamma\{p_{a,k}(x)\}]$.
			Plugging in \eqref{eq:Fa_decomp} yields sharp bounds for margins; the formal results are stated in Proposition \ref{prop:binarybounds}.
			
			The margin bounds in Proposition~\ref{prop:binarybounds} are each expressible as functions of the observed-data quantities $e(x)$ and $p_{a,k}(x)$, and hence are identifiable from the data for any fixed $\Gamma$. 
			It is convenient to write the endpoint transforms compactly as
			\[G^\pm_{1,\Gamma}(e,p) = ep + (1-e)r^\pm_\Gamma(p),\quad G^\pm_{0,\Gamma}(e,p) = (1-e)p + e\,r^\pm_\Gamma(p),\]
			so that $F^\pm_{1,\Gamma}(k \mid x) = G^\pm_{1,\Gamma}\{e(x), p_{1,k}(x)\}$ and $F^\pm_{0,\Gamma}(k \mid x) = G^\pm_{0,\Gamma}\{e(x), p_{0,k}(x)\}$.
			Unlike the causal margin $F_a(k \mid x)$ identified only through outcome regression under unconfoundedness, the endpoint margins $F^\pm_{a,\Gamma}(k \mid x)$ depend on both the observed outcome regression and the propensity score. This additional dependence on $e(x)$ is the source of important structural differences in the efficient influence function and the remainder term, discussed below. The derivatives of $r^\pm_\Gamma$ are
			\[
			\dot{r}^-_\Gamma(p) = \frac{\Gamma}{\{\Gamma - (\Gamma-1)p\}^2}, \quad
			\dot{r}^+_\Gamma(p) = \frac{\Gamma}{\{1 + (\Gamma-1)p\}^2},
			\]
			and the partial derivatives of the $G$-transforms are
			\[
			\begin{aligned}
				\frac{\partial G^\pm_{1,\Gamma}}{\partial e}(e,p) &= p - r^\pm_\Gamma(p), &\quad
				\frac{\partial G^\pm_{0,\Gamma}}{\partial e}(e,p) &= r^\pm_\Gamma(p) - p, \\[3pt]
				\frac{\partial G^\pm_{1,\Gamma}}{\partial p}(e,p) &= e + (1-e)\dot{r}^\pm_\Gamma(p), &\quad
				\frac{\partial G^\pm_{0,\Gamma}}{\partial p}(e,p) &= 1-e + e\,\dot{r}^\pm_\Gamma(p).
			\end{aligned}
			\]
			
			As $m_\psi(x)$ is monotone in the conditional margins that decrease in the treated margins and increase in the control margins, by properties of the copula, conservative bounds of $m_\psi(x)$ can be constructed via correspondingly plugging in marginal bounds. Define the lower and upper endpoint maps
			\begin{equation}
				\begin{aligned}
					m^-_{\psi,\Gamma}(x;\rho)
					&=
					\sum_{k=0}^{L-1}
					\left[
					C_{\rho}\{F^+_{1,\Gamma}(k \mid x), F^-_{0,\Gamma}(k{-}1 \mid x)\}
					-
					C_{\rho}\{F^+_{1,\Gamma}(k{-}1 \mid x), F^-_{0,\Gamma}(k{-}1 \mid x)\}
					\right], \label{eq:m_lower} \\
					m^+_{\psi,\Gamma}(x;\rho)
					&=
					\sum_{k=0}^{L-1}
					\left[
					C_{\rho}\{F^-_{1,\Gamma}(k \mid x), F^+_{0,\Gamma}(k{-}1 \mid x)\}
					-
					C_{\rho}\{F^-_{1,\Gamma}(k{-}1 \mid x), F^+_{0,\Gamma}(k{-}1 \mid x)\}
					\right]
				\end{aligned}
			\end{equation}
			with boundary conventions $F^\pm_{a,\Gamma}(-1 \mid x) = 0$ and $F^\pm_{a,\Gamma}(L-1 \mid x) = 1$, and define the endpoint parameters $\psi^-_\Gamma(\rho) = E\{m^-_{\psi,\Gamma}(X;\rho)\}$ and $\psi^+_\Gamma(\rho) = E\{m^+_{\psi,\Gamma}(X;\rho)\}$.
			
			\begin{proposition}[Endpoint bounds]\label{prop:endpointbounds}
				Under Assumptions \ref{asmp: causal}(ii), \ref{asmp: copula}, \ref{ass:li}, we have bounds $m_\psi(X) \in [m^-_{\psi,\Gamma}(X;\rho),\; m^+_{\psi,\Gamma}(X;\rho)]$, and therefore $\psi(\rho) \in [\psi^-_\Gamma(\rho),\; \psi^+_\Gamma(\rho)]$.
			\end{proposition}
			
			We now turn to estimation and inference for the endpoint parameters $\psi^\pm_\Gamma(\rho)$. These are smooth functionals of the observed-data nuisance $\eta = (e, \{p_{1,k}\},\{p_{0,k}\})$ through the $G$-transforms, and we adopt the same semiparametric framework as Section~\ref{sec: est} to derive efficient influence function and parallel one-step estimators. 
			We introduce some necessary notations. For the lower endpoint, the copula derivatives are
			\begin{align*}
				\Delta^{\psi,-}_{1k}(x)
				&=
				\dot{C}_1\{F^+_{1,\Gamma}(k \mid x), F^-_{0,\Gamma}(k{-}1 \mid x)\}
				-
				\dot{C}_1\{F^+_{1,\Gamma}(k \mid x), F^-_{0,\Gamma}(k \mid x)\},
				\\
				\Delta^{\psi,-}_{0k}(x)
				&=
				\dot{C}_0\{F^+_{1,\Gamma}(k{+}1 \mid x), F^-_{0,\Gamma}(k \mid x)\}
				-
				\dot{C}_0\{F^+_{1,\Gamma}(k \mid x), F^-_{0,\Gamma}(k \mid x)\},
			\end{align*}
			where $\dot{C}_1, \dot{C}_0$ are the two partial derivatives for copula $C_\rho$.  For the upper endpoint,
			\begin{align*}
				\Delta^{\psi,+}_{1k}(x)
				&=
				\dot{C}_1\{F^-_{1,\Gamma}(k \mid x), F^+_{0,\Gamma}(k{-}1 \mid x)\}
				-
				\dot{C}_1\{F^-_{1,\Gamma}(k \mid x), F^+_{0,\Gamma}(k \mid x)\},
				\\
				\Delta^{\psi,+}_{0k}(x)
				&=
				\dot{C}_0\{F^-_{1,\Gamma}(k{+}1 \mid x), F^+_{0,\Gamma}(k \mid x)\}
				-
				\dot{C}_0\{F^-_{1,\Gamma}(k \mid x), F^+_{0,\Gamma}(k \mid x)\}.
			\end{align*}
			These are the analogues of $\Delta^\psi_{ak}(x)$ in~\eqref{eq: deltas}, evaluated at the endpoint margins instead of the true margins. 
			An additional complexity arises from $m^\pm_\psi(x)$ depending on $e(x)$ through the $G$-transforms, which introduces additional augmentation terms in the efficient influence function.
			Define the propensity score augmentation coefficients, which capture the sensitivity of the endpoint maps to perturbations in $e(x)$:
			\begin{align*}
				H^-_\psi(x)
				&=
				\sum_{k=0}^{L-2}
				\Delta^{\psi,-}_{1k}(x)
				\frac{\partial G^+_{1,\Gamma}}{\partial e}\{e(x), p_{1,k}(x)\}
				+
				\sum_{k=0}^{L-2}
				\Delta^{\psi,-}_{0k}(x)
				\frac{\partial G^-_{0,\Gamma}}{\partial e}\{e(x), p_{0,k}(x)\},
				\\[3pt]
				H^+_\psi(x)
				&=
				\sum_{k=0}^{L-2}
				\Delta^{\psi,+}_{1k}(x)
				\frac{\partial G^-_{1,\Gamma}}{\partial e}\{e(x), p_{1,k}(x)\}
				+
				\sum_{k=0}^{L-2}
				\Delta^{\psi,+}_{0k}(x)
				\frac{\partial G^+_{0,\Gamma}}{\partial e}\{e(x), p_{0,k}(x)\},
			\end{align*}
			and the outcome augmentation coefficients, which capture the sensitivity to perturbations in $p_{a,k}(x)$:
			\[\begin{aligned}
				W^-_{1k}(x)
				&=
				\Delta^{\psi,-}_{1k}(x)
				\frac{\partial G^+_{1,\Gamma}}{\partial p}\{e(x), p_{1,k}(x)\},
				\quad
				&W^-_{0k}(x)
				&=
				\Delta^{\psi,-}_{0k}(x)
				\frac{\partial G^-_{0,\Gamma}}{\partial p}\{e(x), p_{0,k}(x)\},
				\\
				W^+_{1k}(x)
				&=
				\Delta^{\psi,+}_{1k}(x)
				\frac{\partial G^-_{1,\Gamma}}{\partial p}\{e(x), p_{1,k}(x)\},
				\quad
				&W^+_{0k}(x)
				&=
				\Delta^{\psi,+}_{0k}(x)
				\frac{\partial G^+_{0,\Gamma}}{\partial p}\{e(x), p_{0,k}(x)\}.
			\end{aligned}
			\]
			
			\begin{proposition}[Efficient influence function for the endpoint parameters]\label{prop:sa_eif}
				Under Assumptions~\ref{asmp: causal}(ii), \ref{asmp: copula}, \ref{ass:li}, the efficient influence functions for $\psi^-_\Gamma(\rho)$ and $\psi^+_\Gamma(\rho)$ in the nonparametric model $\mathcal{M}_{\mathrm{np}}$ are, respectively,
				\begin{equation}\label{eq: IFbounds}
					\begin{aligned}
						IF^-_{\psi,\Gamma}(O)
						&=
						\{m^-_{\psi,\Gamma}(X;\rho) - \psi^-_\Gamma(\rho)\}
						+
						H^-_\psi(X)\{A - e(X)\} \\
						&\quad
						+
						\sum_{a=0}^{1}
						\frac{\mathbbm{1}(A=a)}{e(X)^a\{1-e(X)\}^{1-a}}
						\sum_{k=0}^{L-2}
						W^-_{ak}(X)\{\mathbbm{1}(Y\leq k) - p_{a,k}(X)\},\\
						IF^+_{\psi,\Gamma}(O)
						&=
						\{m^+_{\psi,\Gamma}(X;\rho) - \psi^+_\Gamma(\rho)\}
						+
						H^+_\psi(X)\{A - e(X)\} \\
						&\quad
						+
						\sum_{a=0}^{1}
						\frac{\mathbbm{1}(A=a)}{e(X)^a\{1-e(X)\}^{1-a}}
						\sum_{k=0}^{L-2}
						W^+_{ak}(X)\{\mathbbm{1}(Y\leq k) - p_{a,k}(X)\}.
					\end{aligned}
				\end{equation}
				The semiparametric efficiency bounds for $\psi^\pm_\Gamma(\rho)$ in $\mathcal{M}_{\mathrm{np}}$ are $E[\{IF^\pm_{\psi,\Gamma}(O)\}^2]$.
			\end{proposition}
			
			Compared with the efficient influence function $IF_\psi(O)$ in Proposition~2, the sensitivity bounds influence functions~\ref{eq: IFbounds} contain additional propensity score augmentation terms $H^\pm_\psi(X)\{A - e(X)\}$ and updated marginal derivatives $ W^\pm$, accounting for the dependency of $m^\pm_\psi$ on $e(x)$ through transformation $G$.
			
			Let $\hat{e}(x)$ and $\hat{p}_{a,k}(x)$ be flexible nuisance estimators for $e(x)$ and $p_{a,k}(x)$ and let $\hat{m}^\pm_\psi$, $\hat{H}^\pm_\psi$, $\hat{W}^\pm_{ak}$ denote the corresponding plug-in estimators.
			Following \eqref{eq: IFbounds}, the lower and upper bound one-step estimators are 
			\begin{equation}\label{eq: onestep-bounds}
				\begin{aligned}
					\hat\psi^-_{\Gamma}(\rho)
					&=
					\hat{E}\left\{
					\hat{m}^-_{\psi,\Gamma}(X;\rho)
					+
					\hat{H}^-_\psi(X)\{A - \hat{e}(X)\}
					\right.\\
					&\quad\left.
					+
					\sum_{a=0}^{1}
					\frac{\mathbbm{1}(A=a)}{\hat{e}(X)^a\{1-\hat{e}(X)\}^{1-a}}
					\sum_{k=0}^{L-2}
					\hat{W}^-_{ak}(X)\{\mathbbm{1}(Y\leq k) - \hat{p}_{a,k}(X)\}
					\right\},
					\\
					\hat\psi^+_{\Gamma}(\rho)
					&=
					\hat{E}\left\{
					\hat{m}^+_{\psi,\Gamma}(X;\rho)
					+
					\hat{H}^+_\psi(X)\{A - \hat{e}(X)\}
					\right.\\
					&\quad\left.
					+
					\sum_{a=0}^{1}
					\frac{\mathbbm{1}(A=a)}{\hat{e}(X)^a\{1-\hat{e}(X)\}^{1-a}}
					\sum_{k=0}^{L-2}
					\hat{W}^+_{ak}(X)\{\mathbbm{1}(Y\leq k) - \hat{p}_{a,k}(X)\}
					\right\},
				\end{aligned}
			\end{equation}
			respectively.
			For inference, the asymptotic variance is estimated by $\hat{E}[\{\hat{IF}^\pm_{\psi,\Gamma}(O)\}^2]$, where $\hat{IF}^\pm_{\psi,\Gamma}$ is the plug-in estimator of \eqref{eq: IFbounds} with $\psi^\pm_\Gamma(\rho)$ replaced by $\hat\psi^\pm_{\Gamma}(\rho)$.
			
			The following theorem collects the main properties of the sensitivity estimators. Part~(i) establishes a rate-doubly-robust asymptotic linear expansion, paralleling Theorem~\ref{thm:rate} for the original estimator. 
			Part~(ii) shows that, under $n^{-1/4}$ convergence rates for the nuisance estimators, the sensitivity estimators are asymptotically normal and achieve the semiparametric efficiency bounds. Part~(iii) verifies the nesting property: at $\Gamma = 1$, the sensitivity interval collapses to the original point estimate, ensuring a smooth transition between the primary analysis and the sensitivity analysis.
			
			\begin{theorem}[Properties of the sensitivity estimators]\label{thm:sa_main}
				Suppose Assumptions~\ref{asmp: causal}(ii), \ref{asmp: copula}, \ref{ass:li} hold.
				\begin{enumerate}
					\item[\textup{(i)}] \textup{(Rate-doubly-robust expansion.)}
					Suppose (a) there exists $c > 0$ such that $c \le e(X), \hat{e}(X), p_{a,k}(X), \hat{p}_{a,k}(X) \le 1-c$ for some constant $c>0$ and all $a,k$; (b) the copula $C_{\rho}$ is twice continuously differentiable on $[c, 1-c]^2$ with bounded second partial derivatives; (c) the nuisance estimator satisfying $\|\hat{e} - e\|_2 = o_p(1)$ and $\|\hat{p}_a - p_a\|_2 = o_p(1)$ for $a = 0,1$; and (d) all nuisance functions and their estimators belong to a Donsker class. Then
					\[
					\hat\psi^\pm_{\Gamma}(\rho) - \psi^\pm_\Gamma(\rho)
					=
					\hat{E}\{IF^\pm_{\psi,\Gamma}(O)\}
					+
					O_p(\mathrm{Rem})
					+
					o_p(n^{-1/2}),
					\]
					where $\|\hat{p}_a - p_a\|_2^2 = \sum_{k=0}^{L-2}\|\hat{p}_{a,k} - p_{a,k}\|_2^2$ and
					\[
					\mathrm{Rem}
					=
					\|\hat{p}_1 - p_1\|_2^2
					+
					\|\hat{p}_0 - p_0\|_2^2
					+
					\|\hat{e} - e\|_2^2
					+
					\|\hat{e} - e\|_2
					\left(\|\hat{p}_1 - p_1\|_2 + \|\hat{p}_0 - p_0\|_2\right).
					\]
					
					\item[\textup{(ii)}] \textup{(Asymptotic normality and efficiency.)}
					Under conditions listed in (i), and in addition, $\|\hat{e} - e\|_2 = o_p(n^{-1/4})$ and $\|\hat{p}_a - p_a\|_2 = o_p(n^{-1/4})$ for $a=0,1$, then $\mathrm{Rem} = o_p(n^{-1/2})$ and
					\[
					n^{1/2}\{\hat\psi^\pm_{\Gamma}(\rho) - \psi^\pm_\Gamma(\rho)\}
					\Rightarrow
					\mathcal{N}\left(0,\; E[\{IF^\pm_{\psi,\Gamma}(O)\}^2]\right).
					\]
					The one-step sensitivity estimators achieve the semiparametric efficiency bounds for $\psi^\pm_\Gamma(\rho)$.
					
					\item[\textup{(iii)}] \textup{(Exact nesting.)}
					If $\Gamma = 1$ and the same nuisance estimators are used, then
					\[
					\hat\psi^-_{1}(\rho) = \hat\psi^+_{1}(\rho) = \hat\psi(\rho),
					\]
					where $\hat\psi(\rho)$ is the one-step estimator in \eqref{eq:onestep}.
				\end{enumerate}
			\end{theorem}

			Finally, the sensitivity interval for the one-step estimator under hidden confounding is
			\[
			\hat\psi(\rho)
			+
			\left[\hat{b}^-_\Gamma(\rho),\; \hat{b}^+_\Gamma(\rho)\right],
			\quad
			\hat{b}^\pm_\Gamma(\rho) = \hat\psi^\pm_{\Gamma}(\rho) - \hat\psi(\rho).
			\]
			This interval quantifies how the original point estimate would shift under
			hidden confounding of strength~$\Gamma$.  When reported alongside the copula
			sensitivity curves, it provides a two-dimensional sensitivity assessment
			over both the dependence parameter~$\rho$ and the confounding
			parameter~$\Gamma$, bounding the maximal deviation from the current curve
			under a permitted degree of unconfoundedness violation.
			
			\section{Proofs of main theoretical results}\label{sec: supp-proofmain}
			\subsection*{Proof of Proposition~\ref{prop:latent-ordinal-copula}}
			
			Fix $x$, under the latent threshold model,
			\[
			Y^*(a)=\eta_a(X)+\varepsilon_a,
			\quad
			Y(a)=\ell \Longleftrightarrow \lambda_{\ell-1}<Y^*(a)\leq \lambda_\ell,
			\quad \ell=0,\ldots,L-1,
			\]
			we have, for $a=0,1$ and $k=0,\ldots,L-2$,
			\[
			\begin{aligned}
				F_a(k\mid x)
				&=
				\text{pr}\{Y(a)\leq k\mid X=x\}  \\
				&=
				\text{pr}\{Y^*(a)\leq \lambda_k\mid X=x\}  \\
				&=
				\text{pr}\{\varepsilon_a\leq \lambda_k-\eta_a(x)\mid X=x\}  \\
				&=
				F_{\varepsilon_a\mid X}\{\lambda_k-\eta_a(x)\}.
			\end{aligned}
			\]
			As $\lambda_{L-1}=+\infty$, the above equation holds also for $k=L-1$.
			
			Suppose that the latent residuals satisfy the conditional copula model
			\[
			\text{pr}(\varepsilon_1\leq e_1,\varepsilon_0\leq e_0\mid X=x)
			=
			C_\rho\{F_{\varepsilon_1\mid X}(e_1),
			F_{\varepsilon_0\mid X}(e_0)\},
			\quad e_1,e_0\in\mathbb R .
			\]
			For \(k,j=0,\ldots,L-2\), set $
			e_1=\lambda_k-\eta_1(x),
			e_0=\lambda_j-\eta_0(x). $
			Then
			\[
			\begin{aligned}
				&\text{pr}\{Y(1)\leq k,Y(0)\leq j\mid X=x\} \\
				&\quad =
				\text{pr}\{Y^*(1)\leq \lambda_k,Y^*(0)\leq \lambda_j\mid X=x\} \\
				&\quad =
				\text{pr}\{\varepsilon_1\leq \lambda_k-\eta_1(x),
				\varepsilon_0\leq \lambda_j-\eta_0(x)\mid X=x\} \\
				&\quad =
				C_\rho\left[
				F_{\varepsilon_1\mid X}\{\lambda_k-\eta_1(x)\},
				F_{\varepsilon_0\mid X}\{\lambda_j-\eta_0(x)\}
				\right] \\
				&\quad =
				C_\rho\{F_1(k\mid x),F_0(j\mid x)\}.
			\end{aligned}
			\]
			Boundary cases when at least one of $k,j$ takes the value of $L-1$ are easy to check using $\lambda_{L-1}=+\infty$. 
			Therefore, $\text{pr}\{Y(1)\leq k,Y(0)\leq j\mid X=x\} = C_\rho\{F_1(k\mid x),F_0(j\mid x)\}$ for all $k,j$ and Assumption~\ref{asmp: copula} holds.
			
			\subsection*{Proof of Proposition \ref{prop: eif}}
			Recall the observed data $O = (X,A,Y)$. 
			The observed data law $\text{pr}(O)$ has a factorization that $\text{pr}(o)=\text{pr}_X(x)\text{pr}_{A\mid X}(a\mid x)\text{pr}_{Y\mid A,X}(y\mid a,x)$.
			Consider the regular parametric submodels
			$\{\text{pr}(o;\varepsilon) = \text{pr}_X(x;\varepsilon)\text{pr}_{A\mid X}(a\mid x;\varepsilon)\text{pr}_{Y\mid A,X}(y\mid a,x;\varepsilon):\varepsilon\in(-\delta,\delta)\}$ indexed by $\varepsilon$ with $\left.\text{pr}(O;\varepsilon)\right|_{\varepsilon=0}=\text{pr}(O)$.
			For the unrestricted nonparametric model $\mathcal{M}_{\rm np}$, any score admits the decomposition
			\[
			S(O)=\left.\frac{\partial}{\partial\varepsilon}\log\left\{\text{pr}(O;\varepsilon)\right\}\right|_{\varepsilon=0}, \quad S(o)=S_X(x)+S_{A\mid X}(a,x)+S_{Y\mid A,X}(y,a,x),
			\]
			with the constraints $
			E\{S_X(X)\}=0,
			E\{S_{A\mid X}(A,X)\mid X\}=0,
			E\{S_{Y\mid A,X}(Y,A,X)\mid A,X\}=0.$
			This yields the nonparametric tangent space that 
			\[\mathcal{T} = \left\{b\left(Y,A,X\right):E\left\{b(Y,A,X)\right\} = 0\right\}.\]
			To verify the efficient influence function, we first check $IF_{\psi} \in\mathcal{T}$, equivalently 
			$E\{IF_{\psi}(O)\} = 0$.
			By definition, $E\{m_\psi(X)-\psi\}=0$. Moreover, $
			E\left\{\mathbbm{1}\{Y\le k\}-F_a(k\mid X)\mid A=a,X\right\}=0.$
			So the weighted residual terms have conditional mean zero, implying $E\{IF_{\psi}(O)\}=0$.
			
			It remains to verify that the pathwise differentiability is satisfied if 
			\[
			\left.\frac{\partial}{\partial\varepsilon}\psi\left\{\text{pr}\left(O;\varepsilon\right)\right\}\right|_{\varepsilon=0}
			=
			E\left\{IF_{\psi}(O)S(O)\right\}.
			\]
			
			First, consider the submodel $\text{pr}(O;\varepsilon)$ where only $\text{pr}_X$ is indexed by $\varepsilon$ and other components are at the truth.
			Then $m_\psi(x)$ is fixed with respect to $\varepsilon$ and we have
			\begin{equation}\label{eq: pathdiff1}
				\left.\frac{\partial}{\partial\varepsilon}\psi\left\{\text{pr}\left(O;\varepsilon\right)\right\}\right|_{\varepsilon=0}
				=
				E\left\{m_\psi(X)S_X(X)\right\}
				=
				E\left\{(m_\psi(X)-\psi)S_X(X)\right\}
				=
				E\left\{IF_{\psi}(O)S_X(X)\right\}.
			\end{equation}
			
			Next, consider the submodel $\text{pr}(O;\varepsilon)$ with $\text{pr}_{A\mid X}$ perturbations only.
			In this case, $F_a(\cdot\mid x)$ and $m_\psi(x)$ are fixed, so the derivative is zero. 
			Since $E\{S_{A\mid X}(A,X)\mid X\}=0$ and the residual terms have conditional mean zero given $(A,X)$, one has
			\begin{equation}\label{eq: pathdiff2}
				\left.\frac{\partial}{\partial\varepsilon}\psi\left\{\text{pr}\left(O;\varepsilon\right)\right\}\right|_{\varepsilon=0} = 0=
				E\left\{IF_{\psi}\left(O\right)S_{A\mid X}\left(A,X\right)\right\}.
			\end{equation}
			
			Finally, consider the submodel $\text{pr}(O;\varepsilon)$ with $\text{pr}_{Y\mid A,X}$ perturbations only.
			Then $\text{pr}_X$ is fixed and
			\[
			\left.\frac{\partial}{\partial\varepsilon}\psi(P_\varepsilon)\right|_{\varepsilon=0}
			=
			E\left\{\left.\frac{\partial}{\partial\varepsilon}m_{\psi,\varepsilon}(X)\right|_{\varepsilon=0}\right\}.
			\]
			Since $m_\psi(x)$ is a smooth functional of the finite-dimensional vector
			$\{F_1(k\mid x),F_0(k\mid x)\}_{k=0}^{L-2}$, by the chain rule, we have
			\[
			\left.\frac{\partial}{\partial\varepsilon}m_{\psi,\varepsilon}(x)\right|_{\varepsilon=0}
			=
			\sum_{k=0}^{L-2} \Delta_{1k}(x)\left.\frac{\partial}{\partial\varepsilon}F_{1,\varepsilon}(k\mid x)\right|_{\varepsilon=0}
			+
			\sum_{k=0}^{L-2} \Delta_{0k}(x)\left.\frac{\partial}{\partial\varepsilon}F_{0,\varepsilon}(k\mid x)\right|_{\varepsilon=0}.
			\]
			For each $a$ and $k$ with $F_a(k\mid x)=E\{\mathbbm{1}\{Y\le k\}\mid A=a,X=x\}$, we have 
			\[
			\left.\frac{\partial}{\partial\varepsilon}F_{a,\varepsilon}(k\mid x)\right|_{\varepsilon=0}
			=
			E\left(\left\{\mathbbm{1}\left(Y\le k\right)-F_a\left(k\mid X\right)\right\}\,S_{Y\mid A,X}(Y,A,X)\mid A=a,X=x\right).
			\]
			Combining these identities and using
			\[
			E\left\{\frac{A}{e(X)}Z\mid X\right\}=E\left(Z\mid A=1,X\right),
			\quad
			E\left\{\frac{1-A}{1-e(X)}Z\mid X\right\}=E\left(Z\mid A=0,X\right),
			\]
			for any random variable $Z$, we have the pathwise differentiability that
			\begin{equation}\label{eq: pathdiff3}
				\left.\frac{\partial}{\partial\varepsilon}\psi\left\{\text{pr}\left(O;\varepsilon\right)\right\}\right|_{\varepsilon=0}
				=
				E\left\{IF_{\psi}(O)S_{Y\mid A,X}(Y,A,X)\right\}.
			\end{equation}
			
			Moreover, in the unrestricted model, each of the sets of such scores corresponds to perturbations of exactly one factor in the likelihood, as we have an orthogonal direct-sum decomposition of the tangent space.
			\begin{equation}\label{eq: tangentspace}
				\mathcal{T}=\mathcal{T}_X\oplus\mathcal{T}_{A\mid X}\oplus\mathcal{T}_{Y\mid A,X}.
			\end{equation}
			
			Because the map $S\mapsto \left.\frac{\partial}{\partial\varepsilon}\psi\{\text{pr}(O;\varepsilon)\}\right|_{\varepsilon=0}$ is linear in the score direction and combine \eqref{eq: tangentspace}, for an arbitrary score $S\in\mathcal{T}$ we have
			\[
			\left.\frac{\partial}{\partial\varepsilon}\psi\{\text{pr}(O;\varepsilon)\}\right|_{\varepsilon=0}
			=
			\left.\frac{\partial}{\partial\varepsilon}\psi\{\text{pr}(O;\varepsilon)\}\right|_{\varepsilon=0}\Big|_{S_X}
			+
			\left.\frac{\partial}{\partial\varepsilon}\psi\{\text{pr}(O;\varepsilon)\}\right|_{\varepsilon=0}\Big|_{S_{A\mid X}}
			+
			\left.\frac{\partial}{\partial\varepsilon}\psi\{\text{pr}(O;\varepsilon)\}\right|_{\varepsilon=0}\Big|_{S_{Y\mid A,X}},
			\]
			and combine with equations \eqref{eq: pathdiff1}--\eqref{eq: pathdiff3}, we have
			\[
			\begin{aligned}
				\left.\frac{\partial}{\partial\varepsilon}\psi\{\text{pr}(O;\varepsilon)\}\right|_{\varepsilon=0}
				&=
				E\{IF_{\psi}(O)S_X(X)\}+E\{IF_{\psi}(O)S_{A\mid X}(A,X)\}+E\{IF_{\psi}(O)S_{Y\mid A,X}(Y,A,X)\}\\
				&=
				E\{IF_{\psi}(O)S(O)\}.
			\end{aligned}
			\]
			Therefore, $\psi$ is pathwise differentiable with influence function $IF_{\psi}$. 
			Since $IF_{\psi}\in \mathcal{T}$, it is the nonparametric efficient influence function.
			
			\subsection*{Proof of Theorem \ref{thm:rate}}
			Recall that the one-step estimator in \eqref{eq:onestep} can be written as
			\[
			\hat\psi=\hat E\{\Psi(O;\hat e,\hat F_1,\hat F_0)\},
			\]
			where
			\[
			\Psi(O;e,F_1,F_0)
			=
			m_\psi(X)
			+
			\sum_{a=0}^1
			\frac{\mathbbm{1}(A=a)}{e(X)^a\{1-e(X)\}^{1-a}}
			\sum_{k=0}^{L-2}
			\Delta_{ak}(X)
			\{\mathbbm{1}(Y\le k)-F_a(k\mid X)\}.
			\]
			By Proposition \ref{prop: eif}, $IF_\psi(O)=\Psi(O;e,F_1,F_0)-\psi$.
			We can write
			\[
			\hat\psi-\psi
			=
			(\hat E-E)\{\Psi(O;\hat e,\hat F_1,\hat F_0)\}
			+
			E\{\Psi(O;\hat e,\hat F_1,\hat F_0)-\Psi(O;e,F_1,F_0)\}.
			\]
			Adding and subtracting $(\hat E-E)\Psi(O;e,F_1,F_0)$ yields
			\[
			\hat\psi-\psi
			=
			\hat E\{IF_\psi(O)\}
			+
			(\hat E-E)\{\Psi(O;\hat e,\hat F_1,\hat F_0)-\Psi(O;e,F_1,F_0)\}
			+
			E\{\Psi(O;\hat e,\hat F_1,\hat F_0)-\Psi(O;e,F_1,F_0)\}.
			\]
			
			The second term is $o_p(n^{-1/2})$ by the Donsker condition and consistency conditions, using the empirical process theory. 
			Consider the third term and denote $
			R=E\{\Psi(O;\hat e,\hat F_1,\hat F_0)-\Psi(O;e,F_1,F_0)\}.$
			Conditioning on $X$, for each $a$ and $k$,
			\[
			E\!\left\{
			\frac{\mathbbm{1}(A=a)}{\hat e(X)^a\{1-\hat e(X)\}^{1-a}}
			\hat\Delta_{ak}(X)\{1(Y\le k)-\hat F_a(k\mid X)\}
			\mid X
			\right\}
			=
			\frac{e(X)^a\{1-e(X)\}^{1-a}}{\hat e(X)^a\{1-\hat e(X)\}^{1-a}}
			\hat\Delta_{ak}(X)
			\{F_a(k\mid X)-\hat F_a(k\mid X)\}.
			\]
			Hence
			\[
			R
			=
			E\left[
			\hat m_\psi(X)-m_\psi(X)
			+
			\sum_{a=0}^1\sum_{k=0}^{L-2}
			\frac{e(X)^a\{1-e(X)\}^{1-a}}{\hat e(X)^a\{1-\hat e(X)\}^{1-a}}
			\hat\Delta_{ak}(X)
			\{F_a(k\mid X)-\hat F_a(k\mid X)\}
			\right].
			\]
			
			Since $m_\psi$ is a linear combination of copulas increments and is therefore twice continuously differentiable with bounded derivatives, apply a first-order Taylor expansion of $m_\psi(X)$ around the true marginal distributions, we have
			\[
			\hat m_\psi(X)-m_\psi(X)
			=
			\sum_{k=0}^{L-2}\Delta_{1k}(X)\{\hat F_1(k\mid X)-F_1(k\mid X)\}
			+
			\sum_{k=0}^{L-2}\Delta_{0k}(X)\{\hat F_0(k\mid X)-F_0(k\mid X)\}
			+
			r(X),
			\]
			where $
			|r(X)|
			\le
			C
			\sum_{a=0}^1\sum_{k=0}^{L-2}
			|\hat F_a(k\mid X)-F_a(k\mid X)|^2.
			$
			
			Substituting this into $R$ gives
			\[
			R
			=
			E\!\left[
			\sum_{a=0}^1\sum_{k=0}^{L-2}
			\left\{
			\Delta_{ak}(X)
			-
			\frac{e(X)^a\{1-e(X)\}^{1-a}}{\hat e(X)^a\{1-\hat e(X)\}^{1-a}}
			\hat\Delta_{ak}(X)
			\right\}
			\{\hat F_a(k\mid X)-F_a(k\mid X)\}
			+
			r(X)
			\right].
			\]
			By adding and subtracting $\hat\Delta_{ak}(X)$,
			\[
			\left|
			\Delta_{ak}(X)
			-
			\frac{e(X)^a\{1-e(X)\}^{1-a}}{\hat e(X)^a\{1-\hat e(X)\}^{1-a}}
			\hat\Delta_{ak}(X)
			\right|
			\le
			|\Delta_{ak}(X)-\hat\Delta_{ak}(X)|
			+
			\left|
			1-\frac{e(X)^a\{1-e(X)\}^{1-a}}{\hat e(X)^a\{1-\hat e(X)\}^{1-a}}
			\right||\hat\Delta_{ak}(X)|.
			\]
			The positivity condition $c\le \hat e(X)\le 1-c$ implies
			\[
			\left|
			1-\frac{e(X)^a\{1-e(X)\}^{1-a}}{\hat e(X)^a\{1-\hat e(X)\}^{1-a}}
			\right|
			\le c^{-1}|\hat e(X)-e(X)|,\quad a=0,1.
			\]
			By the smoothness of the copula derivatives in Condition (iv), $\Delta_{ak}$ is Lipschitz in the margins, so
			\[
			|\Delta_{ak}(X)-\hat\Delta_{ak}(X)|
			\le
			C\sum_{b=0}^1\sum_{j=0}^{L-2}|\hat F_b(j\mid X)-F_b(j\mid X)|,
			\]
			with boundary conventions $F_a(-1\mid X)=0$ and $F_a(L-1\mid X)=1$, and $|\hat\Delta_{ak}(X)|$ is uniformly bounded by general copula properties. Combining the three bounds gives
			\[
			\left|
			\Delta_{ak}(X)
			-
			\frac{e(X)^a\{1-e(X)\}^{1-a}}{\hat e(X)^a\{1-\hat e(X)\}^{1-a}}
			\hat\Delta_{ak}(X)
			\right|
			\le
			C\left(
			|\hat e(X)-e(X)|
			+
			\sum_{b=0}^1\sum_{j=0}^{L-2}
			|\hat F_b(j\mid X)-F_b(j\mid X)|
			\right).
			\]
			Taking expectations and applying the Cauchy--Schwarz inequality yields
			\[
			|R|
			\le
			C\left(
			\|\hat F_1-F_1\|_2^2
			+
			\|\hat F_0-F_0\|_2^2
			+
			\|\hat e-e\|_2(\|\hat F_1-F_1\|_2+\|\hat F_0-F_0\|_2)
			\right).
			\]
			This establishes the stated remainder rate.
			Combining the previous steps gives
			\[
			\hat\psi-\psi
			=
			\hat E\{IF_\psi(O)\}
			+
			O_p(\mathrm{Rem})
			+
			o_p(n^{-1/2}),
			\]
			where
			\[
			\mathrm{Rem}
			=
			\|\hat F_1-F_1\|_2^2
			+
			\|\hat F_0-F_0\|_2^2
			+
			\|\hat e-e\|_2(\|\hat F_1-F_1\|_2+\|\hat F_0-F_0\|_2).
			\]

			If $
			\|\hat e-e\|_2=o_p(n^{-1/4}),\quad
			\|\hat F_a-F_a\|_2=o_p(n^{-1/4}),$
			then $\mathrm{Rem}=o_p(n^{-1/2})$. Therefore
			\[
			n^{1/2}(\hat\psi-\psi)
			=
			n^{1/2}\hat E\{IF_\psi(O)\}+o_p(1).
			\]
			Since $IF_\psi(O)$ has mean zero and variance $E\{IF_\psi(O)^2\}$, the central limit theorem implies
			\[
			n^{1/2}(\hat\psi-\psi) \Rightarrow
			N\!\left(0,E\{IF_\psi(O)^2\}\right).
			\]
			
			This completes the proof.
			
			\subsection*{Proof of Proposition~\ref{prop:binarybounds}}
			\begin{proof}
				Following \cite{yadlowsky2022bounds}, the worst-case lower bound for $r_{a,k}(x)$ under Assumption~3 solves
				\[
				\min_{0 \le \theta \le 1}
				\frac{1}{2}\left\{p_{a,k}(x)(1-\theta)^2 + \Gamma\{1-p_{a,k}(x)\}\theta^2\right\}.
				\]
				Write $p = p_{a,k}(x)$ for brevity. Differentiating with respect to $\theta$ and setting the derivative to zero gives $-p(1-\theta) + \Gamma(1-p)\theta = 0$, hence $\theta = p/\{p + \Gamma(1-p)\} = r^-_\Gamma(p)$. The second derivative $p + \Gamma(1-p) > 0$ confirms this is a minimum. The upper bound solves the analogous problem
				\[
				\min_{0 \le \theta \le 1}
				\frac{1}{2}\left\{\Gamma\, p(1-\theta)^2 + (1-p)\theta^2\right\},
				\]
				yielding $\theta = \Gamma p/\{\Gamma p + (1-p)\} = r^+_\Gamma(p)$.
				
				Substituting the above bounds into the decomposition \eqref{eq:Fa_decomp} with $\text{pr}(A=1 \mid X=x) = e(x)$, for $a = 1$:
				\[
				F^\pm_{1,\Gamma}(k \mid x) = e(x)\,p_{1,k}(x) + \{1 - e(x)\}\,r^\pm_\Gamma\{p_{1,k}(x)\},
				\]
				and for $a = 0$, using $\text{pr}(A=0 \mid X=x) = 1 - e(x)$:
				\[
				F^\pm_{0,\Gamma}(k \mid x) = \{1 - e(x)\}\,p_{0,k}(x) + e(x)\,r^\pm_\Gamma\{p_{0,k}(x)\}. 
				\]
			\end{proof}

			\section{Proofs of supplementary results}
			\label{sec: supp-proofs-others}
			
			\subsection*{Proof of Proposition~\ref{prop:unc-latent-ordinal-copula}}
			\begin{proof}
				Under the latent threshold model, for $a=0,1$ and $k=0,\ldots,L-2$,
				\[
				F_a(k)
				=
				\text{pr}\{Y(a)\leq k\}
				=
				\text{pr}\{Y^*(a)\leq \lambda_k\}
				=
				F_a^*(\lambda_k).
				\]
				As $\lambda_{L-1}=+\infty$, the equation also holds for $k=L-1$.
				
				Suppose the latent continuous potential outcomes satisfy
				\[
				\text{pr}\{Y^*(1)\leq y_1,\,Y^*(0)\leq y_0\}
				=
				C_\rho\{F_1^*(y_1),F_0^*(y_0)\},
				\quad y_1,y_0\in\mathbb{R}.
				\]
				For $k,j=0,\ldots,L-2$, set $y_1=\lambda_k$ and $y_0=\lambda_j$. Then
				\[
				\begin{aligned}
					&\text{pr}\{Y(1)\leq k,\,Y(0)\leq j\} \\
					&\quad=
					\text{pr}\{Y^*(1)\leq \lambda_k,\,Y^*(0)\leq \lambda_j\}\\
					&\quad=
					C_\rho\{F_1^*(\lambda_k),F_0^*(\lambda_j)\}\\
					&\quad=
					C_\rho\{F_1(k),F_0(j)\}.
				\end{aligned}
				\]
				Boundary cases when $k=L-1$ or $j=L-1$ are easy to check using $\lambda_{L-1}=+\infty$.
				Therefore $\text{pr}\{Y(1)\leq k,Y(0)\leq j\}=C_\rho\{F_1(k),F_0(j)\}$ for all $k,j$, and Assumption~\ref{asmp: copulauncond} holds with the same copula $C_\rho$.
			\end{proof}
			
			\subsection*{Proof of Proposition~\ref{prop:eifuncond}}
			Recall the observed data $O=(X,A,Y)$ with the factorization $\text{pr}(o)=\text{pr}_X(x)\text{pr}_{A\mid X}(a\mid x)\text{pr}_{Y\mid A,X}(y\mid a,x)$.
			From the proof of Proposition~\ref{prop: eif}, the nonparametric tangent space admits the orthogonal direct-sum decomposition
			\[
			\mathcal{T}=\mathcal{T}_X\oplus\mathcal{T}_{A\mid X}\oplus\mathcal{T}_{Y\mid A,X}.
			\]
			
			For $a=0,1$ and $k=0,\ldots,L-2$, the marginal probability $F_a(k)=E\{F_a(k\mid X)\}$ is identified under Assumption~\ref{asmp: causal}. 
			Recall
			\[
			D_{ak}(O)
			=
			\frac{\mathbbm{1}(A=a)}{e(X)^a\{1-e(X)\}^{1-a}}
			\{\mathbbm{1}(Y \le k)-F_a(k \mid X)\}
			+
			F_a(k \mid X)-F_a(k).
			\]
			We verify that $D_{ak}(O)$ is the efficient influence function for $F_a(k)$ in $\mathcal{M}_{\rm np}$.
			First, $D_{ak}\in\mathcal{T}$: by definition $E\{F_a(k\mid X)-F_a(k)\}=0$, and the weighted residual term has conditional mean zero given $(A,X)$, so $E\{D_{ak}(O)\}=0$.
			
			For pathwise differentiability, consider the same regular parametric submodels $\{\text{pr}(O;\varepsilon):\varepsilon\in(-\delta,\delta)\}$ as in the proof of Proposition~\ref{prop: eif}.
			Under perturbation of $\text{pr}_X$ only, $F_a(k\mid x)$ is fixed and
			\[
			\left.\frac{\partial}{\partial\varepsilon}F_a(k;\varepsilon)\right|_{\varepsilon=0}
			=
			E\{F_a(k\mid X)\,S_X(X)\}
			=
			E\bigl[\{F_a(k\mid X)-F_a(k)\}\,S_X(X)\bigr]
			=
			E\{D_{ak}(O)\,S_X(X)\},
			\]
			using $E\{S_X(X)\}=0$ and that the inverse-probability-weighted residual term has expectation zero.
			Under perturbation of $\text{pr}_{A\mid X}$ only, both $F_a(k\mid x)$ and $F_a(k)$ are fixed, so the derivative is zero; on the other hand, $E\{D_{ak}(O)\,S_{A\mid X}(A,X)\}=0$ because both terms in $D_{ak}$ have conditional mean zero given $X$.
			Under perturbation of $\text{pr}_{Y\mid A,X}$ only, $\text{pr}_X$ and $\text{pr}_{A\mid X}$ are fixed; following the same calculation as~\eqref{eq: pathdiff3} in the proof of Proposition~\ref{prop: eif},
			\[
			\left.\frac{\partial}{\partial\varepsilon}F_a(k;\varepsilon)\right|_{\varepsilon=0}
			=
			E\!\left\{
			\frac{\mathbbm{1}(A=a)}{e(X)^a\{1-e(X)\}^{1-a}}
			\{\mathbbm{1}(Y\leq k)-F_a(k\mid X)\}
			S_{Y\mid A,X}(Y,A,X)
			\right\}
			=
			E\{D_{ak}(O)\,S_{Y\mid A,X}\}.
			\]
			By the orthogonal tangent decomposition and linearity, $F_a(k)$ is pathwise differentiable with efficient influence function $D_{ak}(O)$.
			
			Next, we consider the efficient influence function for $\psi$.
			Under Assumption~\ref{asmp: copulauncond}, $\psi$ is a smooth functional of the finite-dimensional vector $\{F_1(k),F_0(k)\}_{k=0}^{L-2}$ given in~\eqref{eq: uncond-iden-formula}, with partial derivatives $\partial\psi/\partial F_a(k)=\Delta^\psi_{ak}$ defined in~\eqref{eq: uncond-deriv}.
			For any score $S\in\mathcal{T}$, by linearity of the pathwise derivative,
			\[
			\left.\frac{\partial}{\partial\varepsilon}\psi\{\text{pr}(O;\varepsilon)\}\right|_{\varepsilon=0}
			=
			\sum_{a=0}^1\sum_{k=0}^{L-2}\Delta^\psi_{ak}
			\left.\frac{\partial}{\partial\varepsilon}F_a(k;\varepsilon)\right|_{\varepsilon=0}
			=
			\sum_{a=0}^1\sum_{k=0}^{L-2}\Delta^\psi_{ak}\,E\{D_{ak}(O)S(O)\}
			=
			E\{IF_\psi(O)S(O)\},
			\]
			where
			\[
			IF_\psi(O)
			=
			\sum_{a=0}^1\sum_{k=0}^{L-2}\Delta^\psi_{ak}\,D_{ak}(O).
			\]
			As a finite linear combination of elements of $\mathcal{T}$, $IF_\psi\in\mathcal{T}$, hence $IF_\psi$ is the nonparametric efficient influence function for $\psi$. The semiparametric efficiency bound is $E\{IF_\psi(O)^2\}$.
			
			\subsection*{Proof of Theorem~\ref{thm: uncthm}}
			
			For each $a=0,1$ and $k=0,\ldots,L-2$, define
			\[
			\Psi_{ak}(O;e,F_a)
			=
			\frac{\mathbbm{1}(A=a)}{e(X)^a\{1-e(X)\}^{1-a}}
			\{\mathbbm{1}(Y\leq k)-F_a(k\mid X)\}
			+
			F_a(k\mid X),
			\]
			so that $\hat F_a^{\rm dr}(k)=\hat E\{\Psi_{ak}(O;\hat e,\hat F_a)\}$ and $D_{ak}(O)=\Psi_{ak}(O;e,F_a)-F_a(k)$.
			
			\textit{Part (i): Double robustness.}
			Conditioning on $X$ and using $E\{\mathbbm{1}(Y\leq k)\mid A=a,X\}=F_a(k\mid X)$ under Assumption~1,
			\[
			E\{\Psi_{ak}(O;\hat e,\hat F_a)\}
			=
			E\!\left[
			\frac{e(X)^a\{1-e(X)\}^{1-a}}{\hat e(X)^a\{1-\hat e(X)\}^{1-a}}
			\{F_a(k\mid X)-\hat F_a(k\mid X)\}
			+
			\hat F_a(k\mid X)
			\right].
			\]
			Rearranging,
			\begin{equation}\label{eq:uncond-DR-bias}
				E\{\Psi_{ak}(O;\hat e,\hat F_a)\}-F_a(k)
				=
				E\!\left[
				\frac{e(X)^a\{1-e(X)\}^{1-a}-\hat e(X)^a\{1-\hat e(X)\}^{1-a}}{\hat e(X)^a\{1-\hat e(X)\}^{1-a}}
				\{\hat F_a(k\mid X)-F_a(k\mid X)\}
				\right].
			\end{equation}
			The right-hand side is a product of two estimation errors and converges to zero in probability if either $\|\hat e-e\|_2=o_p(1)$ or $\|\hat F_a(k\mid\cdot)-F_a(k\mid\cdot)\|_2=o_p(1)$.
			Combined with the law of large numbers applied to the empirical mean defining $\hat F_a^{\rm dr}(k)$, this yields $\hat F_a^{\rm dr}(k)\overset{p}{\to}F_a(k)$ for each $a,k$.
			By the continuity of the copula $C$, $\hat\psi^{\rm dr}$, viewed as a continuous functional of $\{\hat F_a^{\rm dr}(k)\}_{a,k}$, converges in probability to $\psi$.
			
			\textit{Part (ii): Asymptotic linear expansion and efficiency.}
			As $\psi$ is a smooth map for unconditional margins through copula increments and therefore twice continuously differentiable with bounded derivatives on $[c,1-c]^{2(L-1)}$.
			A first-order Taylor expansion gives
			\begin{equation}\label{eq:uncond-Taylor}
				\hat\psi^{\rm dr}-\psi
				=
				\sum_{a=0}^1\sum_{k=0}^{L-2}\Delta^\psi_{ak}\{\hat F_a^{\rm dr}(k)-F_a(k)\}+r,
				\quad
				|r|\leq C\sum_{a=0}^1\sum_{k=0}^{L-2}\{\hat F_a^{\rm dr}(k)-F_a(k)\}^2.
			\end{equation}
			
			For each $a,k$, decompose
			\[
			\hat F_a^{\rm dr}(k)-F_a(k)
			=
			\hat E\{D_{ak}(O)\}
			+
			(\hat E-E)\{\Psi_{ak}(O;\hat e,\hat F_a)-\Psi_{ak}(O;e,F_a)\}
			+
			R_{ak},
			\]
			where $R_{ak}=E\{\Psi_{ak}(O;\hat e,\hat F_a)\}-F_a(k)$. The middle term is $o_p(n^{-1/2})$ by the Donsker condition and the $L_2$-consistency of $(\hat e,\hat F_a)$, using empirical process theory. 
			By~\eqref{eq:uncond-DR-bias}, the positivity condition $c\leq\hat e(X)\leq 1-c$, applying the Cauchy--Schwarz inequality yields
			\[
			|R_{ak}|
			\leq
			c^{-1}\,\|\hat e-e\|_2\,\|\hat F_a(k\mid\cdot)-F_a(k\mid\cdot)\|_2.
			\]
			Therefore
			\[
			\hat F_a^{\rm dr}(k)-F_a(k)
			=
			\hat E\{D_{ak}(O)\}
			+
			O_p\bigl(\|\hat e-e\|_2\,\|\hat F_a(k\mid\cdot)-F_a(k\mid\cdot)\|_2\bigr)
			+
			o_p(n^{-1/2})
			=
			O_p(n^{-1/2}),
			\]
			under the consistency conditions, so the Taylor remainder in~\eqref{eq:uncond-Taylor} satisfies $|r|=O_p(n^{-1})=o_p(n^{-1/2})$. Combining and using $IF_\psi(O)=\sum_{a,k}\Delta^\psi_{ak}D_{ak}(O)$,
			\[
			\hat\psi^{\rm dr}-\psi
			=
			\hat E\{IF_\psi(O)\}
			+
			O_p(\mathrm{Rem})
			+
			o_p(n^{-1/2}),
			\quad
			\mathrm{Rem}
			=
			\|\hat e-e\|_2\sum_{a=0}^1\|\hat F_a-F_a\|_2,
			\]
			which establishes the stated rate-doubly-robust expansion.
			
			If $\|\hat e-e\|_2=o_p(n^{-1/4})$ and $\|\hat F_a-F_a\|_2=o_p(n^{-1/4})$ for $a=0,1$, then $\mathrm{Rem}=o_p(n^{-1/2})$, and
			\[
			n^{1/2}(\hat\psi^{\rm dr}-\psi)
			=
			n^{1/2}\hat E\{IF_\psi(O)\}+o_p(1).
			\]
			Since $IF_\psi(O)$ has mean zero and variance $E\{IF_\psi(O)^2\}$, the central limit theorem gives
			\[
			n^{1/2}(\hat\psi^{\rm dr}-\psi)\Rightarrow N\bigl(0,\,E\{IF_\psi(O)^2\}\bigr).
			\]
			This completes the proof.
			
			\subsection*{Proof of Proposition~\ref{prop:pseudolimit}}
			\begin{proof}
				Define
				\[
				\bar\Delta^{\psi}_{1k}(x) = \dot{C}_1\{p_{1,k}(x), p_{0,k-1}(x)\} - \dot{C}_1\{p_{1,k}(x), p_{0,k}(x)\},
				\]
				\[
				\bar\Delta^{\psi}_{0k}(x) = \dot{C}_0\{p_{1,k+1}(x), p_{0,k}(x)\} - \dot{C}_0\{p_{1,k}(x), p_{0,k}(x)\}.
				\]
				By the assumed nuisance convergence, the plug-in component $\hat{E}\{\hat{m}_\psi(X)\}$ of \eqref{eq:onestep} converges in probability to $E\{\bar{m}_\psi(X;\rho)\}$, and $\hat\Delta^\psi_{ak}(X) \to \bar\Delta^\psi_{ak}(X)$ in probability. For the augmentation term, conditioning on $X$ and then on $(A,X)$ gives, for each $a$ and $k$,
				\[
				E\left[
				\frac{\mathbbm{1}(A=a)}{e(X)^a\{1-e(X)\}^{1-a}}
				\bar\Delta^\psi_{ak}(X)
				\{\mathbbm{1}(Y\leq k) - p_{a,k}(X)\}
				\;\middle|\;
				X
				\right]
				=
				\bar\Delta^\psi_{ak}(X)
				\,
				E\{\mathbbm{1}(Y\leq k) - p_{a,k}(X) \mid A=a, X\}
				= 0,
				\]
				because $p_{a,k}(X) = E(\mathbbm{1}(Y\leq k) \mid A=a, X)$, and the inverse probability weight $\mathbbm{1}(A=a)/\{e(X)^a(1-e(X))^{1-a}\}$ integrates to $1$ when conditioned on $(A=a, X)$. Hence the population mean of each augmentation term vanishes when centered at $p_{a,k}(X)$, and the probability limit of \eqref{eq:onestep} equals $\bar\psi(\rho)$.
			\end{proof}
			
			\subsection*{Proof of Proposition~\ref{prop:endpointbounds}}
			\begin{proof}
				By Proposition~\ref{prop:binarybounds},
				\[
				F^-_{1,\Gamma}(k \mid X) \le F_1(k \mid X) \le F^+_{1,\Gamma}(k \mid X), \quad F^-_{0,\Gamma}(k \mid X) \le F_0(k \mid X) \le F^+_{0,\Gamma}(k \mid X)
				\]
				for all $k$. By \eqref{eq: deltas}, the conditional map $m_\psi(x) = \sum_{k=0}^{L-1}\{C(F_1(k \mid x), F_0(k{-}1 \mid x)) - C(F_1(k{-}1 \mid x), F_0(k{-}1 \mid x))\}$ depends on the conditional margins $\{F_1(k \mid x), F_0(k \mid x)\}_{k=0}^{L-2}$. For any copula $C$, $\dot{C}_1(u,v)$ is nondecreasing in $v$ and $\dot{C}_0(u,v)$ is nondecreasing in $u$ on $(0,1)^2$. This implies that $\Delta^\psi_{1k}(x) = \dot{C}_1\{F_1(k \mid x), F_0(k{-}1 \mid x)\} - \dot{C}_1\{F_1(k \mid x), F_0(k \mid x)\} \le 0$, so $m_\psi(x)$ is coordinatewise decreasing in $F_1(k \mid x)$; similarly, $\Delta^\psi_{0k}(x) \ge 0$, so $m_\psi(x)$ is coordinatewise increasing in $F_0(k \mid x)$. Therefore, the smallest value of $m_\psi(x)$ is obtained by setting the treated margins to their upper endpoints and the control margins to their lower endpoints:
				\[
				m_\psi(x) \ge \sum_{k=0}^{L-1} \left[C\{F^+_{1,\Gamma}(k \mid x), F^-_{0,\Gamma}(k{-}1 \mid x)\} - C\{F^+_{1,\Gamma}(k{-}1 \mid x), F^-_{0,\Gamma}(k{-}1 \mid x)\}\right] = m^-_{\psi,\Gamma}(x;\rho).
				\]
				The upper bound $m_\psi(x) \le m^+_{\psi,\Gamma}(x;\rho)$ follows by reversing the endpoint choices. Taking expectations over $X$ yields $\psi(\rho) \in [\psi^-_\Gamma(\rho), \psi^+_\Gamma(\rho)]$.
			\end{proof}
			
			\subsection*{Proof of Proposition~\ref{prop:sa_eif}}
			\begin{proof}
				The proof parallels that of Proposition~2; we provide full details for the lower endpoint $\psi^-_\Gamma(\rho)$, as the upper endpoint is identical. 
				The parameter $\psi^-_\Gamma(\rho) = E\{m^-_{\psi,\Gamma}(X;\rho)\}$ is a smooth functional of the observed-data distribution through the nuisance $\eta = (e, \{p_{1,k}\}, \{p_{0,k}\})$, where the dependence on $\eta$ is mediated by the $G$-transforms: $m^-_\psi$ depends on $\{G^+_{1,\Gamma}(e(x), p_{1,k}(x)),\, G^-_{0,\Gamma}(e(x), p_{0,k}(x))\}_{k=0}^{L-2}$.
				
				\medskip
				\noindent\textit{Step 1: Mean zero.}\quad
				We verify $E\{IF^-_{\psi,\Gamma}(O)\} = 0$. By definition, $E\{m^-_{\psi,\Gamma}(X;\rho) - \psi^-_\Gamma(\rho)\} = 0$. For the propensity score augmentation, conditioning on $X$:
				\[
				E\{H^-_\psi(X)(A - e(X)) \mid X\} = H^-_\psi(X)\,E\{A - e(X) \mid X\} = 0,
				\]
				since $e(X) = E(A \mid X)$. For the outcome augmentation, conditioning first on $(A,X)$:
				\[
				E\left[\frac{\mathbbm{1}(A=a)}{e(X)^a\{1-e(X)\}^{1-a}} W^-_{ak}(X)\{\mathbbm{1}(Y\leq k) - p_{a,k}(X)\} \;\middle|\; X\right]
				= W^-_{ak}(X)\,E\{\mathbbm{1}(Y\leq k) - p_{a,k}(X) \mid A = a, X\} = 0,
				\]
				because $p_{a,k}(X) = E(\mathbbm{1}(Y\leq k) \mid A=a, X)$. Hence $E\{IF^-_{\psi,\Gamma}(O)\} = 0$ and $IF^-_{\psi,\Gamma} \in \mathcal{T}$, the nonparametric tangent space.
				
				\medskip
				\noindent\textit{Step 2: Pathwise differentiability.}\quad
				Consider regular parametric submodels $\{\text{pr}(O;\varepsilon) = \text{pr}_X(x;\varepsilon)\,\text{pr}_{A \mid X}(a \mid x;\varepsilon)\,\text{pr}_{Y \mid A,X}(y \mid a, x;\varepsilon) : \varepsilon \in (-\delta,\delta)\}$ through the truth at $\varepsilon = 0$.
				Parallel to the proof of Proposition \ref{prop: eif}, we consider perturbations for each component.
				
				\medskip
				\textit{Perturbation of $\text{pr}_X$ only.}\quad
				Under this submodel, $e(x)$ and $p_{a,k}(x)$ are fixed, so $m^-_{\psi,\Gamma}(x;\rho)$ does not vary with $\varepsilon$. Therefore
				\[
				\frac{\partial}{\partial\varepsilon}\psi^-_\Gamma\left\{\text{pr}\left(O;\varepsilon\right)\right\}\bigg|_{\varepsilon=0}
				=
				E\left\{m^-_{\psi,\Gamma}(X;\rho)\,S_X(X)\right\}
				=
				E\left[\left\{m^-_{\psi,\Gamma}(X;\rho) - \psi^-_\Gamma(\rho)\right\}\,S_X\left(X\right)\right],
				\]
				using $E\{S_X(X)\} = 0$. On the other hand, $E\{IF^-_{\psi,\Gamma}(O)\,S_X(X)\} = E[\{m^-_{\psi,\Gamma}(X;\rho) - \psi^-_\Gamma(\rho)\}\,S_X(X)]$ because $E\{H^-_\psi(X)(A - e(X))\,S_X(X)\} = E[H^-_\psi(X)\,S_X(X)\,E\{A - e(X) \mid X\}] = 0$, and similarly the outcome augmentation terms have mean zero when multiplied by $S_X(X)$, by iterated expectation.
				
				\medskip
				\textit{Perturbation of $\text{pr}_{A \mid X}$ only.}\quad
				This is the key difference from Proposition~2: $m^-_{\psi,\Gamma}(x;\rho)$ depends on $e(x)$ through $F^\pm_{a,\Gamma}(k \mid x) = G^\pm_{a,\Gamma}\{e(x), p_{a,k}(x)\}$, and $e(x)$ varies under perturbation of $\text{pr}_{A \mid X}$. Denote $e_\varepsilon(x) = \text{pr}(A=1 \mid X=x;\varepsilon)$ and let $\dot{e}(x) = \partial e_\varepsilon(x)/\partial\varepsilon|_{\varepsilon=0}$. Since only $\text{pr}_{A \mid X}$ is perturbed while $p_{a,k}(x) = \text{pr}(Y \le k \mid A=a, X=x)$ is fixed, the chain rule gives
				\begin{align}
					\frac{\partial}{\partial\varepsilon}m^-_{\psi,\Gamma,\varepsilon}(x)\bigg|_{\varepsilon=0}
					&=
					\sum_{k=0}^{L-2}
					\Delta^{\psi,-}_{1k}(x)\,
					\frac{\partial G^+_{1,\Gamma}}{\partial e}\{e(x), p_{1,k}(x)\}\,\dot{e}(x)
					+
					\sum_{k=0}^{L-2}
					\Delta^{\psi,-}_{0k}(x)\,
					\frac{\partial G^-_{0,\Gamma}}{\partial e}\{e(x), p_{0,k}(x)\}\,\dot{e}(x)
					\notag\\
					&=
					H^-_\psi(x)\,\dot{e}(x).
					\label{eq:chain_e}
				\end{align}
				To identify $\dot{e}(x)$, note that $e_\varepsilon(x) = E_\varepsilon(A \mid X=x) = \int a\,\text{pr}_{A \mid X}(a \mid x;\varepsilon)\,da$. Differentiating with respect to $\varepsilon$ at $\varepsilon = 0$:
				\[
				\dot{e}(x)
				=
				\int a\,\text{pr}_{A \mid X}(a \mid x)\,S_{A \mid X}(a,x)\,da
				=
				E\{A\,S_{A \mid X}(A,X) \mid X = x\}.
				\]
				Since $E\{S_{A \mid X}(A,X) \mid X\} = 0$, we have $E\{e(X)S_{A \mid X}(A,X) \mid X\} = 0$, and therefore
				\[
				\dot{e}(x) = E\{(A - e(X))\,S_{A \mid X}(A,X) \mid X = x\}.
				\]
				Substituting into \eqref{eq:chain_e} and taking expectations over $X$:
				\begin{align*}
					\frac{\partial}{\partial\varepsilon}\psi^-_\Gamma\{\text{pr}(O;\varepsilon)\}\bigg|_{\varepsilon=0}
					&=
					E\left[H^-_\psi(X)\,E\{(A - e(X))\,S_{A \mid X}(A,X) \mid X\}\right]\\
					&=
					E\{H^-_\psi(X)\,(A - e(X))\,S_{A \mid X}(A,X)\}.
				\end{align*}
				On the other hand, we claim $E\{IF^-_{\psi,\Gamma}(O)\,S_{A \mid X}(A,X)\} = E\{H^-_\psi(X)(A - e(X))\,S_{A \mid X}(A,X)\}$. This follows because the remaining terms in $IF^-_{\psi,\Gamma}$ contribute zero: for the plug-in term, $E[\{m^-_{\psi,\Gamma}(X;\rho) - \psi^-_{\Gamma}(\rho)\}\,S_{A \mid X}(A,X)] = E[\{m^-_{\psi,\Gamma}(X;\rho) - \psi^-_\Gamma(\rho)\}\,E\{S_{A \mid X}(A,X) \mid X\}] = 0$; and for the outcome augmentation, by iterated conditioning,
				\begin{align*}
					&E\left[\frac{\mathbbm{1}(A=a)}{e(X)^a\{1-e(X)\}^{1-a}} W^-_{ak}(X)\{\mathbbm{1}(Y\leq k) - p_{a,k}(X)\}\,S_{A \mid X}(A,X)\right]\\
					&\quad=
					E\left[W^-_{ak}(X)\,S_{A \mid X}(a,X)\,E\{\mathbbm{1}(Y\leq k) - p_{a,k}(X) \mid A=a, X\}\right]
					= 0.
				\end{align*}
				Therefore the pathwise derivative along $\mathcal{T}_{A \mid X}$ equals $E\{IF^-_{\psi,\Gamma}(O)\,S_{A \mid X}(A,X)\}$.
				
				\medskip
				\textit{Perturbation of $\text{pr}_{Y \mid A,X}$ only.}\quad
				Under this submodel, $\text{pr}_X$ and $e(x) = \text{pr}(A = 1 \mid X=x)$ are fixed, but $p_{a,k}(x)$ varies. Denote $p_{a,k,\varepsilon}(x) = E_\varepsilon\{\mathbbm{1}(Y\leq k) \mid A=a, X=x\}$ and $\dot{p}_{a,k}(x) = \partial p_{a,k,\varepsilon}(x)/\partial\varepsilon|_{\varepsilon=0}$. By the chain rule through the $G$-transforms:
				\begin{align*}
					\frac{\partial}{\partial\varepsilon}m^-_{\psi,\Gamma,\varepsilon}(x)\bigg|_{\varepsilon=0}
					&=
					\sum_{k=0}^{L-2}
					\Delta^{\psi,-}_{1k}(x)\,
					\frac{\partial G^+_{1,\Gamma}}{\partial p}\{e(x), p_{1,k}(x)\}\,\dot{p}_{1,k}(x)\\
					&\quad+
					\sum_{k=0}^{L-2}
					\Delta^{\psi,-}_{0k}(x)\,
					\frac{\partial G^-_{0,\Gamma}}{\partial p}\{e(x), p_{0,k}(x)\}\,\dot{p}_{0,k}(x)\\
					&=
					\sum_{a=0}^{1}\sum_{k=0}^{L-2}
					W^-_{ak}(x)\,\dot{p}_{a,k}(x).
				\end{align*}
				To identify $\dot{p}_{a,k}(x)$, differentiate $p_{a,k,\varepsilon}(x) = E_\varepsilon\{\mathbbm{1}(Y\leq k) \mid A=a, X=x\}$ to obtain
				\[
				\dot{p}_{a,k}(x) = E\{(\mathbbm{1}(Y\leq k) - p_{a,k}(X))\,S_{Y \mid A,X}(Y,A,X) \mid A=a, X=x\}.
				\]
				Therefore
				\begin{align*}
					\frac{\partial}{\partial\varepsilon}\psi^-_\Gamma\bigg|_{\varepsilon=0}
					&=
					E\left[\sum_{a=0}^{1}\sum_{k=0}^{L-2}
					W^-_{ak}(X)\,
					E\{(\mathbbm{1}(Y\leq k) - p_{a,k}(X))\,S_{Y \mid A,X} \mid A=a, X\}
					\right].
				\end{align*}
				Using the identity $E\{g(Y,X) \mid A=a, X\} = E\left[\mathbbm{1}(A=a)/\{e(X)^a\{1-e(X)\}^{1-a}\}g(Y,X) \;\middle|\; X\right]$ and the law of iterated expectations:
				\begin{align*}
					\frac{\partial}{\partial\varepsilon}\psi^-_\Gamma\bigg|_{\varepsilon=0}
					&=
					E\left[\sum_{a=0}^{1}
					\frac{\mathbbm{1}(A=a)}{e(X)^a\{1-e(X)\}^{1-a}}
					\sum_{k=0}^{L-2}
					W^-_{ak}(X)\{\mathbbm{1}(Y\leq k) - p_{a,k}(X)\}\,S_{Y \mid A,X}
					\right].
				\end{align*}
				We claim this equals $E\{IF^-_{\psi,\Gamma}(O)\,S_{Y \mid A,X}(Y,A,X)\}$. The remaining terms in $IF^-_{\psi,\Gamma}$ contribute zero: $E[\{m^-_{\psi,\Gamma}(X;\rho) - \psi^-_\Gamma(\rho)\}\,S_{Y \mid A,X}] = E[\{m^-_{\psi,\Gamma}(X;\rho) - \psi^-_\Gamma(\rho)\}\,E\{S_{Y \mid A,X} \mid A,X\}] = 0$ since $E\{S_{Y \mid A,X} \mid A,X\} = 0$; and $E\{H^-_\psi(X)(A - e(X))\,S_{Y \mid A,X}\} = E[H^-_\psi(X)\,E\{(A - e(X))\,S_{Y \mid A,X} \mid X\}]$, which equals $E[H^-_\psi(X)\,E\{(A-e(X))\,E(S_{Y \mid A,X} \mid A,X) \mid X\}] = 0$.
				
				\medskip
				\noindent\textit{Step 3: Combining via the orthogonal decomposition.}\quad
				The pathwise derivative map $S \mapsto \partial\psi^-_\Gamma\{\text{pr}(O;\varepsilon)\}/\partial\varepsilon|_{\varepsilon=0}$ is linear in the score direction, and the tangent space decomposes as $\mathcal{T} = \mathcal{T}_X \oplus \mathcal{T}_{A \mid X} \oplus \mathcal{T}_{Y \mid A,X}$. Combining the results from the three submodel calculations gives, for any $S \in \mathcal{T}$,
				\[
				\frac{\partial}{\partial\varepsilon}\psi^-_\Gamma\{\text{pr}(O;\varepsilon)\}\bigg|_{\varepsilon=0}
				=
				E\{IF^-_{\psi,\Gamma}(O)\,S(O)\}.
				\]
				Therefore $\psi^-_\Gamma$ is pathwise differentiable with influence function $IF^-_{\psi,\Gamma}$. Since $IF^-_{\psi,\Gamma} \in \mathcal{T}$ by Step~1, it is the efficient influence function, and the semiparametric efficiency bound is $E[\{IF^-_{\psi,\Gamma}(O)\}^2]$. The proof for $IF^+_{\psi,\Gamma}$ is identical, with all $+/-$ superscripts interchanged throughout.
			\end{proof}
			
			\subsection*{Proof of Theorem~\ref{thm:sa_main}}
			\begin{proof}
				We prove each part for the lower endpoint; the upper endpoint is identical.
				
				\medskip
				\textit{Part~(i): Rate-doubly-robust expansion.}\quad
				Define the score functional
				\begin{align*}
					\Psi^-(O;\eta)
					&=
					m^-_\psi(X;\eta)
					+
					H^-_\psi(X;\eta)\{A - e(X)\}\\
					&\quad+
					\sum_{a=0}^{1}
					\frac{\mathbbm{1}(A=a)}{e(X)^a\{1-e(X)\}^{1-a}}
					\sum_{k=0}^{L-2}
					W^-_{ak}(X;\eta)\{\mathbbm{1}(Y\leq k) - p_{a,k}(X)\},
				\end{align*}
				so that $\hat\psi^-_{\Gamma} = \hat{E}\{\Psi^-(O;\hat\eta)\}$ and $IF^-_{\psi,\Gamma}(O) = \Psi^-(O;\eta) - \psi^-_\Gamma$. Adding and subtracting $(\hat{E} - E)\Psi^-(O;\eta)$ yields
				\begin{equation}\label{eq:three_term}
					\hat\psi^-_{\Gamma} - \psi^-_\Gamma
					=
					\hat{E}\{IF^-_{\psi,\Gamma}(O)\}
					+
					\underbrace{(\hat{E} - E)\{\Psi^-(O;\hat\eta) - \Psi^-(O;\eta)\}}_{\text{(I)}}
					+
					\underbrace{E\{\Psi^-(O;\hat\eta) - \Psi^-(O;\eta)\}}_{\text{(II): }R}.
				\end{equation}
				Term~(I) is $o_p(n^{-1/2})$ by the Donsker condition and the $L_2$-consistency of $\hat\eta$, using the empirical process theory.
				
				For term~(II), conditioning on $X$ and using $E(A \mid X) = e(X)$, $E(\mathbbm{1}(Y\leq k) \mid A=a, X) = p_{a,k}(X)$:
				\begin{align}
					R
					&=
					E\left[\hat{m}^-_{\psi,\Gamma}(X;\rho) - m^-_{\psi,\Gamma}(X;\rho)
					+
					\hat{H}^-_\psi(X)\{e(X) - \hat{e}(X)\}\right.
					\notag\\
					&\quad\left.+
					\sum_{a=0}^{1}\sum_{k=0}^{L-2}
					\frac{e(X)^a\{1-e(X)\}^{1-a}}{\hat{e}(X)^a\{1-\hat{e}(X)\}^{1-a}}\,
					\hat{W}^-_{ak}(X)\{p_{a,k}(X) - \hat{p}_{a,k}(X)\}
					\right].
					\label{eq:R_full}
				\end{align}
				We now expand $\hat{m}^-_{\psi,\Gamma}(X;\rho) - m^-_{\psi,\Gamma}(X;\rho)$. Since $m^-_{\psi,\Gamma}(x;\rho)$ is a smooth function of $\eta(x) = (e(x), \{p_{1,k}(x)\}, \{p_{0,k}(x)\})$ through the $G$-transforms, a first-order Taylor expansion around the true $\eta$ gives
				\begin{align}
					\hat{m}^-_{\psi,\Gamma}(x;\rho) - m^-_{\psi,\Gamma}(x;\rho)
					&=
					H^-_\psi(x)\{\hat{e}(x) - e(x)\}
					+
					\sum_{a=0}^{1}\sum_{k=0}^{L-2}
					W^-_{ak}(x)\{\hat{p}_{a,k}(x) - p_{a,k}(x)\}
					+
					r(x),
					\label{eq:taylor_m}
				\end{align}
				where the remainder satisfies
				\[
				|r(x)|
				\le
				C\left(|\hat{e}(x) - e(x)|^2
				+
				\sum_{a=0}^{1}\sum_{k=0}^{L-2}|\hat{p}_{a,k}(x) - p_{a,k}(x)|^2\right)
				\]
				by the uniform boundedness and twice continuously differentiability of $m^-_\psi$ in $\eta$ implied by condition (b).
				
				Substituting \eqref{eq:taylor_m} into \eqref{eq:R_full}, we examine the first-order contributions. For the propensity score terms:
				\[
				H^-_\psi(x)\{\hat{e}(x) - e(x)\} + \hat{H}^-_\psi(x)\{e(x) - \hat{e}(x)\}
				=
				\{H^-_\psi(x) - \hat{H}^-_\psi(x)\}\{\hat{e}(x) - e(x)\},
				\]
				which is a product of two estimation errors, hence second order. Further using the uniformly boundedness of $\Delta_{ak}^{\psi,-}$ and continuity of $\partial G^{\pm}_{a,\Gamma}/\partial e$, we have
				\[\left|H^-_\psi(x) - \hat{H}^-_\psi(x)\right| \leq C\left(|\hat e(x) - e(x)| + \sum_{a=0}^{1}\sum_{k=0}^{L-2}|\hat{p}_{a,k}(x) - p_{a,k}(x)|^2 + |\hat e(x) - e(x)|\sum_{a=0}^{1}\sum_{k=0}^{L-2}|\hat{p}_{a,k}(x) - p_{a,k}(x)|^2 \right),\]
				for some large enough $C$.
				For the outcome regression terms, consider a representative term for treatment arm $a$ and threshold $k$:
				\[
				W^-_{ak}(x)\{\hat{p}_{a,k}(x) - p_{a,k}(x)\}
				+
				\frac{e(x)^a\{1-e(x)\}^{1-a}}{\hat{e}(x)^a\{1-\hat{e}(x)\}^{1-a}}\,\hat{W}^-_{ak}(x)\{p_{a,k}(x) - \hat{p}_{a,k}(x)\}.
				\]
				Bounding this term is parallel to the proof of Theorem \ref{thm:rate}. Write the propensity ratio as $1 + \delta(x)$ where $|\delta(x)| \le c^{-1}\,|\hat{e}(x) - e(x)|$ by the positivity condition and the mean value theorem. The above expression equals
				\[
				\{W^-_{ak}(x) - \hat{W}^-_{ak}(x)\}\{\hat{p}_{a,k}(x) - p_{a,k}(x)\}
				-
				\delta(x)\,\hat{W}^-_{ak}(x)\{\hat{p}_{a,k}(x) - p_{a,k}(x)\},
				\]
				both of which are products of estimation errors, hence second order. Taking expectations and applying the Cauchy--Schwarz inequality, and noting that $|W^-_{ak}(x) - \hat{W}^-_{ak}(x)| \le C(|\hat{e} - e| + \sum_{b,j}|\hat{p}_{b,j} - p_{b,j}|)$ for some large enough $C$ by the smoothness of the copula and the $G$-transforms.
				Together, we obtain
				\[
				|R|
				=O\left(
				\|\hat{p}_1 - p_1\|_2^2
				+
				\|\hat{p}_0 - p_0\|_2^2
				+
				\|\hat{e} - e\|_2^2
				+
				\|\hat{e} - e\|_2\,\{\|\hat{p}_1 - p_1\|_2 + \|\hat{p}_0 - p_0\|_2\}
				\right).
				\]
				Compared with the remainder in Theorem~1, the additional $\|\hat{e} - e\|_2^2$ term arises because $m^-_\psi$ depends on $e$ through the $G$-transforms, whereas $m_\psi$ does not depend on $e$ under unconfoundedness. Combining (I) and (II) in \eqref{eq:three_term} yields the stated expansion.
				
				\medskip
				\textit{Part~(ii): Asymptotic normality and efficiency.}\quad
				Under $\|\hat{e} - e\|_2 = o_p(n^{-1/4})$ and $\max_{a=0,1}\|\hat{p}_a - p_a\|_2 = o_p(n^{-1/4})$, each term in $\mathrm{Rem}$ is $o_p(n^{-1/2})$: the squared terms satisfy $o_p(n^{-1/4})^2 = o_p(n^{-1/2})$ and the cross-product term satisfies $o_p(n^{-1/4}) \cdot o_p(n^{-1/4}) = o_p(n^{-1/2})$. The expansion from part~(i) therefore becomes
				\[
				\hat\psi^-_{\Gamma} - \psi^-_\Gamma
				=
				\hat{E}\{IF^-_{\psi,\Gamma}(O)\}
				+
				o_p(n^{-1/2}).
				\]
				Since $IF^-_{\psi,\Gamma}(O)$ has mean zero by Proposition~\ref{prop:sa_eif} and finite variance $\sigma^2_- = E[\{IF^-_{\psi,\Gamma}(O)\}^2]$, the central limit theorem gives $n^{1/2}(\hat\psi^-_{\Gamma} - \psi^-_\Gamma) \Rightarrow \mathcal{N}(0, \sigma^2_-)$. By Proposition~\ref{prop:sa_eif}, $\sigma^2_-$ is the semiparametric efficiency bound for $\psi^-_\Gamma(\rho)$ in $\mathcal{M}_{\mathrm{np}}$.
				
				\medskip
				\textit{Part~(iii): Exact nesting.}\quad
				When $\Gamma = 1$, $r^-_\Gamma(p) = r^+_\Gamma(p) = p$ for all $p \in [0,1]$. Substituting into the $G$-transforms:
				\[
				G^\pm_{1,\Gamma}(e,p) = ep + (1-e)p = p, \quad G^\pm_{0,\Gamma}(e,p) = (1-e)p + ep = p.
				\]
				Therefore $F^\pm_{a,\Gamma}(k \mid x) = p_{a,k}(x)$ for all $a$ and $k$. Under Assumption~1, $p_{a,k}(x) = F_a(k \mid x)$, so $m^\pm_\psi(x;1,\rho) = m_\psi(x)$.
				
				For the augmentation coefficients, the $G$-transform derivatives at $\Gamma = 1$ are
				\[
				\frac{\partial G^\pm_{1,\Gamma}}{\partial e}(e,p) = p - p = 0,
				\quad
				\frac{\partial G^\pm_{0,\Gamma}}{\partial e}(e,p) = p - p = 0,
				\]
				\[
				\frac{\partial G^\pm_{1,\Gamma}}{\partial p}(e,p) = e + (1-e) \cdot 1 = 1,
				\quad
				\frac{\partial G^\pm_{0,\Gamma}}{\partial p}(e,p) = (1-e) + e \cdot 1 = 1,
				\]
				where $\dot{r}^\pm_\Gamma(p)|_{\Gamma=1} = 1$. Therefore $H^\pm_\psi(x) = 0$ and $W^\pm_{ak}(x) = \Delta^{\psi,\pm}_{ak}(x) \cdot 1 = \Delta^\psi_{ak}(x)$, since $\Delta^{\psi,\pm}_{ak}(x)$ reduces to $\Delta^\psi_{ak}(x)$ in \eqref{eq: deltas} when the endpoint margins equal $p_{a,k}(x) = F_a(k \mid x)$. Substituting into \eqref{eq: onestep-bounds}:
				\begin{align*}
					\hat\psi^\pm_{1}(\rho)
					&=
					\hat{E}\left\{
					\hat{m}_\psi(X)
					+
					\sum_{a=0}^{1}
					\frac{\mathbbm{1}(A=a)}{\hat{e}(X)^a\{1-\hat{e}(X)\}^{1-a}}
					\sum_{k=0}^{L-2}
					\hat\Delta^\psi_{ak}(X)\{\mathbbm{1}(Y \le k) - \hat{F}_a(k \mid X)\}
					\right\}
					=
					\hat\psi(\rho),
				\end{align*}
				which is exactly the one-step estimator in \eqref{eq:onestep}.
			\end{proof}
			
			\section{Additional simulations}\label{sec: supp-addsimu}
			\subsection*{Working-model misspecification and moderate overlap}
			
			We extend the simulation of Section~\ref{sec: simu} by assessing the conditional copula estimators under nuisance-model misspecification and weaker treatment overlap. The conditional copula model is correctly specified throughout, with true Gumbel copula at $\rho=2$ (Kendall's $\tau=0.5$). We generate $X=(X_1,X_2,X_3)^{\T}$ with $X_j\stackrel{\text{iid}}{\sim}\text{Unif}(-1,1)$ and define $\tilde X_j=(X_j+0.5)^2/2$. The baseline propensity score and outcome models are
			\[
			\begin{aligned}
				\text{logit}\{e(X)\}&=0.5+\sigma(-0.2X_1+0.2X_2-0.2X_3),\\
				\text{logit}\{\text{pr}(Y(a)\le k\mid X)\}&=\lambda_k-\eta_a(X),\quad \eta_a(X)=0.6+0.15(X_1+X_2+X_3)+0.4a,
			\end{aligned}
			\]
			with $\lambda_k=\text{logit}\{(k+1)/5\}$, $k=0,\ldots,3$, where $\sigma$ controls the overlap. The misspecified propensity score replaces the linear predictor by $1.2(-\tilde X_1+\tilde X_2-\tilde X_3)$, and the misspecified outcome model is multinomial logistic
			\[
			\text{pr}\{Y(a)=j\mid X\}=\exp\{s_j(a,X)\}\Big/\sum_{\ell=0}^{4}\exp\{s_\ell(a,X)\},\quad j=0,\ldots,4,
			\]
			with $s_j(a,X)=b_{j0}+b_{j1}a+2.5(b_{j2}\tilde X_1+b_{j3}\tilde X_2+b_{j4}\tilde X_3)$ and coefficients $b_0=(0,0,0,0,0)$, $b_1=(0.9,1.2,1.0,-0.9,0.8)$, $b_2=(1.3,0.2,-1.2,1.1,-0.7)$, $b_3=(0.5,-0.8,0.8,-1.0,1.0)$, $b_4=(-0.7,-1.1,0.4,1.0,-1.2)$. Under the parametric logistic and proportional-odds working models linear in $X$, the misspecified models depart smoothly but nontrivially from the working specification. Three scenarios are considered:
			\begin{itemize}
				\item[] \textit{Scenario FT}: misspecified propensity score, correctly specified outcome model, $\sigma=1$;
				\item[] \textit{Scenario FF}: both nuisance models misspecified, $\sigma=1$;
				\item[] \textit{Scenario TT} (moderate overlap): both nuisance models correctly specified, $\sigma=3$ to weaken overlap.
			\end{itemize}
			We consider the same four estimators as in Section~\ref{sec: simu}: $\hat\psi_{\text{Par}}$, $\hat\psi_{\text{ML}}$, $\hat\psi_{\text{G}}$, and $\hat\psi_{\text{Gb3}}$. Sample sizes are $n\in\{200,1000\}$ with $500$ replications.
			
			\begin{table}[htbp]
				\centering
				\caption{Simulation results under nuisance-model misspecification and moderate overlap. Bias, SD, and RMSE are multiplied by $10^3$; Cov (\%) is empirical coverage of 95\% confidence intervals; SBC (\%) is the frequency that the 95\% confidence interval lies within the sharp bounds.}
				\label{tab:task3_misspec_overlap}
				\setlength{\tabcolsep}{3.5pt}
				\begin{tabular}{lrrrrrrrrrr}
					\hline
					& \multicolumn{5}{c}{$n=200$} & \multicolumn{5}{c}{$n=1000$} \\
					\cline{2-6} \cline{7-11}
					Method & Bias & SD & RMSE & Cov (\%) & SBC (\%) & Bias & SD & RMSE & Cov (\%) & SBC (\%) \\
					\hline
					\multicolumn{11}{l}{\textit{Scenario FT}} \\
					Par & 8.4 & 66.1 & 66.6 & 95.0 & 99.4 & $-$0.2 & 30.0 & 30.0 & 96.0 & 100.0 \\
					ML & 0.1 & 69.5 & 69.5 & 88.6 & 99.4 & $-$4.9 & 31.8 & 32.2 & 90.6 & 100.0 \\
					PG & 9.4 & 66.2 & 66.8 & 95.2 & 99.4 & 1.4 & 29.9 & 29.9 & 95.6 & 100.0 \\
					PGb3 & $-$13.4 & 82.2 & 83.2 & 92.8 & 89.0 & $-$26.7 & 39.0 & 47.2 & 87.4 & 100.0 \\
					\hline
					\multicolumn{11}{l}{\textit{Scenario FF}} \\
					Par & 19.6 & 36.2 & 41.1 & 96.4 & 25.8 & 16.9 & 15.3 & 22.8 & 88.4 & 99.8 \\
					ML & 2.5 & 33.3 & 33.4 & 95.6 & 11.0 & $-$1.6 & 13.9 & 14.0 & 96.4 & 92.8 \\
					PG & 23.9 & 39.1 & 45.8 & 94.2 & 23.8 & 20.8 & 17.1 & 26.9 & 81.6 & 99.6 \\
					PGb3 & $-$10.9 & 34.6 & 36.2 & 88.8 & 2.2 & $-$15.2 & 14.3 & 20.9 & 80.8 & 40.2 \\
					\hline
					\multicolumn{11}{l}{\textit{Scenario TT (moderate overlap)}} \\
					Par & 9.4 & 70.3 & 70.9 & 92.4 & 98.2 & $-$1.1 & 32.2 & 32.1 & 93.8 & 100.0 \\
					ML & 1.3 & 73.5 & 73.4 & 89.0 & 98.2 & $-$5.9 & 33.8 & 34.2 & 92.8 & 100.0 \\
					PG & 10.0 & 70.3 & 71.0 & 92.2 & 98.0 & 0.3 & 32.1 & 32.1 & 93.6 & 100.0 \\
					PGb3 & $-$12.3 & 86.3 & 87.1 & 92.6 & 85.2 & $-$27.7 & 41.6 & 49.9 & 88.2 & 100.0 \\
					\hline
				\end{tabular}
			\end{table}
			
			Table~\ref{tab:task3_misspec_overlap} summarizes the results. In Scenario FT, $\hat\psi_{\text{Par}}$ is nearly unbiased at $n=1000$ with coverage close to nominal, consistent with the rate-expansion results in Theorem \ref{thm:rate}.
			In Scenario FF, the parametric one-step estimators exhibit moderate bias at $n=1000$ with under-cover confidence intervals, while the cross-fitted estimator $\hat\psi_{\text{ML}}$ remains nearly unbiased with near-nominal coverage at $n=1000$ across all scenarios, though it might have larger finite-sample bias. 
			For copula sensitivity, $\hat\psi_{\text{G}}$ with matched $\tau$ tracks $\hat\psi_{\text{Par}}$ closely in FT and TT, and the two move together in FF, illustrating certain robustness to copula family misspecification. In contrast, $\hat\psi_{\text{Gb3}}$ exhibits persistent bias and degraded coverage throughout, reinforcing the importance of calibrating the dependence level via $\tau$. In Scenario TT with moderate overlap, all methods show slightly larger bias and lower coverage.
			
			\subsection*{Covariate-dependent conditional dependence}
			
			We assess the robustness of the constant-$\rho$ conditional copula analysis when the true
			conditional dependence varies with covariates.  Relative to the simulation design in Section \ref{sec: simu}, the
			propensity score model, proportional-odds outcome model, thresholds, and support of~$X$
			are all kept fixed; only the conditional Gumbel dependence parameter is modified.
			Specifically, the latent errors are
			generated from a conditional Gumbel copula with covariate-dependent Kendall 's~$\tau$:
			\[
			\tau(X)=\mathrm{expit}(sX_1),
			\quad
			\rho(X)=\frac{1}{1-\tau(X)}.
			\]
			Three scenarios are considered: Scenario~1 ($s=0$) with  correctly specified constant
			dependence with $\tau(X)\equiv 0.5$; Scenario~2 ($s=0.8$) with moderate heterogeneity; and Scenario~3 ($s=1.6$) with strong heterogeneity.  
			By symmetry of~$X_1$,
			all three designs satisfy $E\{\tau(X)\}\approx 0.5$.
			
			Let $\psi_0(s)=E[m_\psi\{X;\tau(X)\}]$ denote the true target under the heterogeneous
			conditional copula model.  
			We compare two estimators, both using the correctly specified
			parametric nuisance models linear in $(X_1,X_2,X_3)$: \emph{Oracle heterogeneous}, which
			uses the true $\tau(X)$ inside the Gumbel copula map, and \emph{Constant~$\tau=0.5$},
			which uses the constant conditional Gumbel analysis with Kendall's~$\tau$ fixed at~$0.5$.
			
			\begin{table}[t]
				\centering
				\caption{Simulation results under covariate-dependent conditional dependence.
					Bias, SD, and RMSE are multiplied by~$10^3$; Cov~(\%) is the empirical coverage of
					95\% confidence intervals.}
				\label{tab:conddep}
				\setlength{\tabcolsep}{3.5pt}
				\begin{tabular}{lrrrrrrrr}
					\hline
					& \multicolumn{4}{c}{$n=200$} & \multicolumn{4}{c}{$n=1000$} \\
					\cline{2-5} \cline{6-9}
					Method & Bias & SD & RMSE & Cov (\%) & Bias & SD & RMSE & Cov (\%) \\
					\hline
					\multicolumn{9}{l}{\textit{Scenario 1: $s=0$ (constant dependence)}} \\
					Oracle        &  7.8 & 66.9 & 67.3 & 94.6 & $-$1.4 & 29.8 & 29.8 & 96.2 \\
					Constant      &  7.8 & 66.9 & 67.3 & 94.6 & $-$1.4 & 29.8 & 29.8 & 96.2 \\
					\hline
					\multicolumn{9}{l}{\textit{Scenario 2: $s=0.8$ (moderate heterogeneity)}} \\
					Oracle        &  4.0 & 67.0 & 67.1 & 93.2 &    1.2 & 31.0 & 31.0 & 95.2 \\
					Constant      &  3.7 & 66.2 & 66.2 & 94.2 &    1.8 & 30.2 & 30.2 & 95.6 \\
					\hline
					\multicolumn{9}{l}{\textit{Scenario 3: $s=1.6$ (strong heterogeneity)}} \\
					Oracle        & 10.0 & 72.4 & 73.0 & 93.4 & $-$3.3 & 35.0 & 35.1 & 94.0 \\
					Constant      &  9.0 & 66.9 & 67.4 & 94.4 & $-$2.0 & 30.3 & 30.3 & 94.4 \\
					\hline
				\end{tabular}
			\end{table}
			
			The true targets are $\psi_0(0)=0.370$, $\psi_0(0.8)=0.369$, and $\psi_0(1.6)=0.369$,
			nearly identical across scenarios.  The constant-$\tau$ analysis is remarkably stable: the structural bias of the constant-$\tau=0.5$ target relative to the true
			heterogeneous target is $0.0\times 10^{-3}$, $1.1\times 10^{-3}$, and
			$0.5\times 10^{-3}$ for $s=0$, $0.8$, and $1.6$, respectively.
			Table~\ref{tab:conddep} shows that this structural robustness carries over to
			finite-sample estimation.  In Scenario~1, the oracle and constant-$\tau$ rows coincide by
			construction.  In Scenarios~2 and~3, the two methods remain numerically close: at
			$n=1000$, their biases differ by at most $1.3\times 10^{-3}$, and both achieve coverage
			near the nominal 95\% level.  The stronger heterogeneity increases variability modestly
			for the oracle estimator but does not produce a substantial gap between the two analyses.
			
			\subsection*{Unconditional copula model}
			We assess the unconditional copula estimator in~\eqref{eq: uncdrest} under a
			data-generating process in which Assumption~\ref{asmp: copulauncond} holds exactly at
			the population level.
			
			Let $p_{ak}(x) = \mathrm{pr}\{Y(a) = k \mid X = x\}$ and
			$F_a(k) = \mathrm{pr}\{Y(a)\leq k\} = E\{F_a(k \mid X)\}$ denote the unconditional
			margins.  We impose the target unconditional Gumbel copula on $(F_1, F_0)$ to obtain the
			joint potential-outcome probability matrix $\Pi = \{\pi_{kj}\}$, define the correction
			$\Delta = \Pi - E\{p_1(X)\,p_0(X)^\top\}$, and set the conditional joint law to
			$\pi_{kj}(x) = p_{1k}(x)\,p_{0j}(x) + \Delta_{kj}$.  By construction,
			$E\{\pi_{kj}(X)\} = \pi_{kj}$, so the unconditional copula model holds exactly while the
			conditional dependence structure varies with~$X$.
			
			We generate $X = (X_1,X_2,X_3)^\T$ with
			$X_j \stackrel{\mathrm{iid}}{\sim} \mathrm{Unif}(-1,1)$, fix $L = 5$, and use a true
			unconditional Gumbel copula with $\rho = 2$ (Kendall's $\tau = 0.5$). The baseline
			propensity score and outcome models are
			\[
			\mathrm{logit}\{e(X)\} = 0.5 - 0.2X_1 + 0.2X_2 - 0.2X_3,
			\]
			\[
			\mathrm{logit}\{\mathrm{pr}(Y(a) \le k \mid X)\} = \lambda_k - \eta_a(X), \quad
			\eta_a(X) = 0.2 + 0.13(X_1 + X_2 + X_3) + 0.4a,
			\]
			where $\lambda_k = \mathrm{logit}\{(k+1)/5\}$ for $k = 0,\ldots,3$.  To induce
			misspecification, define $Q_j = X_j^2/2$ and set the misspecified truth models to
			\[
			\mathrm{logit}\{e(X)\} = 0.8(-Q_1 + Q_2 - Q_3), \quad
			\eta_a(X) = 0.30(-Q_1 + Q_2 - Q_3) + 0.4a.
			\]
			The parametric working models are logistic regression for the propensity score and
			proportional-odds regression, both linear in $(X_1,X_2,X_3)$.  Four scenarios are formed
			by crossing correctly specified~(T) and misspecified~(F) truth models for the propensity
			score and outcome: TT, FT, TF, and~FF.

			We consider four estimators, each the unconditional analogue of those in
			Section~\ref{sec: simu}: (i)~$\hat\psi_{\mathrm{Par}}$, the doubly robust estimator
			with correctly specified Gumbel copula and parametric nuisance models;
			(ii)~$\hat\psi_{\mathrm{ML}}$, the five-fold cross-fitted doubly robust estimator with
			random forest nuisance estimates; (iii)~$\hat\psi_{\mathrm{G}}$, the doubly robust
			estimator with a misspecified Gaussian copula calibrated to match Kendall's~$\tau$;
			and (iv)~$\hat\psi_{\mathrm{Gb3}}$, the doubly robust estimator with a misspecified
			Gumbel parameter $\rho = 3$.  Sample sizes are $n \in \{200, 1000\}$ with $500$
			replications.
			
			\begin{table}[tbp]
				\centering
				\caption{Simulation results for the unconditional copula model. Bias, SD, and RMSE are multiplied by $10^3$; Cov (\%) is empirical coverage of 95\% confidence intervals; SBC (\%) is the frequency that the 95\% confidence interval lies within the sharp bounds.}
				\label{tab:task5_unc_exact}
				\setlength{\tabcolsep}{3.5pt}
				\begin{tabular}{lrrrrrrrrrr}
					\hline
					& \multicolumn{5}{c}{$n=200$} & \multicolumn{5}{c}{$n=1000$} \\
					\cline{2-6} \cline{7-11}
					Method & Bias & SD & RMSE & Cov (\%) & SBC (\%) & Bias & SD & RMSE & Cov (\%) & SBC (\%) \\
					\hline
					\multicolumn{11}{l}{\textit{Scenario TT}} \\
					Par & 0.3 & 69.4 & 69.3 & 93.6 & 100.0 & $-$0.6 & 31.7 & 31.7 & 95.4 & 100.0 \\
					ML & 5.6 & 123.5 & 123.5 & 94.8 & 50.2 & 2.2 & 42.7 & 42.8 & 96.4 & 100.0 \\
					PG & 2.5 & 67.8 & 67.8 & 93.0 & 100.0 & 2.1 & 30.8 & 30.8 & 95.6 & 100.0 \\
					PGb3 & $-$25.5 & 90.0 & 93.5 & 90.6 & 85.4 & $-$27.7 & 42.3 & 50.6 & 88.4 & 100.0 \\
					\hline
					\multicolumn{11}{l}{\textit{Scenario FT}} \\
					Par & 0.2 & 66.0 & 65.9 & 93.2 & 100.0 & 0.3 & 28.9 & 28.9 & 95.8 & 100.0 \\
					ML & 5.3 & 104.7 & 104.8 & 94.2 & 74.2 & 4.8 & 36.9 & 37.1 & 96.4 & 100.0 \\
					PG & 2.7 & 64.5 & 64.5 & 93.4 & 100.0 & 3.1 & 28.1 & 28.2 & 95.6 & 100.0 \\
					PGb3 & $-$25.5 & 86.1 & 89.7 & 89.8 & 89.2 & $-$26.5 & 38.7 & 46.9 & 89.6 & 100.0 \\
					\hline
					\multicolumn{11}{l}{\textit{Scenario TF}} \\
					Par & 4.0 & 73.2 & 73.2 & 92.6 & 99.8 & 1.4 & 33.0 & 33.0 & 95.2 & 100.0 \\
					ML & 16.2 & 118.9 & 119.9 & 96.0 & 61.8 & 3.5 & 44.8 & 44.9 & 95.8 & 100.0 \\
					PG & 6.2 & 70.4 & 70.6 & 92.8 & 100.0 & 4.1 & 31.5 & 31.7 & 94.4 & 100.0 \\
					PGb3 & $-$20.5 & 95.4 & 97.5 & 90.8 & 86.4 & $-$25.1 & 44.4 & 51.0 & 88.8 & 100.0 \\
					\hline
					\multicolumn{11}{l}{\textit{Scenario FF}} \\
					Par & 7.9 & 70.6 & 71.0 & 93.4 & 100.0 & 5.1 & 30.7 & 31.1 & 93.8 & 100.0 \\
					ML & 15.0 & 110.0 & 110.9 & 92.6 & 75.6 & 3.7 & 39.2 & 39.3 & 96.0 & 100.0 \\
					PG & 10.3 & 67.8 & 68.5 & 93.6 & 100.0 & 7.8 & 29.4 & 30.4 & 93.6 & 100.0 \\
					PGb3 & $-$15.8 & 92.9 & 94.2 & 91.2 & 91.2 & $-$20.2 & 41.4 & 46.0 & 91.6 & 100.0 \\
					\hline
				\end{tabular}
			\end{table}
			
			Table~\ref{tab:task5_unc_exact} summarises the results under a fixed copula.  The
			parametric estimator $\hat\psi_{\mathrm{Par}}$ has little bias and achieves nominal
			coverage in scenarios TT, FT, and TF, where at least one nuisance model is correctly
			specified, confirming double robustness.  In scenario~FF, a small but nonvanishing bias
			emerges, consistent with the remainder structure in Theorem~\ref{thm: uncthm}.  The
			cross-fitted estimator $\hat\psi_{\mathrm{ML}}$ is approximately unbiased with
			satisfactory coverage at $n = 1000$ but exhibits greater variability at $n = 200$,
			reflecting the finite-sample cost of nonparametric nuisance estimation.  At $n = 1000$,
			the confidence intervals lie within the sharp bounds across all scenarios.  For copula
			sensitivity, $\hat\psi_{\mathrm{G}}$ with matched $\tau$ incurs only mild additional
			bias, whereas $\hat\psi_{\mathrm{Gb3}}$ produces substantial bias and degraded coverage,
			reinforcing the importance of calibrating the dependence level through~$\tau$.  These
			patterns closely parallel the conditional copula results in Section \ref{sec: simu}.
			
			For the sensitivity curves across the full $(L,\delta)$ grid, the single DGP
			configuration $(L,\delta) = (5,0.4)$ used above cannot be applied uniformly.  
			We therefore adopt a slightly different outcome model,
			$\eta_a(X) = 0.2 + 0.05(X_1+X_2+X_3) + \delta a$, while retaining the same propensity
			score model and Gumbel dependence as in the TT scenario.
			
			\begin{figure}[tbp]
				\centering
				\includegraphics[width=0.85\textwidth]{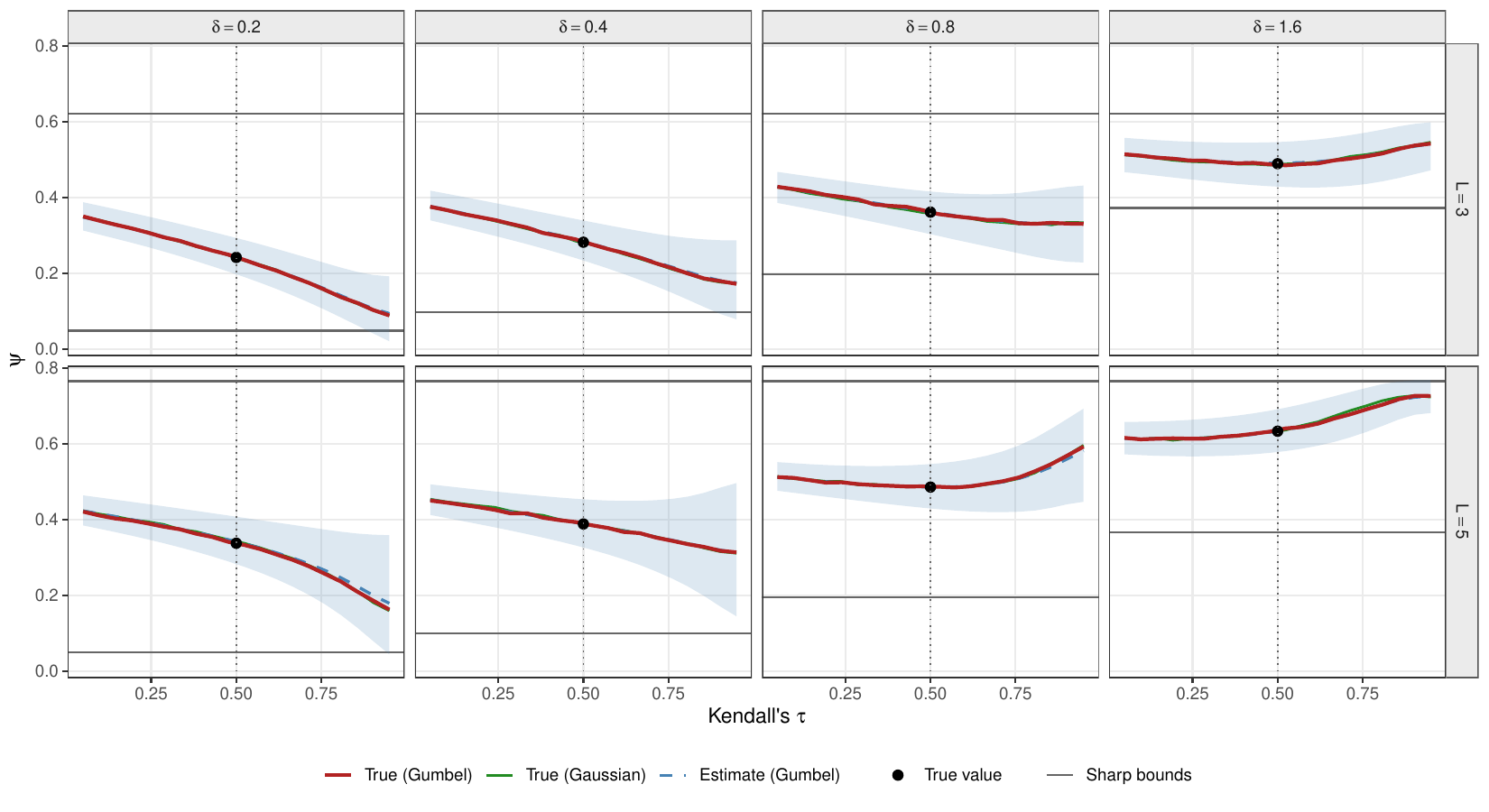}
				\caption{Sensitivity curves $\psi(\tau)$ for the unconditional copula model, with rows indexed by $L \in \{3,5\}$ and columns by $\delta \in \{0.2, 0.4, 0.8, 1.6\}$. Solid red and green curves are population truth under Gumbel and Gaussian copulas; dashed blue curve with shaded band is the doubly robust estimator and Monte Carlo 95\% range over 500 replications at $n = 1000$; solid grey lines are sharp bounds; black dot marks the truth at $\tau = 0.5$.}
				\label{fig:task5_unc_panel}
			\end{figure}
			
			Figure~\ref{fig:task5_unc_panel} displays the resulting sensitivity curves.  The
			estimated unconditional Gumbel curve closely tracks the population truth under both the
			Gumbel and Gaussian copulas across the full grid, and all confidence bands lie within
			the sharp bounds.  These patterns are qualitatively consistent with the conditional
			copula panel in the main text, indicating that both identification strategies yield
			similar inferential conclusions under correct model specification.
			
			\subsection*{Conditional versus unconditional copula sensitivity analysis}
			
			We compare the conditional and unconditional copula analyses to assess when the resulting sensitivity curves are practically indistinguishable and when they diverge. Relative to the main-text simulation design, we retain the same covariate distribution, propensity score model, treatment effect size, and ordinal threshold specification, and vary only the strength of the prognostic outcome effect and the dependence structure across three scenarios:
			\begin{itemize}
				\item[] \textit{Scenario~1}: the baseline conditional Gumbel model with constant Kendall's $\tau=0.5$, same as the Section \ref{sec: simu};
				\item[] \textit{Scenario~2}: a stronger prognostic effect with $\eta_a(X)=0.6+0.45(X_1+X_2+X_3)+0.4a$ and the same constant conditional Gumbel copula;
				\item[] \textit{Scenario~3}: the same outcome model as Scenario~2 but with heterogeneous conditional Gumbel dependence, $\tau(X)=0.5+0.30(|X_1|+|X_2|+|X_3|-1.5)$ and $\rho(X)=\{1-\tau(X)\}^{-1}$.
			\end{itemize}
			In all three scenarios, $(U_1,U_0) \mid X$ are coupled by the specified Gumbel copula through the same latent-variable construction as in Section \ref{sec: simu}. Sample sizes $n = 200,1000$ and number of replications is $500$
			
			For each scenario, we obtained two Gumbel sensitivity curves $\hat\psi(\tau)$ over $\tau \in \{0,0.1,\ldots,0.9\}$: the one-step estimator $\hat\psi^c(\tau)$ based on the conditional copula model and the doubly robust estimator $\hat\psi^{\mathrm{unc}}(\tau)$ based on the unconditional copula model.
			Both estimators use the same correctly specified parametric nuisance models: logistic regression for the propensity score and proportional-odds regression linear in $X$ for the outcome. We use $n \in \{200, 1000\}$ with $500$ replications.
			
			To summarize agreement, we report for each replication the mean absolute curve discrepancy $D_{\mathrm{mean}} = 10^{-1}\sum_{j=1}^{10}|\hat\psi^{c}(\tau_j) - \hat\psi^{\mathrm{unc}}(\tau_j)|$, the maximum discrepancy $D_{\max} = \max_{j}|\hat\psi^{c}(\tau_j) - \hat\psi^{\mathrm{unc}}(\tau_j)|$, and the pointwise discrepancy at $\tau = 0.5$, denoted $D_{0.5}$.
			
			\begin{table}[tbp]
				\centering
				\caption{Agreement between conditional and unconditional Gumbel sensitivity curves. Entries are Monte Carlo mean (standard deviation) of $D_{\mathrm{mean}}$, $D_{\max}$, and $D_{0.5}$, each multiplied by $10^3$, over $500$ replications.}
				\label{tab:cond_uncond_curve}
				\begin{tabular}{llccc}
					\hline
					& & \multicolumn{3}{c}{Curve discrepancy ($\times 10^3$)} \\
					\cline{3-5}
					Scenario & $n$ & $D_{\mathrm{mean}}$ & $D_{\max}$ & $D_{0.5}$ \\
					\hline
					1: Baseline     & $200$  & $10.5$ ($8.1$) & $37.3$ ($30.2$) & $6.5$ ($6.3$) \\
					& $1000$ & $2.1$ ($1.8$)  & $7.6$ ($8.1$)   & $1.2$ ($1.2$) \\[4pt]
					2: Strong prognostic & $200$  & $12.2$ ($8.7$) & $41.7$ ($34.1$) & $6.9$ ($6.5$) \\
					& $1000$ & $3.8$ ($1.8$)  & $8.8$ ($6.0$)   & $2.4$ ($1.7$) \\[4pt]
					3: Heterogeneous $\tau$ & $200$  & $11.5$ ($8.4$) & $37.7$ ($32.7$) & $6.9$ ($6.6$) \\
					& $1000$ & $3.9$ ($1.9$)  & $9.4$ ($7.0$)   & $2.4$ ($1.7$) \\
					\hline
				\end{tabular}
			\end{table}
			
			\begin{figure}[tbp]
				\centering
				\includegraphics[width=0.85\textwidth]{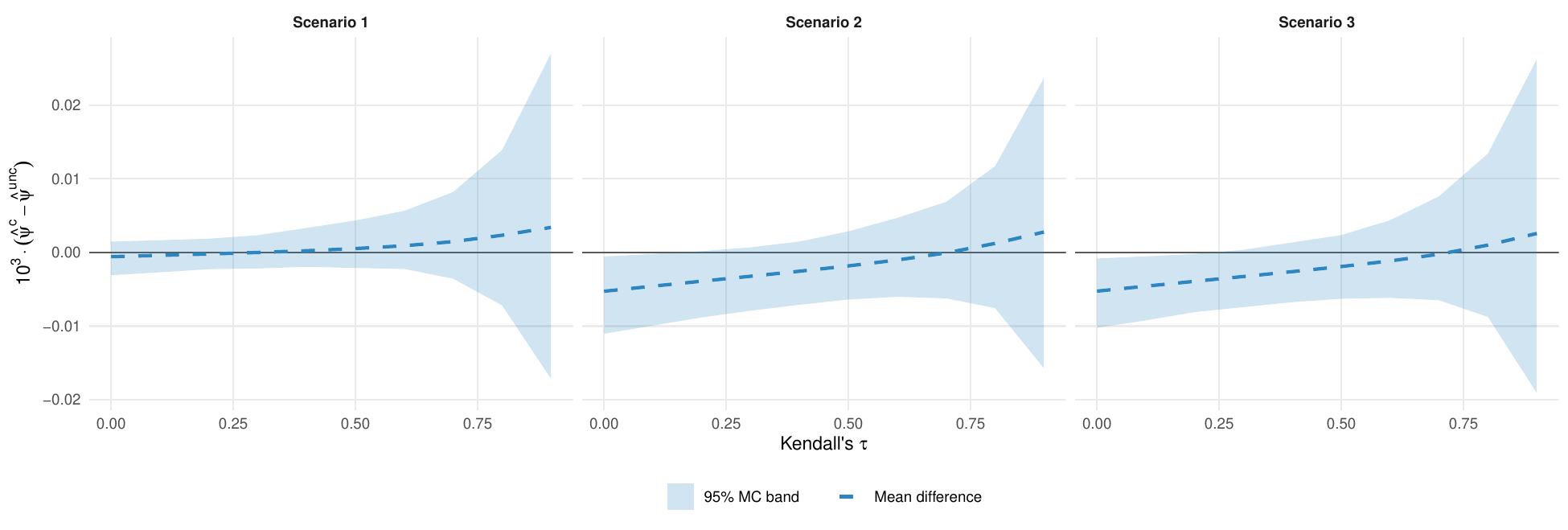}
				\caption{Estimated Gumbel sensitivity curves differences under three scenarios. The dashed blue line is the mean differences, and the blue area is the Monte Carlo 95\% range over 500 replications at $n = 1000$.}
				\label{fig:cond_uncond_curve}
			\end{figure}
			
			Table~\ref{tab:cond_uncond_curve} and Figure~\ref{fig:cond_uncond_curve} present the results. In Scenario~1, where the covariate effect on the outcome margins is weakest, the two curves are nearly indistinguishable: at $n = 1000$, $D_{\mathrm{mean}}$ is $2.1 \times 10^{-3}$ and $D_{0.5}$ is $1.2 \times 10^{-3}$. The discrepancy increases modestly in Scenarios~2 and~3, with $D_{\mathrm{mean}}$ rising to approximately $3.8$--$3.9 \times 10^{-3}$ and $D_{0.5}$ to $2.4 \times 10^{-3}$, driven by the stronger prognostic structure; introducing heterogeneous dependence in Scenario~3 adds little additional divergence beyond the prognostic effect.  
			These results indicate that the unconditional copula analysis closely approximates the conditional one, supporting its use as a robustness check for the estimated sensitivity curves.
			
			\section{Additional results of the application}\label{sec:supp-addapp}
			We provide the unconditional copula sensitivity curves here to serve as an additional robustness check.
			\begin{figure}[tbp]
				\centering
				\includegraphics[width=.9\textwidth]{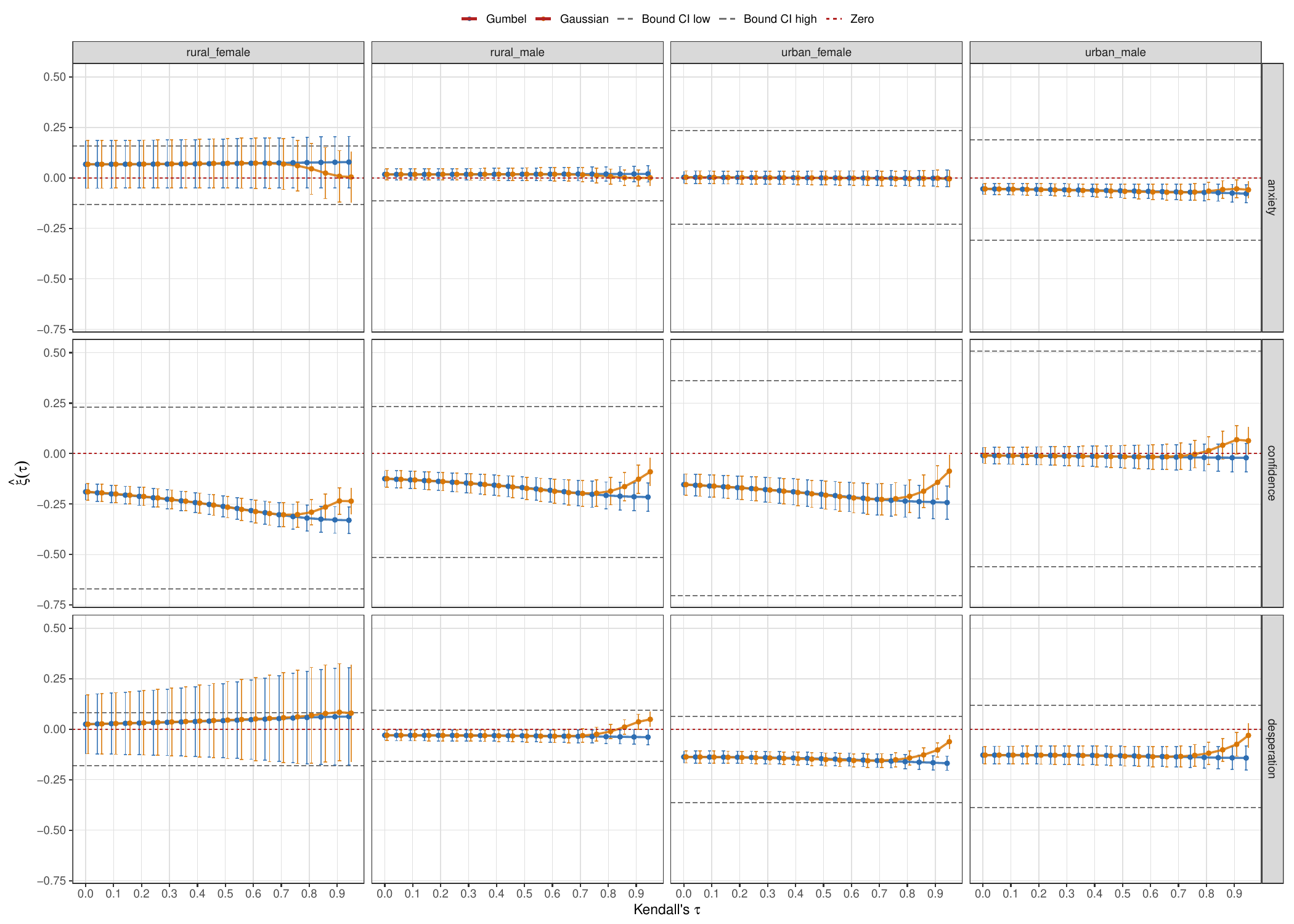}
				\caption{Unconditional copula-based subgroup sensitivity curves $\hat\xi(\tau)$ under Gaussian and Gumbel copulas, with pointwise 95\% confidence intervals.  The horizontal dashed red
					line marks $\xi=0$. Two dashed gray lines mark the bootstrapped confidence interval based on the sharp bounds. A curve lying below zero indicates that only-child status is
					harmful for the corresponding subgroup and outcome.}
				\label{fig:unc-application-xi}
			\end{figure}
			
			Figure~\ref{fig:unc-application-xi} displays the unconditional copula subgroup
			sensitivity curves for~$\xi$ under the Gaussian and Gumbel copula families,
			paralleling the conditional copula analysis in Figure~\ref{fig:application-xi} of the main text.
			The two analyses yield consistent conclusions regarding the direction of causal
			effects across all subgroup--outcome combinations: only-child status is
			associated with significantly lower anxiety among urban males, lower
			confidence among rural females, rural males, and urban females, and greater
			desperation among both urban subgroups.
			The point estimates and confidence
			bands are closely aligned between the conditional and unconditional
			approaches, with no case in which the two models disagree on the sign or
			significance of~$\hat\xi(\tau)$.  
			This demonstrates certain 
			robustness of the Section \ref{sec:application} findings.

			\putbib
		\end{bibunit}
		

\begin{thebibliography}{22}
\providecommand{\natexlab}[1]{#1}
\providecommand{\url}[1]{\texttt{#1}}
\expandafter\ifx\csname urlstyle\endcsname\relax
  \providecommand{\doi}[1]{doi: #1}\else
  \providecommand{\doi}{doi: \begingroup \urlstyle{rm}\Url}\fi

\bibitem[Bartolucci and Grilli(2011)]{bartolucci2011modeling}
Francesco Bartolucci and Leonardo Grilli.
\newblock Modeling partial compliance through copulas in a principal
  stratification framework.
\newblock \emph{Journal of the American Statistical Association}, 106\penalty0
  (494):\penalty0 469--479, 2011.

\bibitem[Chernozhukov et~al.(2018)Chernozhukov, Chetverikov, Demirer, Duflo,
  Hansen, Newey, and Robins]{chernozhukov2018double}
Victor Chernozhukov, Denis Chetverikov, Mert Demirer, Esther Duflo, Christian
  Hansen, Whitney Newey, and James Robins.
\newblock Double/debiased machine learning for treatment and structural
  parameters, 2018.

\bibitem[Chiba(2018)]{chiba2018bayesian}
Yasutaka Chiba.
\newblock Bayesian inference of causal effects for an ordinal outcome in
  randomized trials.
\newblock \emph{Journal of Causal Inference}, 6\penalty0 (2):\penalty0
  20170019, 2018.

\bibitem[Huang et~al.(2017)Huang, Fang, Hanley, and
  Rosenblum]{huang2017inequality}
Emily~J Huang, Ethan~X Fang, Daniel~F Hanley, and Michael Rosenblum.
\newblock Inequality in treatment benefits: Can we determine if a new treatment
  benefits the many or the few?
\newblock \emph{Biostatistics}, 18\penalty0 (2):\penalty0 308--324, 2017.

\bibitem[Kennedy(2024)]{kennedy2024semiparametric}
Edward~H Kennedy.
\newblock Semiparametric doubly robust targeted double machine learning: a
  review.
\newblock \emph{Handbook of statistical methods for precision medicine}, pages
  207--236, 2024.

\bibitem[Li and Tobias(2008)]{li2008bayesian}
Mingliang Li and Justin~L Tobias.
\newblock Bayesian analysis of treatment effects in an ordered potential
  outcomes model.
\newblock In \emph{Modelling and Evaluating Treatment Effects in Econometrics},
  volume~21, pages 57--91. Emerald Group Publishing Limited, 2008.

\bibitem[Lu et~al.(2018)Lu, Ding, and Dasgupta]{lu2018treatment}
Jiannan Lu, Peng Ding, and Tirthankar Dasgupta.
\newblock Treatment effects on ordinal outcomes: Causal estimands and sharp
  bounds.
\newblock \emph{Journal of Educational and Behavioral Statistics}, 43\penalty0
  (5):\penalty0 540--567, 2018.

\bibitem[Lu et~al.(2020)Lu, Zhang, and Ding]{lu2020sharp}
Jiannan Lu, Yunshu Zhang, and Peng Ding.
\newblock Sharp bounds on the relative treatment effect for ordinal outcomes.
\newblock \emph{Biometrics}, 76\penalty0 (2):\penalty0 664--669, 2020.

\bibitem[Lu et~al.(2025)Lu, Jiang, and Ding]{lu2025principal}
Sizhu Lu, Zhichao Jiang, and Peng Ding.
\newblock Principal stratification with continuous post-treatment variables:
  Nonparametric identification and semiparametric estimation.
\newblock \emph{Journal of the Royal Statistical Society Series B: Statistical
  Methodology}, page qkaf049, 2025.

\bibitem[Mao(2018)]{mao2018causal}
Lu~Mao.
\newblock {On causal estimation using $U$-statistics}.
\newblock \emph{Biometrika}, 105\penalty0 (1):\penalty0 215--220, 2018.

\bibitem[Newey(1990)]{newey1990semiparametric}
Whitney~K Newey.
\newblock Semiparametric efficiency bounds.
\newblock \emph{Journal of applied econometrics}, 5\penalty0 (2):\penalty0
  99--135, 1990.

\bibitem[Pocock et~al.(2012)Pocock, Ariti, Collier, and Wang]{pocock2012win}
Stuart~J Pocock, Cono~A Ariti, Timothy~J Collier, and Duolao Wang.
\newblock The win ratio: a new approach to the analysis of composite endpoints
  in clinical trials based on clinical priorities.
\newblock \emph{European heart journal}, 33\penalty0 (2):\penalty0 176--182,
  2012.

\bibitem[Robins et~al.(2008)Robins, Li, Tchetgen, van~der Vaart,
  et~al.]{robins2008higher}
James Robins, Lingling Li, Eric Tchetgen, Aad van~der Vaart, et~al.
\newblock Higher order influence functions and minimax estimation of nonlinear
  functionals.
\newblock In \emph{Probability and statistics: essays in honor of David A.
  Freedman}, volume~2, pages 335--422. Institute of Mathematical Statistics,
  2008.

\bibitem[Rosenbaum(2002)]{rosenbaum2002observational}
Paul~R. Rosenbaum.
\newblock \emph{Observational studies}.
\newblock Springer, 2002.

\bibitem[Sklar(1959)]{sklar1959fonctions}
M~Sklar.
\newblock Fonctions de r{\'e}partition {\`a} n dimensions et leurs marges.
\newblock In \emph{Annales de l'ISUP}, volume~8, pages 229--231, 1959.

\bibitem[Tsiatis(2006)]{tsiatis2006semiparametric}
Anastasios~A Tsiatis.
\newblock \emph{Semiparametric theory and missing data}.
\newblock Springer, 2006.

\bibitem[van~der Vaart and Wellner(1997)]{vaart1997weak}
AW~van~der Vaart and Jon~A Wellner.
\newblock Weak convergence and empirical processes with applications to
  statistics.
\newblock \emph{Journal of the Royal Statistical Society-Series A Statistics in
  Society}, 160\penalty0 (3):\penalty0 596--608, 1997.

\bibitem[Volfovsky et~al.(2015)Volfovsky, Airoldi, and
  Rubin]{volfovsky2015causal}
Alexander Volfovsky, Edoardo~M Airoldi, and Donald~B Rubin.
\newblock Causal inference for ordinal outcomes.
\newblock \emph{arXiv preprint arXiv:1501.01234}, 2015.

\bibitem[Yadlowsky et~al.(2022)Yadlowsky, Namkoong, Basu, Duchi, and
  Tian]{yadlowsky2022bounds}
Steve Yadlowsky, Hongseok Namkoong, Sanjay Basu, John Duchi, and Lu~Tian.
\newblock Bounds on the conditional and average treatment effect with
  unobserved confounding factors.
\newblock \emph{Annals of statistics}, 50\penalty0 (5):\penalty0 2587, 2022.

\bibitem[Zeng et~al.(2020)Zeng, Li, and Ding]{zeng2020being}
Shuxi Zeng, Fan Li, and Peng Ding.
\newblock Is being an only child harmful to psychological health?: evidence
  from an instrumental variable analysis of china's one-child policy.
\newblock \emph{Journal of the Royal Statistical Society Series A: Statistics
  in Society}, 183\penalty0 (4):\penalty0 1615--1635, 2020.

\bibitem[Zhang et~al.(2025)Zhang, Geng, Li, and Ding]{zhang2025identifying}
Chao Zhang, Zhi Geng, Wei Li, and Peng Ding.
\newblock Identifying and bounding the probability of necessity for causes of
  effects with ordinal outcomes.
\newblock \emph{Biometrika}, 112\penalty0 (3):\penalty0 asaf049, 2025.

\bibitem[Zheng et~al.(2025)Zheng, D'Amour, and Franks]{zheng2025copula}
Jiajing Zheng, Alexander D'Amour, and Alexander Franks.
\newblock Copula-based sensitivity analysis for multi-treatment causal
  inference with unobserved confounding.
\newblock \emph{Journal of Machine Learning Research}, 26\penalty0
  (36):\penalty0 1--60, 2025.

\end{thebibliography}


\begin{thebibliography}{2}
\providecommand{\natexlab}[1]{#1}
\providecommand{\url}[1]{\texttt{#1}}
\expandafter\ifx\csname urlstyle\endcsname\relax
  \providecommand{\doi}[1]{doi: #1}\else
  \providecommand{\doi}{doi: \begingroup \urlstyle{rm}\Url}\fi

\bibitem[Chernozhukov et~al.(2018)Chernozhukov, Chetverikov, Demirer, Duflo,
  Hansen, Newey, and Robins]{chernozhukov2018double}
Victor Chernozhukov, Denis Chetverikov, Mert Demirer, Esther Duflo, Christian
  Hansen, Whitney Newey, and James Robins.
\newblock Double/debiased machine learning for treatment and structural
  parameters, 2018.

\bibitem[Yadlowsky et~al.(2022)Yadlowsky, Namkoong, Basu, Duchi, and
  Tian]{yadlowsky2022bounds}
Steve Yadlowsky, Hongseok Namkoong, Sanjay Basu, John Duchi, and Lu~Tian.
\newblock Bounds on the conditional and average treatment effect with
  unobserved confounding factors.
\newblock \emph{Annals of statistics}, 50\penalty0 (5):\penalty0 2587, 2022.

\end{thebibliography}
	\end{document}